\documentclass[acmtog]{acmart}
\acmSubmissionID{220}

\usepackage{booktabs}

\setcopyright{rightsretained}
\acmJournal{TOG}
\acmYear{2021}\acmVolume{40}\acmNumber{4}\acmArticle{156}\acmMonth{8}
\acmDOI{10.1145/3450626.3459770}

\author{Kacper Pluta}
\orcid{0000-0002-3623-1412}
\affiliation{%
	\institution{Technion - Israel Institute of Technology}
	\city{Haifa}
	\country{Israel}}
\email{kacper.pluta@gmail.com}

\author{Michal Edelstein}
\affiliation{%
	\institution{Technion - Israel Institute of Technology}
	\city{Haifa}
	\country{Israel}}
\email{smichale@cs.technion.ac.il}

\author{Amir Vaxman}
\orcid{0000-0001-6998-6689}
\affiliation{%
	\institution{Utrecht University}
	\city{Utrecht}
	\country{The Netherlands}}
\email{A.Vaxman@uu.nl}

\author{Mirela Ben-Chen} 
\orcid{0000-0002-1732-2327}
\affiliation{%
	\institution{Technion - Israel Institute of Technology}
	\city{Haifa}
	\country{Israel}}
\email{mirela@cs.technion.ac.il}
	
\authorsaddresses{
Authors' addresses: Kacper Pluta, Michal Edelstein, Mirela Ben-Chen, The Henry and Marilyn Taub Faculty of Computer Science, Technion - Israel Institute of Technology, Haifa, Israel, 3200003; Amir Vaxman, Princetonplein 5, De Uithof, 3584 CC Utrecht, The Netherlands.
}

\usepackage[ruled]{algorithm2e} 

\SetAlFnt{\small}
\SetAlCapFnt{\small}
\SetAlCapNameFnt{\small}
\SetAlCapHSkip{0pt}

\acmJournal{TOG}

\usepackage[utf8]{inputenc}
\usepackage{amsmath}
\usepackage{mathrsfs}
\usepackage{subfig}
\usepackage{bbm, dsfont}
\usepackage{graphicx}
\graphicspath{{figures/}} 
\usepackage{wrapfig}
\usepackage{mathtools}
\usepackage{optidef}
\usepackage{physics}

\usepackage{thmtools,thm-restate}

\usepackage{multicol}
\newcommand\blfootnote[1]{%
	\begingroup
	\renewcommand\thefootnote{}\footnote{#1}%
	\addtocounter{footnote}{-1}%
	\endgroup
}

\usepackage{xcolor}
\definecolor{applegreen}{rgb}{0.55, 0.71, 0.0}
\definecolor{ao(english)}{rgb}{0.0, 0.5, 0.0}

\newcommand{\REV}[1]{{ #1}}

\newcommand{\tin}{\!\in\!}
\newcommand{\invf}[1]{\frac{1}{#1}}

\newcommand{\xR}{\mathbb{R}}
\newcommand{\xC}{\mathbb{C}}
\newcommand{\xZ}{\mathbb{Z}}
\newcommand{\xI}{\mathbb{I}}
\newcommand{\sM}{\mathcal{M}}

\newcommand{\sU}{\mathcal{U}}
\newcommand{\sH}{\mathcal{H}}
\newcommand{\s}[1]{\mathtt{#1}}

\newcommand{\mM}{M}
\newcommand{\mV}{\mathcal{V}}
\newcommand{\mE}{\mathcal{E}}
\newcommand{\mEi}{\mathcal{E}^I}
\newcommand{\mF}{\mathcal{F}}
\newcommand{\nov}{n}

\newcommand{\noei}{|\mEi|}
\newcommand{\nof}{m}

\newcommand{\tpm}{T_p\mathcal{M}} 

\newcommand{\h}{t} 

\newcommand{\kmin}{\kappa_1}
\newcommand{\kmax}{\kappa_2}
\newcommand{\kn}{\kappa_n}
\newcommand{\kgauss}{K}
\newcommand{\kmean}{H}
\newcommand{\shapeop}{S}

\newcommand{\dmin}{d_1} 
\newcommand{\dmax}{d_2} 
\newcommand{\damin}{\bar{d}} 

\newcommand{\dmindmax}{\begin{psmallmatrix}\dmin & \dmax\end{psmallmatrix}}
\newcommand{\skminskmax}{\begin{psmallmatrix}\skmin & 0 \\ 0 & \skmax\end{psmallmatrix}}

\newcommand{\maxel}{\eta}     

\newcommand{\skmin}{\tilde{\kappa}_{1,\epsilon}} 
\newcommand{\skmax}{\tilde{\kappa}_{2,\epsilon}} 
\newcommand{\skn}{\tilde{\kappa}_{n,\epsilon}} 
\newcommand{\sshapeop}{\tilde{\shapeop}_\epsilon} 

\newcommand{\ar}{\alpha} 

\newcommand{\shp}{{\sH_\s{p}}}

\newcommand{\pf}{z}  
\newcommand{\uv}[1]{\begin{psmallmatrix}U_{#1} & V_{#1}\end{psmallmatrix}}
\newcommand{\mvu}[1]{\begin{psmallmatrix}-V_{#1} & U_{#1}\end{psmallmatrix}}

\newcommand{\uvinv}[1]{{\uv{#1}}^{-1}}

\newcommand{\id}{\begin{psmallmatrix}1 & 0 \\ 0 & 1\end{psmallmatrix}}

\newcommand{\aln}[1]{{n_{#1}}}

\newcommand{\mFh}{{\mF}_h}      
\newcommand{\mFp}{{\mF}_p}  
\newcommand{\mFe}{{\mF}_e}      
\newcommand{\mFnp}{{\mF}_{npl}}  
\newcommand{\mFu}{{\mF}_{u}}  

\newcommand{\mFs}{{\mF}_2}      
\newcommand{\mFq}{{\mF}_4}      

\newcommand{\rhom}{\rho_m}
\newcommand{\deltae}{\varphi_e}
\newcommand{\deltap}{\varphi_p}

\newcommand{\gf}{y} 

\newcommand{\rM}{\mM_R}  
\newcommand{\rdM}{\mM_D}  

\newcommand{\hM}{\mM_H} 
\newcommand{\hV}{\mV_H} 
\newcommand{\hF}{\mF_H} 
\newcommand{\hE}{\mE_H} 

\newcommand{\qV}{\mV_Q} 
\newcommand{\qF}{\mF_Q}

\newcommand{\planerr}{\varepsilon_p}  
\newcommand{\planerrM}{\epsilon_p}  
\newcommand{\planerravg}{\bar{\epsilon}_p}  
\newcommand{\hdist}{\epsilon_d} 

\begin{document}
\title{PH-CPF: Planar Hexagonal Meshing using Coordinate Power Fields}

\begin{abstract}
We present a new approach for computing planar hexagonal meshes that approximate a given surface, represented as a triangle mesh. Our method is based on two novel technical contributions. First, we introduce \emph{Coordinate Power Fields}, which are a pair of tangent vector fields on the surface that fulfill a certain \emph{continuity} constraint. We prove that the fulfillment of this constraint guarantees the existence of a seamless parameterization with quantized rotational jumps, which we then use to regularly remesh the surface. We additionally propose an optimization framework for finding Coordinate Power Fields, which also fulfill additional constraints, such as alignment, sizing and bijectivity. Second, we build upon this framework to address a challenging meshing problem: planar hexagonal meshing. To this end, we suggest a combination of conjugacy, scaling and alignment constraints, which together lead to planarizable hexagons. We demonstrate our approach on a variety of surfaces, automatically generating planar hexagonal meshes on complicated meshes, which were not achievable with existing methods.

\end{abstract}

\ccsdesc[500]{Computing methodologies~Mesh geometry models}

\keywords{geoemtry processing, planar hexagonal meshing, parameterization, tangent vector fields}

\begin{teaserfigure}
\centering
\includegraphics[height=2.5in]{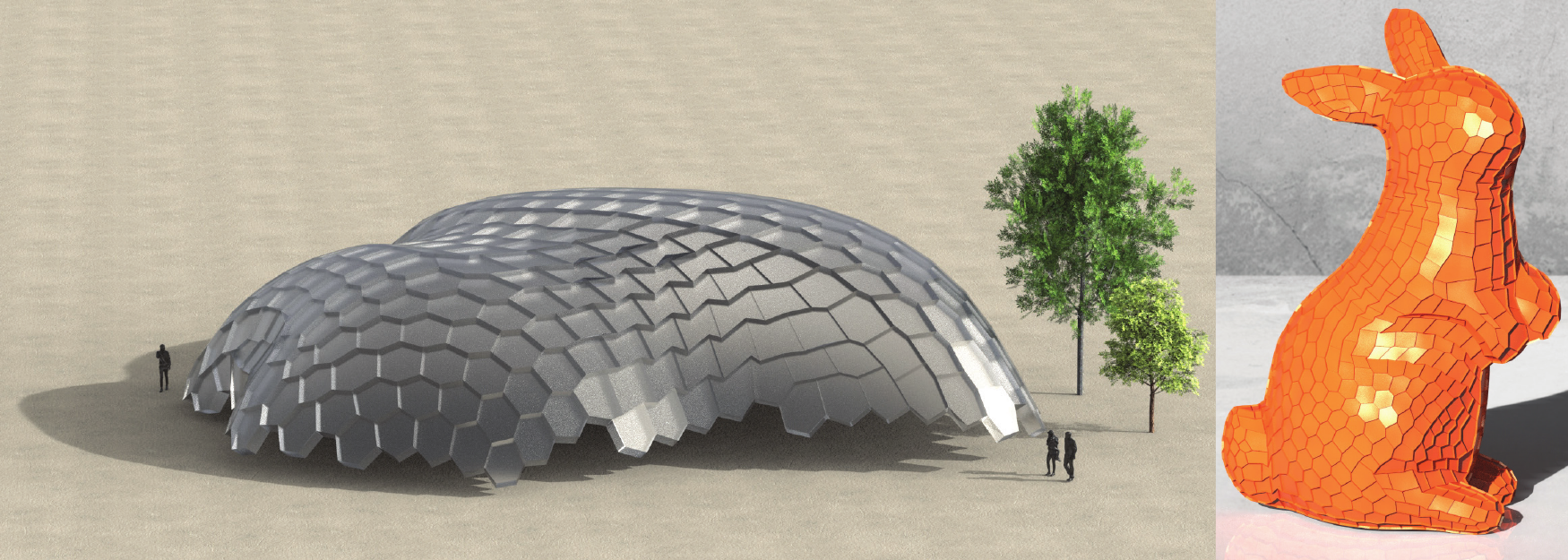}
\caption{Planar hexagonal meshes of computed with our approach.}
\label{fig:teaser}
\end{teaserfigure}

\maketitle

\section{Introduction}
Trianglular meshes are the most prevalent surface representation in geometry processing. However, for many applications, such as architectural design, animation and numerical simulation, a more structured representation is required, such as a quadrangular or hexagonal mesh. Specifically in architectural design, meshes with planar faces are amenable to realization using wood or glass panels, and are thus highly sought after~\cite{Li2015segmental,Sonntag2019lightweight}. While planar quadrangular meshes have been intensively researched in recent years~\cite{Liu2006geometric,pottmann2015architectural}, 
planar hexagonal meshes remain an open challenge. Due to considerable technical difficulties, such as a strict restriction on allowable element orientations, and the necessary existence of non-convex faces, existing approaches are either limited to specific types of surfaces, or are not fully automatic, or both~\cite{pottmann2015architectural}.

From the perspective of the geometry processing framework available for meshing, much progress has been made in recent years~\cite{bommes2009mixed,bommes2013integer,campen2015quantized}. 
Still, the ``standard'' approach relies on computing vector fields on the surface, which are the \emph{candidate gradients} of the parameterization functions~\cite{bommes2013quad}. However, with the exception of a few recent works, there is no guarantee that these candidates are in fact \emph{integrable}, i.e. that there exist parameterization functions whose gradients are these vector fields. Furthermore, 
while using the candidate gradients as the primary variables is convenient in some applications, in many cases, such as the planar hexagonal meshing application, and especially when the gradients are not expected to be orthogonal, this is not the most natural choice. 

We propose a novel optimization framework for parameterization-based meshing, whose primary variables are a pair of tangent vector fields, which are candidate \emph{Coordinate Power Fields} (CPF). CPFs are a generalization of \emph{coordinate vector fields}, which are the pushforward of the $2D$ coordinate grid under a given local patch parameterization. Unlike coordinate vector fields, coordinate power fields allow for rotational jumps in the parameterization, and thus (up to global holonomy constraints) lead to a seamless global paramterization. Specifically, at the heart of our framework lies a new \emph{continuity} constraint, which two tangent vector fields should fulfill in order to be CPFs, and which provably guarantees the local existence of a seamless parameterization with quantized rotational jumps. 

We leverage our framework to compute planar hexagonal meshes. Specifically, we define additional constraints on the CPFs, using a novel \emph{Dupin metric} approach, such that a pushed forward hexagonal grid would be planarizable. We demonstrate that this holds if the CPFs are \emph{conjugate}, and are additionally aligned and sized according to the \emph{Dupin Indicatrix}~\cite{do2016differential}. By combining these constraints with our optimization framework we automatically generate planar hexagonal meshes that approximate a wide variety of shapes, including complex shapes which, to the best of our knowledge, are not achievable by existing automatic methods. Furthermore, on meshes where existing approaches are applicable, we achieve a comparable or better planarity error as existing approaches, while using a considerably smaller number of elements. Finally, we demonstrate that our framework is additionally applicable to planar quad meshing.

\subsection{Related Work}

\paragraph{Global seamless quantized parameterization.}
Global parameterization, meshing and directional field design are vast topics in geometry processing, with ample amount of existing work. For a general overview of these topics we refer the reader to the excellent surveys on parameterization~\cite{sheffer2006mesh,hormann2007mesh}, quad meshing~\cite{bommes2013quad} and directional field design~\cite{vaxman2016directional}. 

Our meshing is parameterization-guided, namely, we compute an \emph{integer grid-map}~\cite{bommes2013integer,campen2015quantized} and push forward a regular grid through this map to the surface to generate the new mesh elements. Within this framework, our main contributions is the computation of the gradients of a \emph{seamless} global parameterization with quantized rotational continuity across the cuts. The remaining parts of the algorithm, namely integrating the parameterization, quantizing the translations, generating the mesh~\cite{vaxman2019directional}, and planarization~\cite{tang2014form} are based on existing work.

One challenge in this framework is achieving smooth gradient vector fields while allowing for discontinuities which lead to quantized singularities (or \emph{cone points}). The most prevalent approach to designing such fields is working with rotationally symmetric vector fields  (RoSy-fields), which are computed by solving a mixed-integer~\cite{bommes2009mixed} or integer-valued~\cite{farchi2018integer} optimization problem, or with a reformulation to a complex-valued problem~\cite{knoppel2013globally,viertel2019approach,azencot2017consistent} (which led to the notion of \emph{power fields}). More versatility is achieved with complex polynomials (PolyVectors)~\cite{diamanti2014designing}, as they relax the requirement of rotational symmetry. Additionally, generalization beyond cross fields has been suggested by using a non-Euclidean metric, either extrinsically~\cite{Panozzo2014Frame} or intrinsically~\cite{Jian2015Frame}.

These rotational fields are then integrated into an \emph{integer grid-map}~\cite{bommes2013integer}, where care is taken to guarantee that some translational components are also quantizable~\cite{lyon2019parametrization}. However, with the exception of a few methods, most approaches that compute direction fields do not guarantee their \emph{integrability}. Thus, the resulting parameterization gradients often deviate from the computed direction fields. Notable exceptions are PGP~\cite{ray2006periodic}, that attempt to reduce the curl from the computed field before integrating, and the work of Myles et al.~\cite{Myles2014robust}, that explicitly optimizes for a parameterization that agrees with an input cross field, yet it may introduce additional cone points. Furthermore, both approaches use a cross-field, which cannot, by definition, define the scale of a non-uniform parameterization, or the orientation of a non-conformal one (since cross-fields are unit-length and orthogonal). Harmonic cone-based seamless parameterization methods (e.g., HGP~\cite{Bright2017harmonic}, and the orbifold methods~\cite{Aigerman2016Hyperbolic}) guarantee local injectivity for given cone values, yet do not consider the field, hence they also cannot generate non-uniform parameterizations. The simplex assembly method~\cite{Fu2016computing} uses affine transformations as the main variables and guarantees injectivity, yet it does not allow for quantized rotational discontinuities, nor does it consider the alignment. The only method that allows for non-orthogonal fields, considers alignment, and also optimizes integrability is the work by~\citet{Diamanti2015Integrable} on integrable PolyVector fields. However, PolyVector fields allow for \emph{any} set of consistently oriented vector fields, and thus form a much bigger space than the space of coordinate vector fields of parameterizations. 

Our method, on the other hand, uses the vectors tangent to the meshing elements as the primary variables, thus considerably simplifying the formulation of constraints. In addition, we allow for non-orthogonal gradients, by posing the rotationally-symmetric constraint in the \emph{parametric domain}. Thus, our variables are exactly as general as required to allow for any seamless  parameterization with quantized rotations, without spurious degrees of freedom. The methods closest to our approach, are the recent algorithms for computing Chebyshev nets~\cite{sageman2019chebyshev,liu2020practical}. The second also used the tangent vector fields as primary variables and introduced an integrability constraint that is formulated using the \emph{Lie Bracket operator}. Notably, this operator is \emph{discretized}, and not \emph{discrete}, hence it will not be exactly zero even if a seamless parameterization exists. We, on the other hand, present a novel \emph{discrete} integrability constraint, which is provably exactly zero if and only if a seamless parameterization exists locally (without considering global homology obstructions). Moreover, our formulation is a generalization of the classic integrability of curl-free piecewise constant vector fields~\cite{Polthier2003identifying}, which is obtained for the special case $N=1$. The second method used both the tangent vector fields and the candidate gradients as variables, defining the smoothness on the first representation with PolyVectors, and the integrability on the second, and used alternating descent for the optimization. Our approach simplifies this formulation by removing the need for PolyVectors, while keeping the discrete integrability guarantee. Finally, both approaches were tailored for Chebyshev quad meshes, and a generalization to planar hexagonal meshes (or any other general formulation) was not presented. 

We further elaborate on the relations between our formulation and the most similar approaches in section~\ref{sec:relation}.

\paragraph{Planar hexagonal meshing.}
Planar hexagonal (PH) and planar quad (PQ) meshes are of paramount importance in architecture~\cite{pottmann2007architectural,pottmann2015architectural}. PH meshes in particular are highly sought after for grid-shell constructions with wood panels~\cite{Li2015segmental,Sonntag2019lightweight,wagner2020towards}. This is due to their valence-$3$ vertices which lead to (conical) meshes with a natural face offset that can be realized with panels of constant thickness. They are further known to have beneficial mechanical properties~\cite{la2012nature,li2017timber}.
A few attempts have been made in the architecture literature to generate such meshes given an input smooth surface~\cite{li2017timber,Contestabile2019digital,Mesnil2017Generalised,henriksson2015rationalizing}, yet the resulting algorithms often require manual work and are mostly limited to simple surfaces. It is worth noting that finding the PH mesh, or the flat segments that will form the panels, is only the first step in the realization of the surface. Further structural analysis is required to validate mechanical stability (see e.g. the work done on the Landesgartenschau exhibition hall~\cite{Li2015segmental}), and thus a fully automatic approach that can generate a variety of meshes based on a small number of parameters (such as element size) is critical. 

In the graphics and geometry processing literature there has also been a limited amount of work on this topic, unlike PQ meshes which were extensively researched~\cite{Liu2006geometric,Zadravec2010designing,liu2011general}. PH meshes are in general more constrained than PQ meshes, as they, for example, require non-convex faces to create a watertight cover of a surface with negative Gaussian curvature~\cite{pottmann2015architectural}. While purely convex hexagonal faces have been used for architectural structures~\cite{deleon2012rationalisation,romero2013bridging}, the resulting facade is not watertight. One approach for hexagonal meshing is using the duality with triangular meshes, either by triangulating the Gauss map~\cite{Almegaard2007surfaces}, or by tangent plane intersections~\cite{troche2008planar,troche2009planar,zimmer2012variational}. These approaches often do not lead to regular planar hexagonal shapes. 

Better hexagonal shapes can be achieved if the primal triangular mesh is aligned with the principal directions, and sized according to the \emph{Dupin indicatrix}~\cite{wang2008hexagonal,li2014planar}. To achieve such an ``ideal triangulation'' the authors first construct an anisotropic quadrangular mesh, which is then decomposed into a triangular mesh, whose dual is the sought after planar hexagonal mesh. These multiple transitions require special care for singularities regions. Furthermore, to obtain aesthetic hexagon shapes in transition regions between positive and negative Gaussian curvature, these transition regions must be specified manually. A related approach~\cite{wang2010note} also starts from a quad mesh obtained from a conjugate curve network, and then shifts the quads to obtain a ``brick-wall'' pattern which is later optimized for planarization, which will again require special care for singularities. Our method, on the other hand, directly computes a triangulation that is aligned and sized according to the Dupin indicatrix, whose barycetric dual is a hexagonal mesh which is planarizable. Our approach does not require any special treatment of singularities, or manual adjustments in curvature transition regions, and generates planar hexagonal meshes on surfaces more complex than previously possible.

Other approaches aim at generating different variations of hexagonal meshes. For example, Nieser et al.~\shortcite{nieser2011hexagonal} generate curvature-aligned triangular meshes (without considering the planarity of the dual meshes), and we follow a similar approach for choosing alignment directions. Other variants include conformal planar hexagonal meshes~\cite{muller2011conformal}, which are only applicable to a restricted set of surfaces, honeycomb structures~\cite{jiang2014freeform} and caravel meshes~\cite{Tellier2020caravel,tellier2020caravelb}, and polyhedral patterns~\cite{Jiang2015polyhedral}. The latter work focuses on the post-processing of a polygonal mesh (PH meshes as a special case), to generate an aesthetic planar polygonal pattern. We note that our approach generates the initial polygonal mesh which this approach receives as input. We use a similar method for post-processing, yet with some modifications as we do not have the regular stripe structure used there.

\subsection{Contributions:}
Our main contributions are:
\begin{itemize}
\item Coordinate Power Fields (CPFs) - a novel discrete formulation of integrability for seamless global parameterizations with quantized rotations.
\item A corresponding general optimization framework that is applicable to different meshing applications.
\item A new parameterization-based method for generating planar hexagonal meshes, using CPFs and a Dupin-indicatrix based metric, that exceeds the state of the art in both the quality of the resulting meshes, and the complexity of the surfaces that can be approximated.
\end{itemize}

\section{Background}
Some of this material, e.g. on patch parameterization~\cite{sageman2019chebyshev} and on the Dupin indicatrix~\cite{wang2009note}, has appeared previously. We repeat it here for completeness and for a quick reference using our notations, and limit ourselves to the concepts required for our exposition. For a full treatment of patch parameterization see the book by do Carmo~\shortcite[Chap. 2]{do2016differential}.

\subsection{Meshing by pushforward through a patch parameterization}
\subsubsection{Patch parameterization}
Given an oriented embedded surface $\sM$, let $r: \sU \subset \xR^2 \to \Omega \subset \sM$ be a regular parameterization of a planar domain to a patch on the surface (see Figure~\ref{fig:notations})~\footnote{Note that we use the classical plane-to-surface definition, rather than the common computer graphics one which is surface-to-plane.}. Let $\hat{u},\hat{v}$ be  the Cartesian unit orthogonal axes in the parameterization domain, and $\s{u},\s{v}$ be the coordinate functions on $\sU$. 
Thus, $r$ is defined on points $(\s{u},\s{v}) \tin \sU$ and $$p = r(\s{u},\s{v}) \in \sM, \quad r^{-1}(p) = (\s{u},\s{v}) \in \sU.$$ 
We define two scalar functions $u, v: \Omega \to \xR$ as the pullback of the functions $\s{u},\s{v} : \sU \to \xR$ through $r^{-1}$:
$$u(p) = \s{u}(r^{-1}(p)), \quad v(p) = \s{v}(r^{-1}(p)).$$
The gradients $\nabla u, \nabla v: \Omega \to T\sM$ are tangent vector fields on $\Omega$, whereas the gradients $\nabla \s{u}, \nabla \s{v} : \sU \to \xR^2$ are planar vector fields. Specifically, $\nabla \s{u} = \hat{u}, \nabla \s{v} = \hat{v}$.

The partial derivatives $U = \frac{\partial r}{\partial \s{u}}, V = \frac{\partial r}{\partial \s{v}}$ are the corresponding \emph{coordinate vector fields}~\cite[pp. 176]{Lee2013introduction} on $\Omega$. Equivalently, we write $U = dr(\hat{u}), V = dr(\hat{v})$, where $dr$ is the \emph{differential} of $r$. The vector fields $U,V$ are tangent to the isolines $v=v_0, u = u_0$, respectively.
\begin{figure}[t]
\centering
\includegraphics[width=\linewidth]{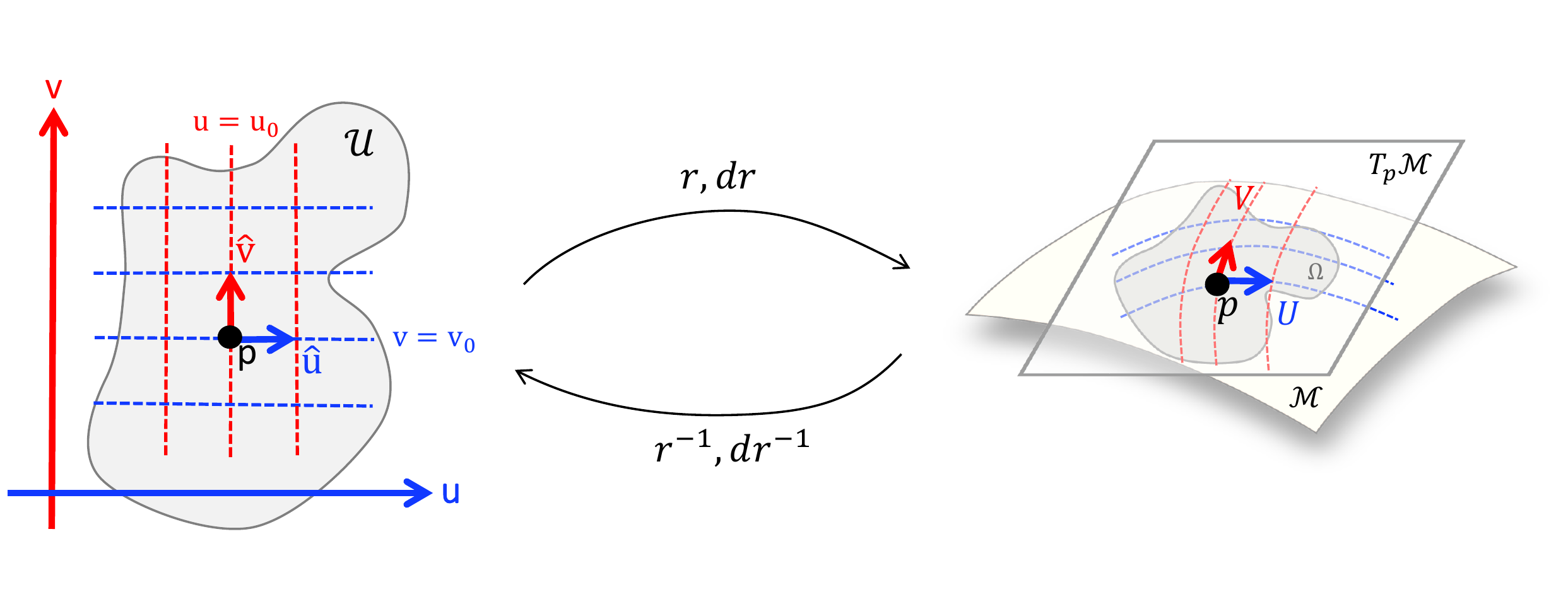}
\caption{Notations, patch parameterization and coordinate vector fields.}
\label{fig:notations}
\end{figure}

\subsubsection{Meshing by grid-pushforward}
A common approach for regular meshing~\cite{bommes2013integer}, is pushing forward a regular grid through a parameterization. We first consider a single patch and quadrangular remeshing. Then, the locations of the new mesh vertices are given by $r(\h\xZ^2 \cap \sU)$, where $\h\tin\xR$ is a constant scalar coefficient that controls the size of the new mesh elements. Let $\s{p} \!=\! (\s{u}_0, \s{v}_0) \tin \h\xZ^2 \cap \sU$, and consider the quad $\{r(\s{p}), r(\s{p} + \h\hat{u}), r(\s{p} + \h\hat{v} + \h\hat{u}),  r(\s{p} + \h\hat{v})\}$. In the limit of increasing refinement, we have, by the definition of the differential, $\lim_{\h\to0} \frac{r(\s{p}+\h{\hat{u}})-r(\s{p})}{h} = dr_\s{p}(\hat{u}) = U(r(\s{p}))$, and similarly for $V$. Hence, the vectors $U(p),V(p), p = r(\s{p})$ are closely related to the edges of the generated quad~\footnote{Note that we abuse notations slightly and treat $r(\s{p})\tin\sM$ as its corresponding embedded point in $\xR^3$ here.}.

Consequentially, many constraints on the mesh elements, whether planarity, sizing, orientation or angles, are easily specified using constraints on the coordinate vector fields $U,V$.
Furthermore, the pullback map $dr^{-1}$, as well as the gradient vector fields $\nabla u, \nabla v$ all have simple expressions in terms of $U,V$.
Specifically, let $B_p$ be a local orthonormal basis of $\tpm$ at $p \tin \Omega$, and $\uv{p}$ be the $2\times2$ matrix whose columns are the coefficients of $U(p),V(p)$ in the basis $B_p$. Then, it is straightforward to show (see \REV{supplemental material}) that: 
\begin{equation}
\label{eq:pullback}
dr^{-1}(X) = \uvinv{p} X, \quad \forall X\tin \tpm, 
\end{equation}
and
\begin{equation}
\label{eq:gradients}
(\nabla u)_p = \spmqty{1 & 0} \uvinv{p}, \quad (\nabla v)_p = \spmqty{0 & 1} \uvinv{p},
\end{equation}
where all the coefficients of vectors in $\tpm$ are with respect to the local basis $B_p$. A straightforward calculation then gives:
\begin{equation}
\label{eq:gradients_noinv}
(\nabla u)_p = -\frac{1}{s_p}(J V_p)^T , \quad (\nabla v)_p = \frac{1}{s_p} (J U_p)^T,
\end{equation}
where $s_p = \langle U_p, -J V_p \rangle$, and $J = \spmqty{0 & -1 \\ 1 & 0}$.
Finally, by combining the previous three equations, we have an explicit expression of $dr^{-1}$:
\begin{equation}
\label{eq:pullback_from_uv}
dr^{-1}_p = \frac{1}{s_p} \mvu{p}^T J^T.
\end{equation}
We note that this discussion appears in~\cite{sageman2019chebyshev}, albeit using differential one-forms instead of gradient vector fields. 
\begin{figure}[t!]
\centering
\includegraphics[width=\linewidth]{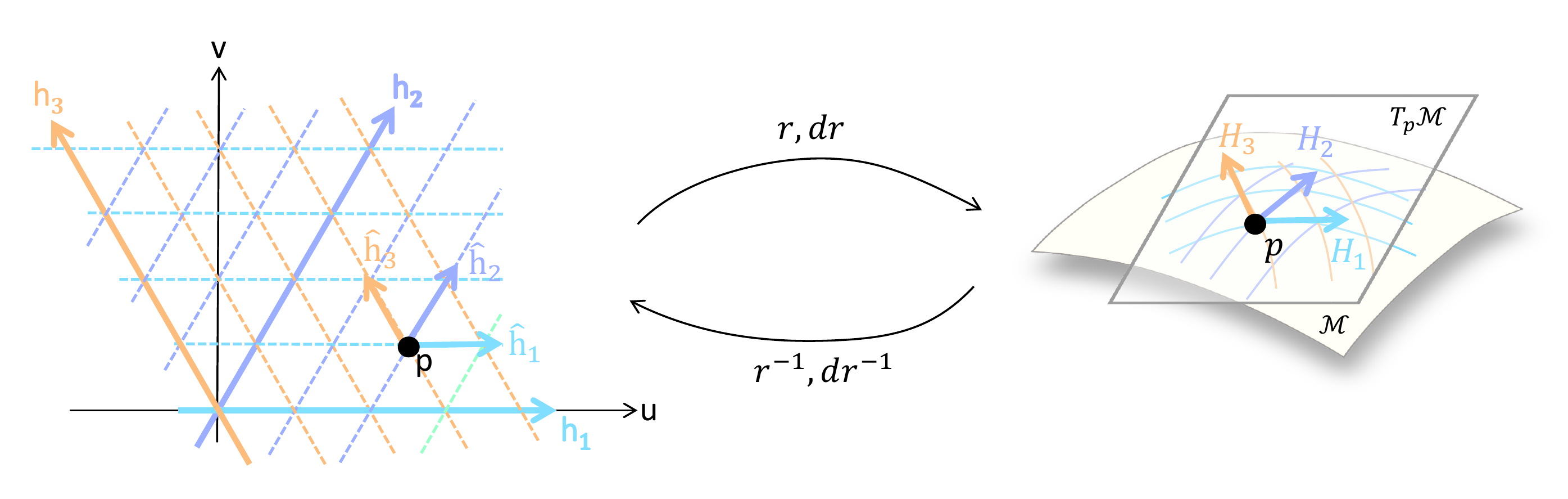}
\caption{Notations, triangular meshing.}
\label{fig:tri_mesh_notations}
\end{figure}

\subsubsection{Triangular meshing} 
\label{sec:triangular_meshing}
Our PH-meshing algorithm is based on tangent duality with a triangular mesh, which we compute using our parameterization framework. Here, the grid $\xZ^2$ is replaced by the triangular grid $\mathbb{A}^2$ (see Figure~\ref{fig:tri_mesh_notations}), where $\hat{h}_1,\hat{h}_2,\hat{h}_3$, are the triangular unit regular coordinate axes in the parameterization domain $\sU$. 

There is a linear relationship between the Cartesian and triangular grids, thus we can alternate between the two representations using matrix multiplication. Specifically, we have $ \begin{psmallmatrix}\hat{h}_1 &  \hat{h}_2 & \hat{h}_3\end{psmallmatrix}_{2 \times 3} =  \begin{psmallmatrix}\hat{u} & \hat{v}\end{psmallmatrix}_{2 \times 2}\, A$, where $A = \begin{psmallmatrix}1 &  1/2 &  -1/2 \\ 0 & \sqrt 3/2  & \sqrt 3/2\end{psmallmatrix}_{2\times3}$. We additionally define, $\s{h}_1, \s{h}_2, \s{h}_3 : \sU \to \xR$ as the triangular coordinate functions, and we have that $\begin{psmallmatrix}\s{h}_1 &  \s{h}_2 & \s{h}_3\end{psmallmatrix} =  \begin{psmallmatrix}\s{u} & \s{v}\end{psmallmatrix}\, A$. Finally, $H_i$ are given by the pushforward of $\hat{h_i}$ through $r$, so $H_i = dr(\hat{h}_i)$. Since $U,V$ are the coordinate vector fields of $r$, and since $dr$ is linear, we have that $\begin{psmallmatrix}H_1 &  H_2 & H_3\end{psmallmatrix}_{2 \times 3} =  \begin{psmallmatrix}U & V\end{psmallmatrix}_{2 \times 2}\, A$, where $H_i, U, V$ and $A$ are represented in an arbitrary orthonormal local basis at each $p \in \Omega$.

For a point $\s{p} \tin t\mathbb{A}^2$, the pushforward of one of its neighboring triangles has vertices at $\{ r(\s{p}), r(\s{p} + \h\hat{h}_1), r(\s{p} + \h\hat{h}_2)\}$, and similarly for the other $5$ neighboring triangles. As in the quadrangular case, in the refinement limit we have $\lim_{\h\to0} \frac{r(\s{p} + \h\hat{h}_1) - r(\s{p})}{\h} = dr(\hat{h}_1) = H_1$, and similarly for $H_2$ (and the other $5$ triangles, by symmetry). Hence, the edges of the pushed forward triangles are closely related to the vector fields $H_1,H_2,H_3$. The dual hexagonal edges are also closely related to $H_i$ and $U,V$, as we discuss in the next sections.

\subsection{Discrete surfaces}
We represent the input surfaces using triangle meshes $\mM = (\mV, \mE, \mF)$, where $\nof = |\mF|, \nov = |\mV|$. Vector fields are piecewise constant per face, either in the trivial basis in $\xR^3$ as $V \tin \xR^{3\nof}$, or using a local basis per face as $V \tin \xR^{2\nof}$. The orthonormal (arbitrary) local basis per face, which we use to represent vector fields, the shape operator and the pullback maps, is given by a matrix $B \tin \xR^{3\nof\times2\nof}$.

\section{Coordinate Power Fields}
\label{sec:cpf}

\subsection{Rotational Continuity across patches}
Consider two parameterizations $r_i \!:\! \sU_i \to \Omega_i \!\subset\! \sM, \, i \tin [1,2]$. In general,  if $r_i$ are part of a valid atlas of $\sM$, then for any $p \tin \Omega_1 \cap \Omega_2$ and any $X \tin \tpm$, the pullback maps agree. Namely: $r_1^{-1}(p) = r_2^{-1}(p) \tin \xR^2$ and $dr_1^{-1}(X) = dr_2^{-1}(X) \tin \xR^2$. 

In the context of meshing, however, we are only interested in the \emph{pushforward of a grid} through these parameterizations. Hence, the pullback maps are part of a valid \emph{grid atlas}, if they agree up to a \emph{grid automorphism} (see e.g.~\cite{nieser2011hexagonal}).  These consist of quantized rotations of integer multiples of  $2\pi/N$ (where $N=4,6$ for quads and triangles respectively), and quantized translations. In the vector field optimization step we consider only quantized rotations, as these are achievable without introducing integer variables. 

\begin{definition}
Two parameterizations $r_1,r_2$ are part of a valid \emph{rotational grid atlas of degree $N$} iff their corresponding pullback maps $dr_1^{-1}, dr_2^{-1}$ agree up to a rotation by an integer multiple of $2\pi/N$. Specifically, for any $p \tin \sM, \, X \tin \tpm$ there exists an integer $\s{k} \tin [0,..,N-1]$, such that:
\begin{equation}
\label{eq:seamless}
dr_1^{-1}(X) = R^{\frac{2\pi}{N}\s{k}} dr_2^{-1}(X),
\end{equation}
where $R^\theta$ is the counter-clockwise rotation by the angle $\theta$.
\end{definition}

By identifying points $(x,y) \tin \xR^2$ with points $x + iy \tin \xC$, we have (following~\citet{knoppel2013globally}):
\begin{equation}
\label{eq:seamless_c}
\bigl(dr_1^{-1}(X)\bigr)^N = \bigl( dr_2^{-1}(X) \bigr)^N.
\end{equation}
Using Equations~\eqref{eq:pullback} and ~\eqref{eq:gradients}, we have that:
\begin{equation}
\label{eq:seamless_c_grads}
\bigl((D_X u_1)_p, (D_X v_1)_p)\bigr)^N = \bigl((D_X u_2)_p, (D_X v_2)_p)\bigr)^N, 
\end{equation}
\begin{equation}
u_i(p) = \s{u}(r_i^{-1}(p)) , \quad v_i(p) = \s{v}(r_i^{-1}(p)), \quad i = [1,2],
\end{equation}
where $(D_X f)_p = \langle (\nabla f)_p, X \rangle_p$, and $\langle,\rangle_p$ is given by the metric at $p$.

\subsection{Smooth CPFs}

The partial derivatives of a regular parameterization $r : \sU \to \Omega$ are its coordinate vector fields. Analogously, given a collection of parameterizations, we define \emph{Coordinate Power Fields}.

\begin{definition}
Let $r_i \!:\! \sU_i \to \Omega_i \!\subset\! \sM, \, i \tin [1,..,m]$ be regular parameterizations such that $\bigcup_{i=1}^m \Omega_i = \sM$, and let $U_i, V_i$ be their partial derivatives. The set of pairs of vector fields $\{(U_i, V_i) \,| \, i \tin [1,..,m]\}$ is a set of \emph{Coordinate Power Fields of degree $N$} iff every pair of parameterizations is part of a rotational grid atlas of degree $N$.
\end{definition}

\subsection{Discrete Rotational Continuity and CPFs}

Given an oriented manifold triangle mesh $\mM$ embedded in $\xR^3$, let $U,V  \tin \xR^{2\nof}$ be a pair of piecewise constant (PC) vector fields represented in a local basis.

\begin{definition}
Two PC vector fields $U,V$ are \emph{Linearly Independent and Consistently Oriented} (LICO) iff  $\langle U_i, -JV_i \rangle > 0, \, \forall t_i \tin \mF$.
\end{definition}

Following previous work~\cite{bommes2013integer,kalberer2007quadcover}, we set each triangle to be the image of a different parameterization from some triangle domain in $\xR^2$. Given two LICO vector fields, the parameterization to the face $t_i \tin \mF$ is induced by the vectors $U_i, V_i$. 
Specifically, let $(p^1_i,p^2_i,p^3_i) \tin \xR^{3\times3}$ be the embedding of the triangle $t_i \tin \mF$, and consider the triangle 
\begin{equation}
\s{t}_i = \big((0;0), \uvinv{i}(p^2_i-p^1_i), \uvinv{i}(p^3_i-p^1_i)\big) \tin \xR^{2\times3},
\end{equation}
where $p^2_i - p^1_i$ and $p^3_i - p^1_i$ are expressed in the local basis of $t_i$. 
We define $r_i : \s{t}_i \to t_i$ as the unique linear map between the two triangles. Clearly, $dr_i = \uv{i}$, and $dr_i^{-1} = \uvinv{i}$. 
In the discrete case, the  pairwise intersections of the images of the parameterizations are the \emph{mesh edges}, leading to the following definition.

\begin{definition}
\label{def:cpf_discrete}
A pair of LICO vector fields $U,V \tin \xR^{2\nof}$ are \emph{discrete coordinate power fields of degree $N$} iff 
\begin{equation}
\label{eq:cpf}
\bigl(dr_i^{-1}(e_{ij})\bigr)^N = \bigl( dr_j^{-1}(e_{ij}) \bigr)^N, \quad \forall e_{ij} \tin \mEi,
\end{equation}
 where $\mEi$ are the interior edges of $\mM$, and $e_{ij} = t_i \cap t_j, \, t_i, t_j \tin \mF$, with a fixed orientation for both triangles.
\end{definition}

Discrete CPFs guarantee the existence of a corresponding seamless, rotationally quantized, global parameterization, as follows.

\begin{theorem}
\label{thm:cpf}
Let $U,V$ be discrete CPFs of degree $N$. Then there exist functions $u^{l}, v^{l} \tin \xR^{\nof}, \, l = [1,2,3],$ with $u_i, v_i$ piecewise linear per face $t_i \tin \mF$ (yet discontinuous between faces), such that:
\begin{enumerate}
\item $\nabla (u_i)  = \spmqty{1 & 0} \uvinv{i}, \, \nabla (v_i) = \spmqty{0 & 1} \uvinv{i}, \, \forall t_i \tin \mF$.
\item The triangle $\s{t}_i = (\s{p}_i^1, \s{p}_i^2, \s{p}_i^3) \tin \xR^{2\times3}$ with coordinates $\s{p}_i^l = (u^l_i, v^l_i)$ is positively oriented.
\item Let $e_{ij} \tin \mEi$, with $e_{ij} = t_i \cap t_j$, and set $\alpha_i, \beta_i \tin [1,..,3]$ the indices of the vertices of $e_{ij}$ in $t_i$, and similarly for $t_j$.
Thus, $e_{ij} =  {p}_i^{\alpha_i} - {p}_i^{\beta_i} = {p}_j^{\alpha_j} - {p}_j^{\beta_j} \tin \xR^3$. 
Then there exists $\s{k}_{ij} \tin \xZ$ such that $R^{2\pi\s{k}_{ij}/N} (\s{p}_i^{\alpha_i} - \s{p}_i^{\beta_i}) = \s{p}_j^{\alpha_j} - \s{p}_j^{\beta_j}$.
\end{enumerate}
\end{theorem}

The proof is a straightforward application of the CPF condition~\eqref{eq:cpf}, and is provided in the \REV{supplemental material}.

\emph{Note 1.}  For the special case $N=1$ and a simply connected mesh, Theorem~\ref{thm:cpf} leads to a global parameterization $u, v \tin \xR^{|\mV|}$. Its existence is guaranteed since Equation~\eqref{eq:cpf} implies that the classic discrete \emph{edge-based} curl operator~\cite[Def. 3]{Polthier2003identifying} is $0$ for the vector fields given by the rows of $dr_i^{-1}$. 

\emph{Note 2.} Our definition is related to Integer Grid Maps~\cite{bommes2013integer}, but not exactly the same. Specifically, we do not enforce integer transitions between charts, nor integer singularity locations. We further discuss this in Section~\ref{sec:int_mesh_gen}.

\subsection{Relation to other approaches}
\label{sec:relation}
The existence conditions of an Integer Grid Map, i.e. a parameterization s.t. the pushforward of an integer grid leads to a valid mesh, have been given previously~\cite{bommes2013integer,kalberer2007quadcover}. 
The full problem leads to a difficult mixed integer optimization problem, which, to the best of our knowledge, is not solved directly in any approach.
Methods vary in the \emph{decomposition} of the problem, and in the method that is used to solve each part.

\paragraph{QuadCover, MIQ, Globally Optimal (GO), and IGM} The most common decomposition is to first compute a \emph{cross-field} either by smoothing the principal curvature directions (QuadCover~\cite{kalberer2007quadcover}), by solving a mixed integer problem (MIQ~\cite{bommes2009mixed}), or by solving a convex complex-valued problem (GO~\cite{knoppel2013globally}). In the second step, a parameterization is computed which aligns with the computed cross field and adheres to the integer transition constraints, either by a mixed integer formulation (MIQ~\cite{bommes2009mixed}), or using a branch-and-cut algorithm (IGM~\cite{bommes2013integer}).

This decomposition has several disadvantages. First, the sizing of the elements is only addressed in the second step, and thus cannot influence the topology (the singularity locations) of the resulting mesh. Furthermore, integrability is only enforced in the second step (by explicitly computing the parameterization functions), and thus the resulting parameterization can be arbitrarily far from the input cross-field.
Finally, using a cross field implies that the resulting gradient directions are orthogonal, which is not always preferable.

Our approach does a \emph{different decomposition}. The only part which is not addressed in our optimization are the integer transitions between charts. All the other components: integrability, sizing and alignment are all addressed in the vector field optimization stage. This allows, for example, the singularities to move to modify the sizing of the elements and not only the directions.

\paragraph{FrameFields, Metric Customization, Simplex Assembly and PolyVectors} A different first step is to allow for non-orthogonal non-unit length fields, and then proceed with the second parameterization step as described above. This is done by either computing a new 3D geometry that is prescribed by a given metric (FrameFields~\cite{Panozzo2014Frame}), by enforcing a different intrinsic metric using a different connection on the input mesh (Metric Customization~\cite{Jian2015Frame}), or by generalizing the complex-valued formulation to a complex-polynomial one (PolyVectors~\cite{diamanti2014designing}). None of these approaches guarantees integrability, and thus when computing the parameterization in the second step, the gradients of the parameterization may differ greatly from the computed fields. 
The FrameFields and Metric Customization approaches optimize for a $2x2$ matrix per face, as do we. However, they consider only the SPD component (without the planar rotation), whereas we compute the full linear map (without the translation). Some notion of integrability of the computed metric exists for both approaches, the first since they construct a new mesh embedding, and the second since the metric of two neighboring triangles should agree on the length of the edge. Note, however, that this constraint is weaker than ours, as we enforce a seamless quantized rotation. The Simplex Assembly~\cite{Fu2016computing} method computes the full affine map per triangle (including the translation part), yet they do not take into account field alignment, or allow for quantized rotational jumps. 

\paragraph{Integrable PolyVectors, Commuting PolyVectors and Discrete Integrable Frame Fields}
The only three approaches that allow for non-orthogonal fields and optimize integrability are Integrable PolyVectors (IPV)~\cite{Diamanti2015Integrable}, Commuting PolyVectors (CPV)~\cite{sageman2019chebyshev} and Discrete Integrable Frame Fields (DIFF)~\cite{liu2020practical}. Integrability is defined for IPVs through the \emph{PolyCurl}, which is $0$ if there exists a matching between neighboring faces which leads to a zero edge-based curl for all the vectors of the PolyVector. However, IPVs describe a larger space than CPFs, and it is not clear how to restrict them to represent only vector fields that are the coordinate vector fields for some parameterization. In CPVs, integrability is defined through the \emph{discretization} of the Lie bracket, which in general does not guarantee discrete integrability. Finally, in DIFFs integrability is defined on the pulled-back edges yet with a known matching, unlike our approach where we allow for rotational jumps in the pulled back edges. 

\section{CPFs Optimization Framework}
\label{sec:optimization}
Given an input mesh $\mM$, we setup an optimization framework for computing CPFs $U,V$. The CPFs are then used to compute a parameterization and a regular mesh, as explained in Section~\ref{sec:meshing}. 

\subsection{Variables}
Our main variables are $U,V \tin \xR^{2\nof}$ , the candidate CPFs, represented using a local basis $B$. 

\subsection{Hard Constraints}
We have two hard constraints. The LICO constraint:
\begin{equation} 
\tag{C1}
\langle U_i, -JV_i \rangle > 0, \quad \forall t_i \tin \mF,
\label{eq:constr_lico}
\end{equation}
and the CPF constraint:
\begin{equation} 
\tag{C2}
\bigl(dr_i^{-1}(e_{ij})\bigr)^N = \bigl( dr_j^{-1}(e_{ij}) \bigr)^N, \quad \forall e_{ij} \tin \mEi. 
\label{eq:constr_cpf}
\end{equation}

\subsection {Soft Constraints}
Conditions on alignment, sizing and orthogonality are required in general for parameterization.

\subsubsection{Alignment} 
Let $d_i \neq 0$ be a direction in the face $t_i \tin \mF$, and $\aln{i}$ be an even divisor of $N$. Our alignment constraint is:
\begin{equation}
\tag{C3}
\Im\big((dr_i^{-1}(d_i)\big)^{\aln{i}/2}) = 0.
\label{eq:constr_alignment}
\end{equation}

\begin{restatable}{lemma}{alignlemma}
\label{lico_lemma}
If $U,V$ are LICO and the constraint~\eqref{eq:constr_alignment} is fulfilled, then one of the pushed forward grid directions in the face $t_i$ will be parallel to $d_i$.
\end{restatable}

\REV{The proof is provided in Appendix~\ref{appendix:lemma_alignment}.}

\subsubsection{Sizing}
The element sizing is determined by the scale of the CPFs $U,V$. We add a constraint on the size of the elements under a user-defined pointwise metric $g$. The metric is the identity for uniform scale,  a diagonal matrix for isotropic scale, or any symmetric positive definite matrix for anisotropic scale. 
We have:
\begin{equation}
\tag{C4}
\| U_i \|^2_{g_i} = 1, \quad \| V_i \|^2_{g_i} = 1, \quad \forall t_i \tin \mF,
\label{eq:constr_sizing}
\end{equation}
where $g_i \tin \xR^{2\times2}$ is a positive definite matrix defined in the local basis of $t_i \tin \mF$, and $\|x\|^2_{g} = x^T g x$.

\begin{figure*}[t]
    \centering
    \includegraphics[width=.69\linewidth]{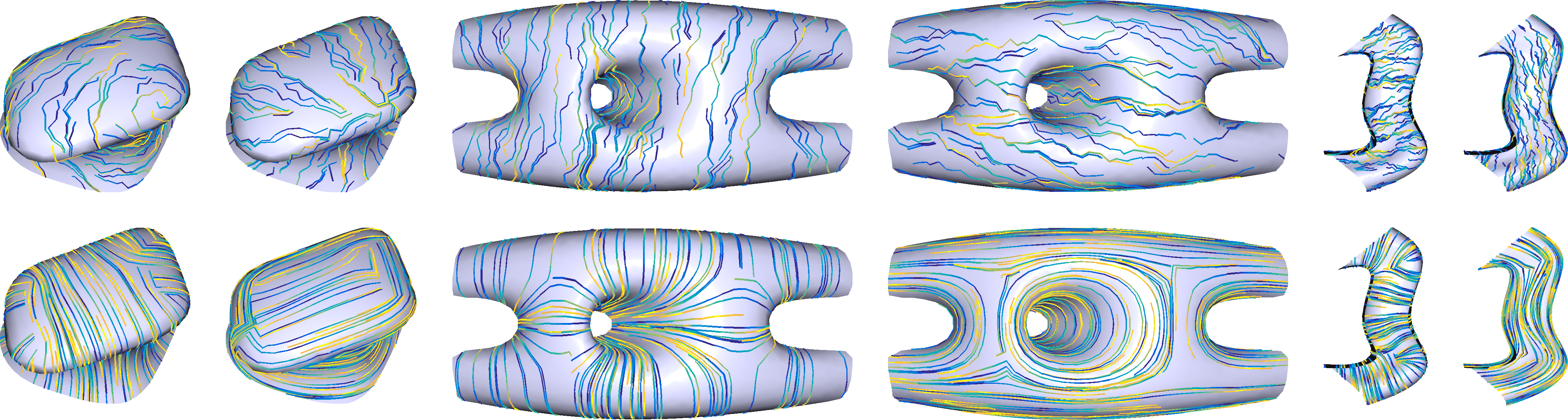}    
    \includegraphics[width=.29\linewidth]{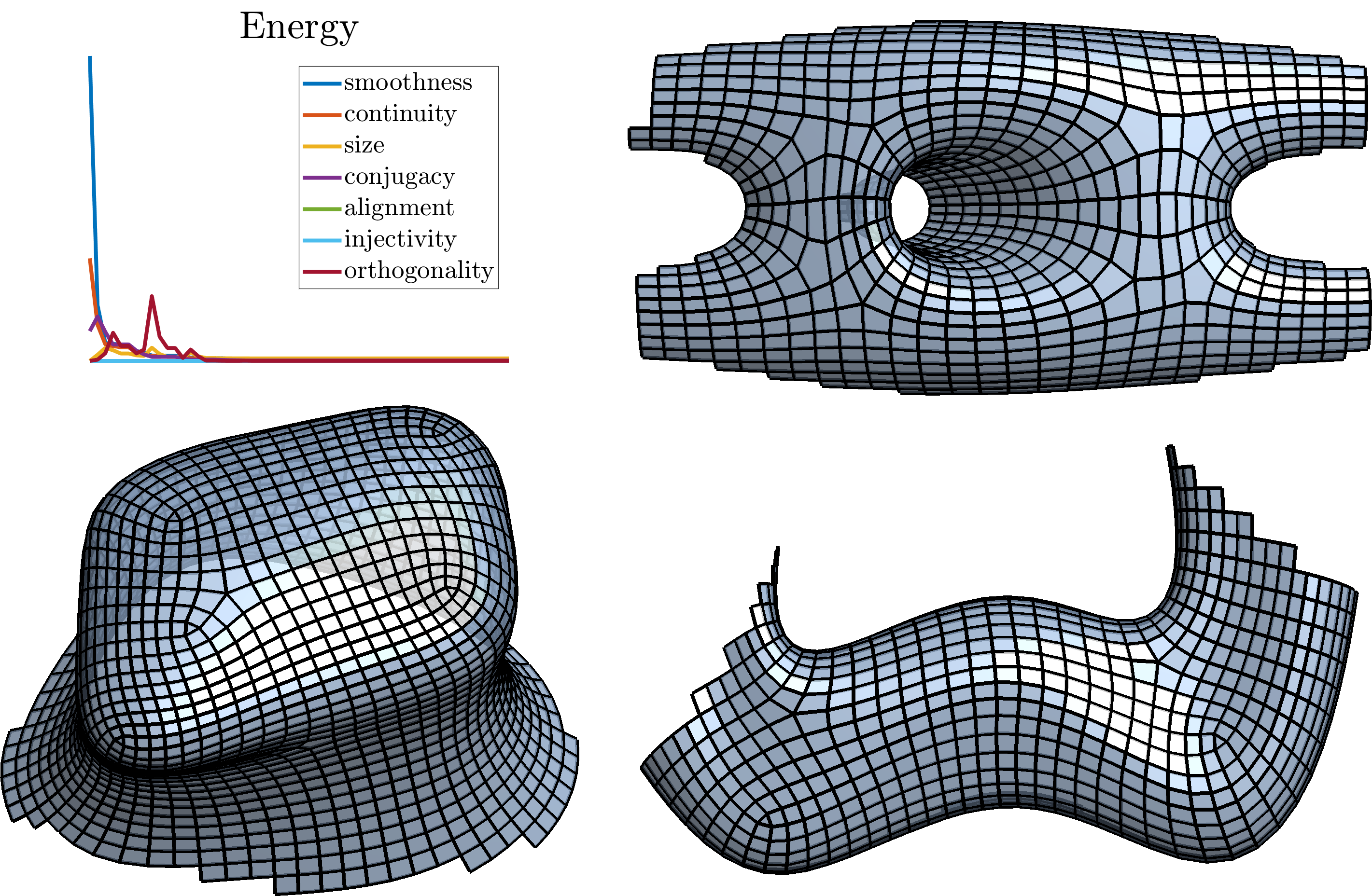}
    \caption{(left) Examples of convergence from a random initialization, with the initial (top) and final (bottom) $U,V$s shown. (right) the resulting quad meshes, and the energy graph for the train station model.}
    \label{fig:random_init}
\end{figure*}

\subsubsection{Orthogonality}
The elements' angles are determined by the pointwise inner product between $U,V$. We add an orthogonality constraint under a symmetric inner product $s$ chosen by the user. The inner product can be the identity for orthogonality, or the shape operator for conjugacy, for example.
We have:
\begin{equation}
\tag{C5}
\langle U_i, V_i \rangle_{s_i} = 0,
\label{eq:constr_ortho}
\end{equation}
where $s_i \tin \xR^{2\times2}$ is a symmetric matrix, defined using the local basis of the face $t_i \tin \mF$, and $\langle x, y \rangle_{g} = x^T g y$. 

\subsection{Relaxed Problem and Optimization Approach}
\label{sec:relaxed_opt_problem}
The constraints we require lead to a non-linear non-convex optimization problem. As a design choice, we opted for posing the problem as an unconstrained non-linear least squares problem, and solving it with the Levenberg-Marquardt~\cite{gavin2020levenberg} algorithm. 

For the LICO constraint~\eqref{eq:constr_lico}, we follow previous work on preserving local injectivity~\shortcite{schuller2013locally}, and formulate it using a spline barrier function.
This approach requires a valid initial guess, and therefore is not applicable to the CPF constraint~\eqref{eq:constr_cpf}. Hence, we opt for a penalty approach~\cite[Sec 4.2]{leyffer2010nonlinear}. 
Specifically, we reformulate the constraint as a quadratic objective, and use an increasing sequence of penalty parameters.
For simplicity, we introduce auxiliary variables $\pf \tin \xC^{\noei}$, which represent the ``consensus'' pulled back edge through the local transformations at its neighboring triangles. 
The soft constraints~\eqref{eq:constr_alignment},~\eqref{eq:constr_sizing},~\eqref{eq:constr_ortho} are reformulated as (potentially weighted) quadratic objectives.

Our relaxed optimization problem is given by:
\begin{mini}|l|
{U,V\tin\xR^{2\nof}, \pf\tin\xC^{\noei}}{E^s(U,V,\pf) + \beta E^{\Psi_c}(U,V,\pf) + E^{\Psi}(U,V)}
{}{}
\label{eq:opt}
\end{mini}
where $E^s$ is a smoothness objective, $E^{\Psi_c}$ is the penalty objective for the CPF constraint~\eqref{eq:constr_cpf}, and $\beta$ is the penalty parameter. Finally, $E^{\Psi}$ is the objective for the rest of the constraints:
\begin{equation}
E^{\Psi}(U,V) = E^{\Psi_{lico}}(U,V) +  E^{\Psi_l}(U,V) + E^{\Psi_s}(U) + E^{\Psi_s}(V) + E^{\Psi_o}(U,V),
\label{eq:opt_psi}
\end{equation}
where $E^{\Psi_{lico}}$,$E^{\Psi_l}$, $E^{\Psi_s}$, $E^{\Psi_o}$ are penalty objectives for the LICO~\eqref{eq:constr_lico}, alignment~\eqref{eq:constr_alignment}, sizing~\eqref{eq:constr_sizing} and orthogonality~\eqref{eq:constr_ortho} constraints.

The objectives are defined either per interior edge (smoothness and CPF penalty), or per face (all the others). For each such objective we specify the formulation for a single element. The total energy is computed by averaging all the elements' objectives values.  
All the gradients are provided in the \REV{supplemental material}.

\begin{figure*}[t]
    \centering
    \includegraphics[width=.9\linewidth]{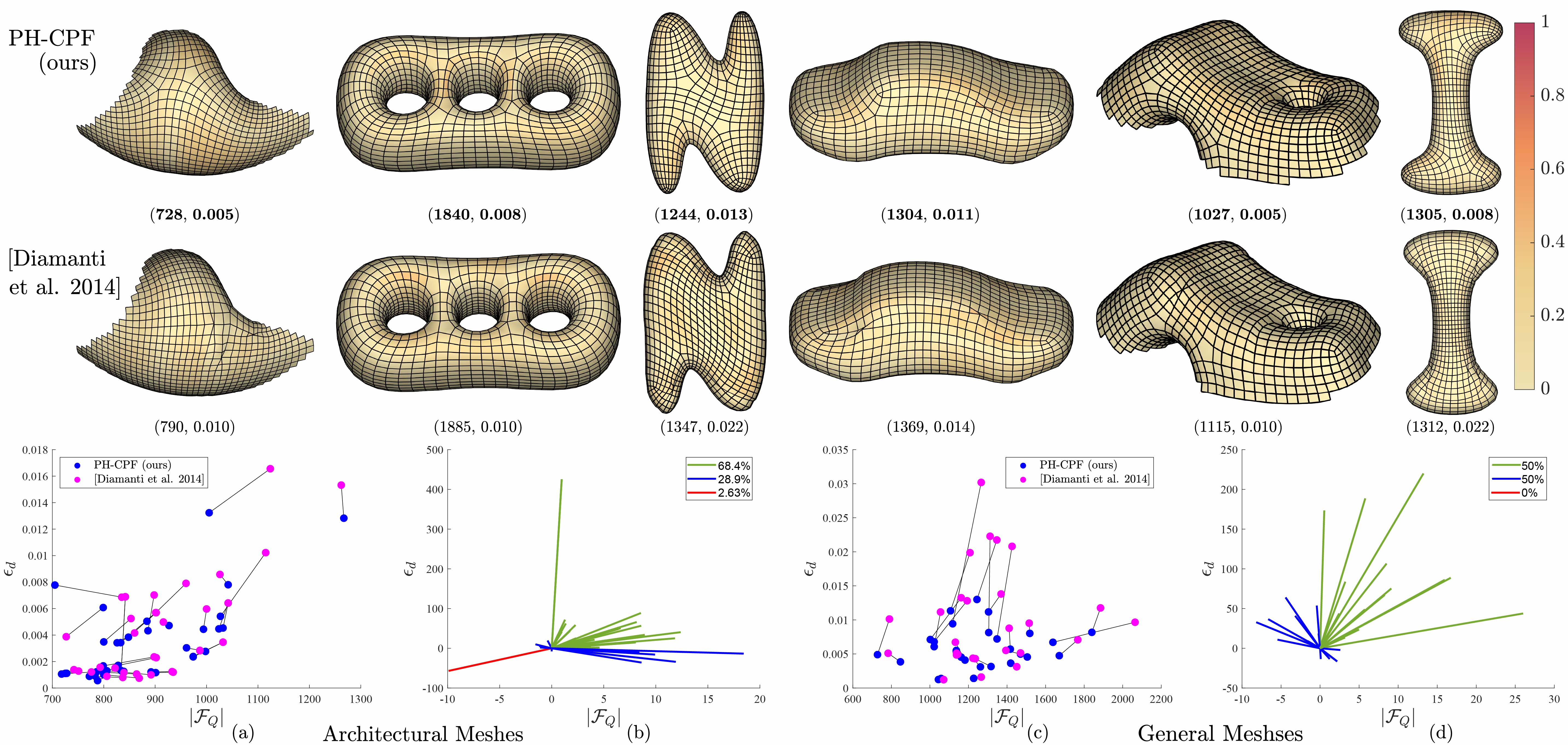}
    \caption{Computing planar quad meshes using our method and PolyVectors~\cite{diamanti2014designing}. (top) Some examples, with the number of faces $|\qF|$ and the Hausdorff distance $\hdist$ to the input mesh in parentheses. (bottom) (a,c) Scatter graph of $|\qF|$ vs $\hdist$, (b,d) tradeoff graph, see the text for details. Green lines indicate our method achieves a lower distance with fewer faces.}
    \label{fig:pq}
\end{figure*}

\subsubsection{LICO penalty}
We use the barrier function designed for Locally Injective Mappings (LIM)~\cite{schuller2013locally}, where the barrier term is applied to $s_i$:
\begin{equation}
E^{\Psi_{lico}}_{i}(U,V) =\phi_i^2\bigl(s_i(U,V) - \epsilon^{\Psi_{lico}} \bigr), \quad s_i(U,V) = \langle U_i , -J V_i \rangle,
\end{equation}
where $\phi_i$ is~\cite[Eq.(5),(6)]{schuller2013locally}:
\begin{equation}
\phi_i(x) = 
    \begin{cases}
       \ \infty, &\quad x \leq 0\\
       \ \frac{1}{g_i(x)}-1, &\quad 0 < x < \bar{s}_i \\
       \ 0, &\quad x \geq \bar{s}_i \\ 
     \end{cases},    
     \quad g_i(x) = \invf{\bar{s}_i^3}x^3 - \frac{3}{\bar{s}_i^2}x^2 + \frac{3}{\bar{s}_i}x.
     \label{eq:barrierTerm}
\end{equation}
Here $\bar{s}_i$ is the value of $s_i$ in the previous iteration times a constant $s=0.1$, and we take $\epsilon^{\Psi_{lico}} = 10^{-3}$. We take the barrier with respect to the previous iteration instead of a fixed constant to be more resilient to different scalings.

\subsubsection{CPF penalty}
The $z$ variables conciliate the pullback of $t_i$ and the pullback of $t_j$:
\begin{equation}
E^{\Psi_c}_{ij}(U,V,\pf) = \bigl|\bigl(dr_i^{-1}(e_{ij})\bigr)^N - z_{ij}\bigr|^2 + \bigl|\bigl(dr_j^{-1}(e_{ij})\bigr)^N - z_{ij}\bigr|^2,
\end{equation}
where we measure absolute difference between complex numbers.

\subsubsection{Soft constraints}
Let $\Psi(U,V) = 0$ be one of the constraints~\eqref{eq:constr_alignment}, \eqref{eq:constr_sizing} or ~\eqref{eq:constr_ortho}. Then the corresponding objective $E^{\Psi}$ is given by $|\Psi(U,V)|^2$.

\subsubsection{Smoothness}

The regularizer that is often used in $N$-RoSy field generation is the vector Laplacian, modified to accommodate the rotational symmetry~\cite{knoppel2013globally}. 
This formulation favors vector fields which are similar in neighboring faces when raising them to the $N$-th power in the appropriate basis.
The integrability constraint requires that the $N$-th power of the pulled back edge agrees across neighboring triangles, yet the orthogonal direction is not constrained.
Hence, for the regularizer we use:
\begin{equation}
E^{S}_{ij}(U,V,\pf) = \bigl|\bigl(dr_i^{-1}(J e_{ij})\bigr)^N - z_{ij}\bigr|^2 + \bigl|\bigl(dr_j^{-1}(J e_{ij})\bigr)^N - z_{ij}\bigr|^2.
\label{eq:harmonic-cpf}
\end{equation}
Note that, similarly to the CPF constraint, for $N=1$ Equation~\eqref{eq:harmonic-cpf} is equivalent to requiring that the discrete \emph{edge-based} divergence~\cite[Def. 4]{Polthier2003identifying} is $0$ for the vector fields given by the rows of $dr_i^{-1}$.

\subsection{Implementation Details}
\paragraph{Initialization}
The only constraint on our initialization is that the initial $U,V$ are LICO. The initial $z$ is always computed as the average of the $N$-th power of the pulled-back edges, namely: $z_{ij} = \big(dr_i^{-1}(e_{ij}))^N + dr_j^{-1}(e_{ij}))^N\big)/2$. 

The best initialization is highly dependent on the application. For simple surfaces,
we can obtain a good solution from random (valid) initialization. For example, taking $U$ to be a random unit length vector field, $V_i = JU_i$, and rescaling both according to $g_i$ if a sizing constraint given. Figure~\ref{fig:random_init} (left) shows examples of optimized CPFs $U,V$ (bottom) computed from a random initialization (top), with curvature based sizing (Section~\ref{sec:ph_scaling}) and the conjugacy and orthogonality constraints. Figure~\ref{fig:random_init} (right) shows the resulting quad meshes and the convergence graph of the energy for the train station model (the other graphs are similar). Note that for quad meshes, conjugacy and orthogonality combined imply alignment to curvature directions, although we did not enforce this explicitly.

\paragraph{Optimization} We use the Levenberg-Marquardt algorithm~\cite{gavin2020levenberg} with Jacobian scaling for optimization, as implemented in Matlab, with $\beta_0 = 1$. If required, we repeat the optimization, using the optimized variables as the initial solution, and increasing $\beta$ by a multiplicative factor of $2$, until mesh integration with integer seams succeeds, or a maximum of $10$ iterations is reached. In most of the results we show, the first iteration sufficed.

\paragraph{Balancing the energy terms}
\sloppy Different energy terms in Equations~\eqref{eq:opt},~\eqref{eq:opt_psi} have different units, which are mesh-dependent. Without balancing, different weighting parameters are required for different meshes. Specifically, $E^{s}$, $E^{\Psi_c}, E^{\Psi_l}$ are defined on the $N$-th power of a pulled back vector, whereas the rest of the objectives are defined on the input mesh. Hence, we rescale these energies, taking $1/\lambda^{2N} E$, where $\lambda$ is the average expected $2D$ edge length, given by $\lambda = \frac{1}{3\nof}\sum_{i=1}^{\nof} \sum_{e_{ij} \in t_i} \|e_{ij}^T \, g_i \, e_{ij}\|$.

\subsection{Meshing with CPFs}
\label{sec:meshing}
Given the computed $U,V$ vector fields, the seamless parameterization can be theoretically constructed using the definitions of the parameterization functions $u,v$ that are given in the proof. However, even if the CPF constraint is accurate to machine precision, such an implementation would require computing integer translations per edge, which is unnecessarily expensive. 
Thus, we follow the standard procedure for extracting a mesh~\cite{vaxman2019directional}. Specifically, integrating the parameterization with the prescribed quantized rotations on the edges using a Poisson solve, computing a full integer grid map by quantizing the translations~\cite{vaxman2021seamless}, and extracting the mesh by overlaying a grid on the integer grid map. We provide the details in Section~\ref{sec:int_mesh_gen}.

\subsection{Example Application: Planar quad meshing}
\label{sec:pq}
Our main interest in this paper is planar hexagonal meshing, which is explained in detail in the next section. As we introduce a general scheme for parameterization-based meshing, we provide also an example for planar quad meshing, and compare with existing work.
We set $N\!=\!4$, set the sizing to be curvature dependent (Section~\ref{sec:ph_scaling}), set alignment with $N\!=\!4$ to the prominent curvature directions (Section~\ref{sec:ph_alignment}), and set the conjugacy and orthogonality constraints. For initialization we use the smoothed curvature directions.

For the comparison, we use the PolyVectors method~\cite{diamanti2014designing} using the publicly available implementation from the Directional library~\cite{vaxman2019directional}.
We note that an implementation of Integrable PolyVectors~\cite{Diamanti2015Integrable} that includes conjugacy optimization of the vector fields (which is necessary for planar quad meshing) is not publicly available. For the competing code we we used the same shape operator as ours and applied the same planarization post processing as for our approach, setting the bound on the same maximal planarity of $1\%$. 
We compare the Hausdorff distance to the source mesh $\hdist$ and the number of faces $|\qF|$ produced by both methods for a set of architectural meshes and a set of general meshes. Figure~\ref{fig:pq} (top) shows a few qualitative examples from the sets, as well as $(|\qF|, \hdist)$. Figure~\ref{fig:pq} (bottom, (a,c)) shows a scatter plot of $\hdist$ vs $|\qF|$ for both methods, where connected points represent the same mesh. The graphs (b,d), show the difference between our results, represented as the origin, to Diamanti's results, in percentages. Green lines, in the first quadrant, represent meshes for which our method produces a lower $\hdist$, with less elements. Blue lines, in the second and third quadrant, represent the trade-off between the methods, one achieved a better distance, while the other, less elements. The fourth quadrant, with a single red line, shows a case where Diamanti's method produced better results in both criteria. Note that most meshes fall in the first quadrant, indicating that for the same planarity values, our method achieved a lower Hausdorff distance with a smaller number of elements.

\begin{figure}[t]
    \centering
    \includegraphics[width=.49\linewidth]{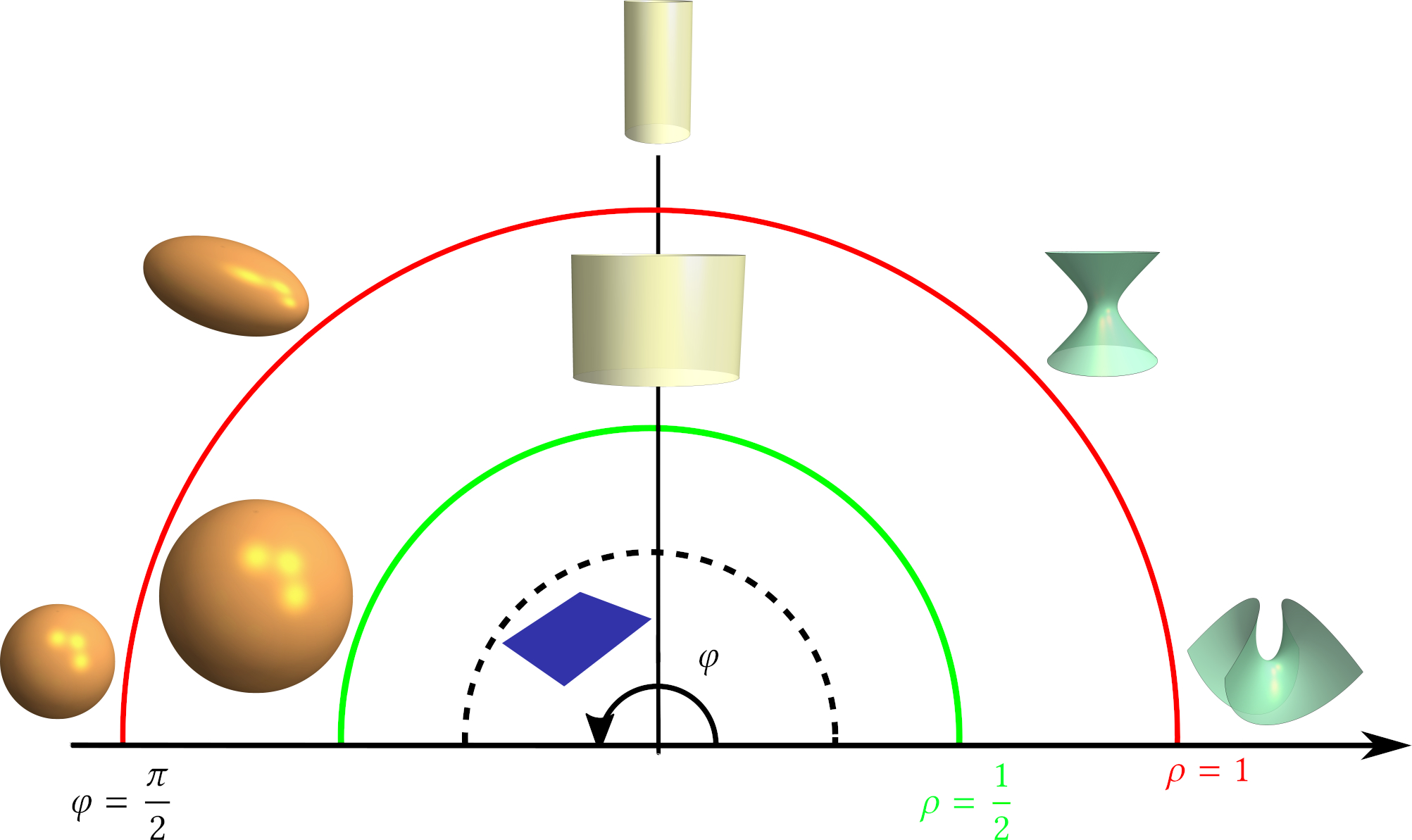}
    \includegraphics[width=.49\linewidth]{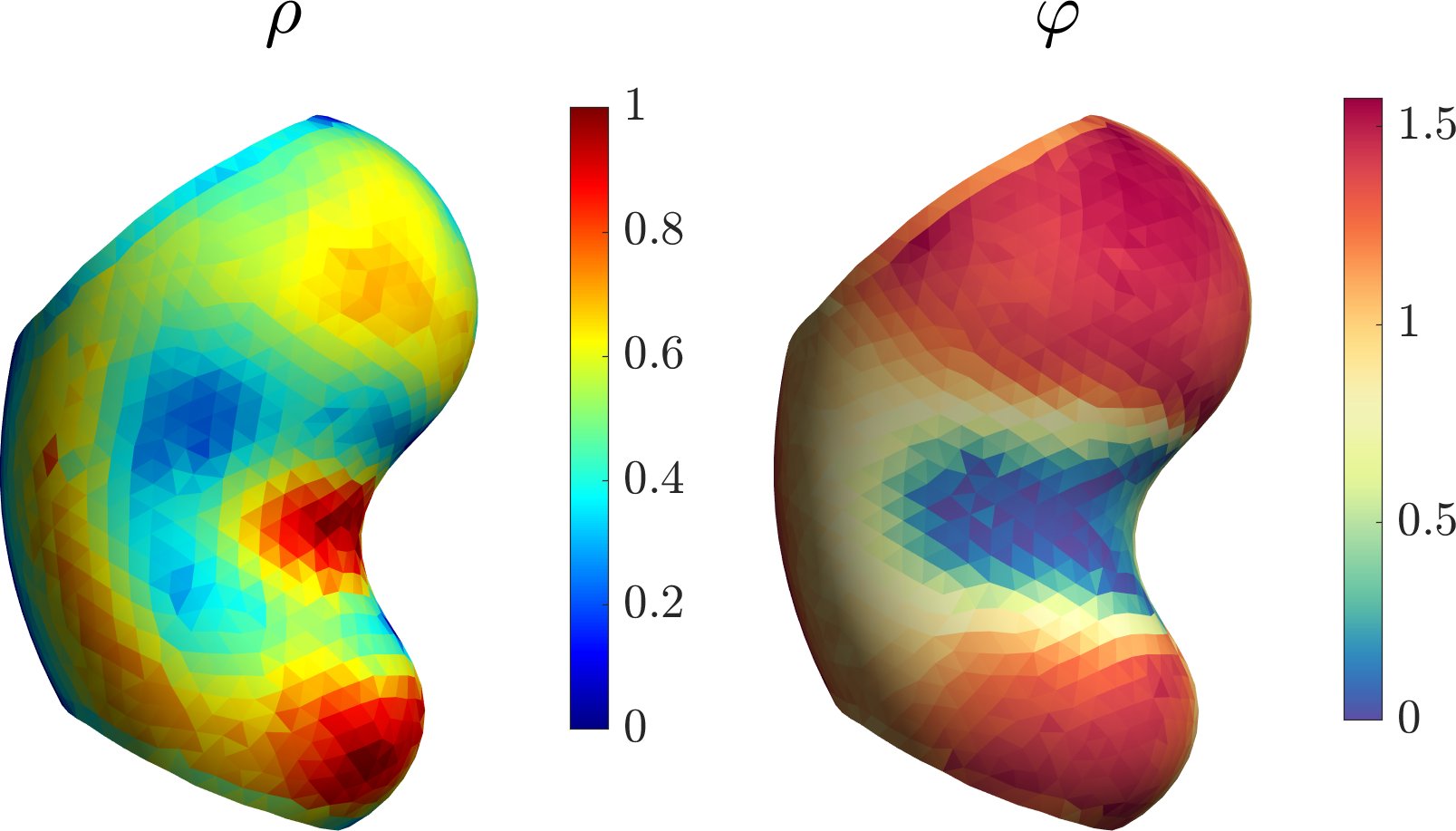}
    \caption{Curvature regions and their polar representation with $\varphi$ and $\rho$.}
    \label{fig:rho_phi}
\end{figure}

\section{Planar Hex Meshing}
\label{sec:ph_meshing}
Our main goal is generating planar hexagonal meshes which approximate a given surface, and we use the CPF framework for the optimization. In the following we describe the constraints that we apply, which are unique for hexagonal meshes. 

\subsection{Background: Curvature and the Dupin Indicatrix}
\paragraph{Notations and polar curvature representation.}
The maximal and minimal curvatures at $p \tin \sM$ are denoted by $\kmin(p)$, $\kmax(p) \tin \xR$, respectively, with corresponding maximal and minimal curvature directions denoted by $\dmin(p), \dmax(p) \tin \tpm$. The Gaussian and mean curvatures are denoted by $\kgauss(p), \kmean(p)$, respectively, and the shape operator by $\shapeop_p$.  The normal curvature in the direction $v \tin \tpm$ is denoted by $\kn(v)$. Two directions $x,y \tin \tpm$ are conjugate if they are orthogonal with respect to the shape operator, namely $\langle x, y \rangle_{\shapeop_p} = x^T \shapeop_p y = 0$. We drop $p$ when the meaning is clear.

Curvature on triangle meshes is notoriously difficult to compute robustly, requiring special treatment in the context of curvature aligned meshing~\cite{campen2016scale}. 
Thus, similarly to Niesser et al.~\shortcite{nieser2011hexagonal}, we use a polar representation of the curvatures: 
\begin{equation*}
\varphi = \left|\operatorname{atan}\left(\frac{\kmin + \kmax}{\kmin - \kmax}\right)\right|, \quad \rho = \sqrt{\kmin^2 + \kmax^2},
\end{equation*}
where  $\varphi$ is a smooth representation of the curvature type of the point (elliptic $\kgauss > 0$, parabolic $\kgauss = 0$ or hyperbolic $\kgauss < 0$), and $\rho$ represents distance from planarity (we normalize it between $0$ and $1$). Thus, $\varphi$ is well defined only for $\rho > 0$ (see Figure~\ref{fig:rho_phi}).

\begin{figure}[b]
    \centering
    \includegraphics[width=.9\linewidth]{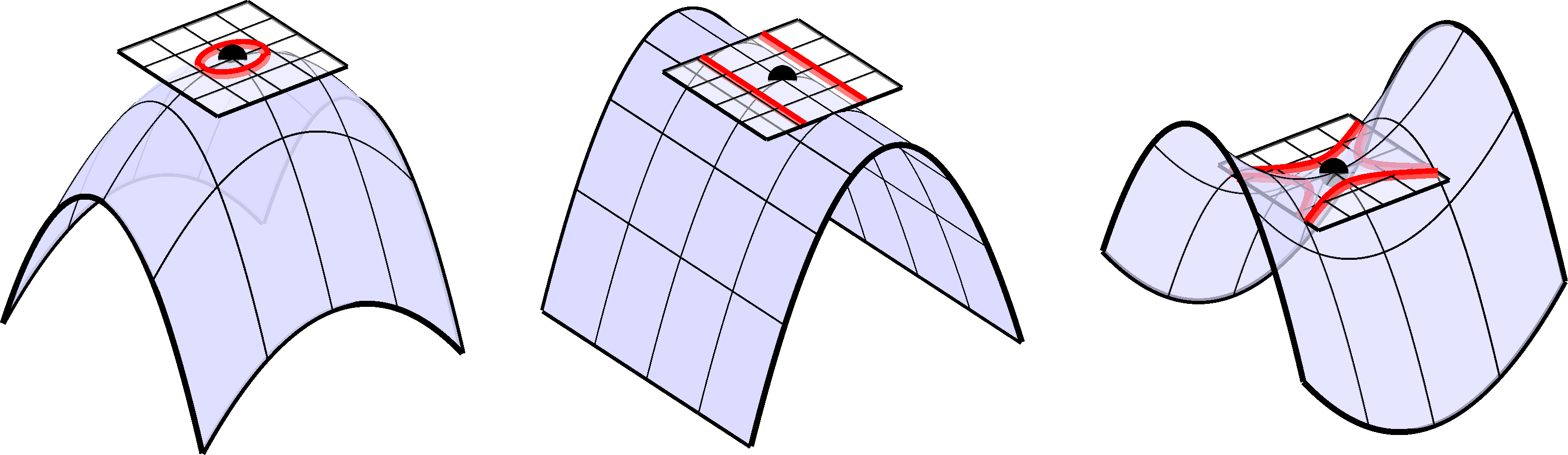}
    \includegraphics[width=.9\linewidth]{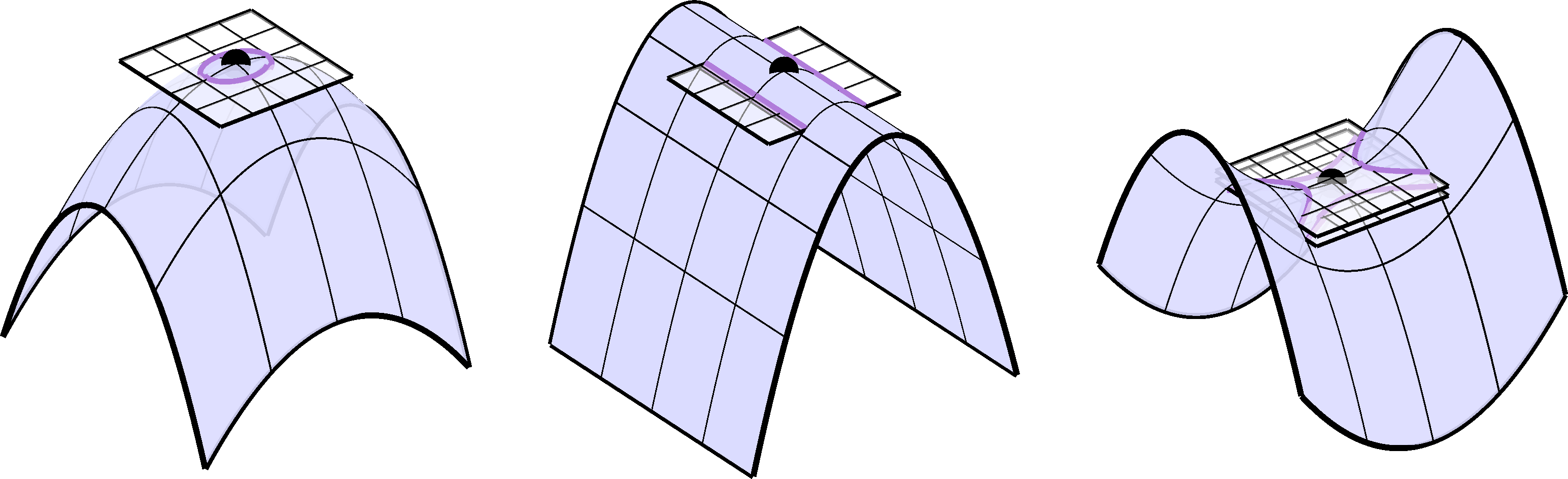}
    \caption{(top) The Dupin indicatrix at $p\tin\sM$. (bottom) The intersection of a plane parallel to the tangent plane and at a small distance from it with the surface is similar to the Dupin Indicatrix. When $p$ is (left) elliptic, (center) parabolic and (right) hyperbolic. }
    \label{fig:dupin}
\end{figure}

\paragraph{The Dupin Indicatrix}
An important notion for computing PH meshes, as has been shown by Wang et al.~\shortcite[Thm. 1]{wang2009note}, is the concept of the \emph{Dupin Indicatrix}~\cite[Chap. IV]{eisenhart1960treatise}.

The Dupin indicatrix at an elliptic point $p \tin \sM$ is a curve on the tangent plane $\tpm$, given by the points $p + |\kn(v)|^{-1/2} \hat{v}$ for all unit length vectors $\hat{v} \in \tpm$. In local curvature coordinates, i.e. taking $p$ as the origin and $\dmin,\dmax$ as the local $x,y$ axes, the Dupin indicatrix is the ellipse $|\kmin|x^2 + |\kmax|y^2 = 1$. At a hyperbolic point it is given by the conjugate hyperbolas $\kmin x^2 + \kmax y^2 = \pm 1$, and at parabolic points by two parallel lines, either $|\kmin| x^2 = 1$ or $|\kmax| y^2 = 1$ (see Fig.~\ref{fig:dupin} (top)).

The relation between the Dupin indicatrix and planar meshing is clearly visualized when considering the local geometry at $p$ (see Figure~\ref{fig:dupin_plane_and_hex_notations}, inspired by~\cite[Fig.2]{lisle2003dupin}).  Consider a normal plane in the direction $\hat{v}\tin\tpm$ at $p$, and its intersection curve with the surface. Further, take a plane parallel to $\tpm$ and at distance $\h$ from it, such that it intersects the surface, and let $2d$ be the distance between the two intersection points of the surface and these two planes. 
If $r$ is the radius of curvature at $p$, then for small enough $\h$ we have $d^2 = r^2 - (r-\h)^2 =  2\h|\kn(\hat{v})|^{-1} + O(t^2)$.
Hence, up to first order in $t$, the intersection curve between the cutting plane and the surface is a uniform scale of the dupin indicatrix (see Figure~\ref{fig:dupin} (bottom)). We use this property when we design our sizing constraints for PH meshing, as we discuss next.

\begin{figure}[t]
    \centering
    \includegraphics[width=.63\linewidth]{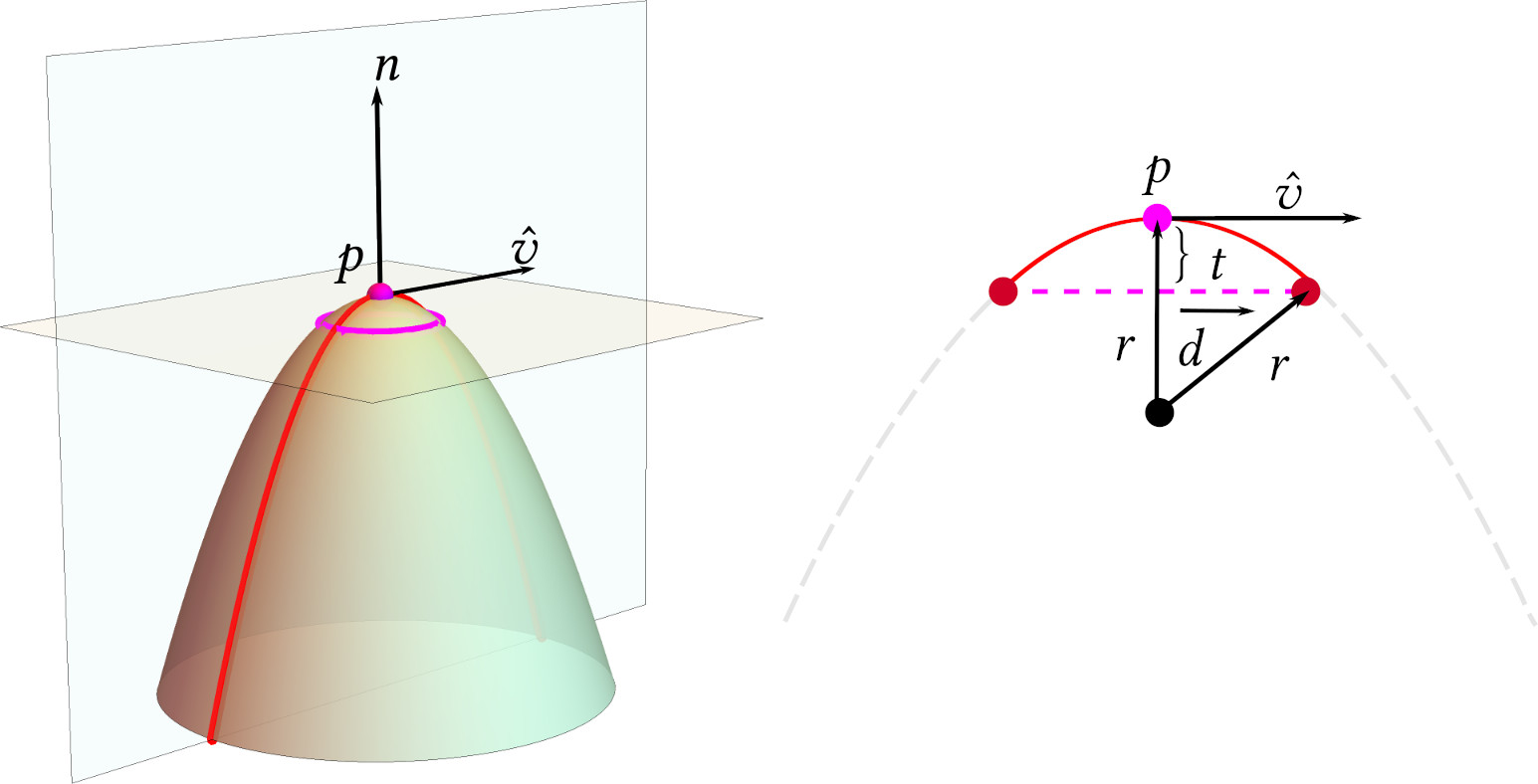}
    \includegraphics[width=.33\linewidth]{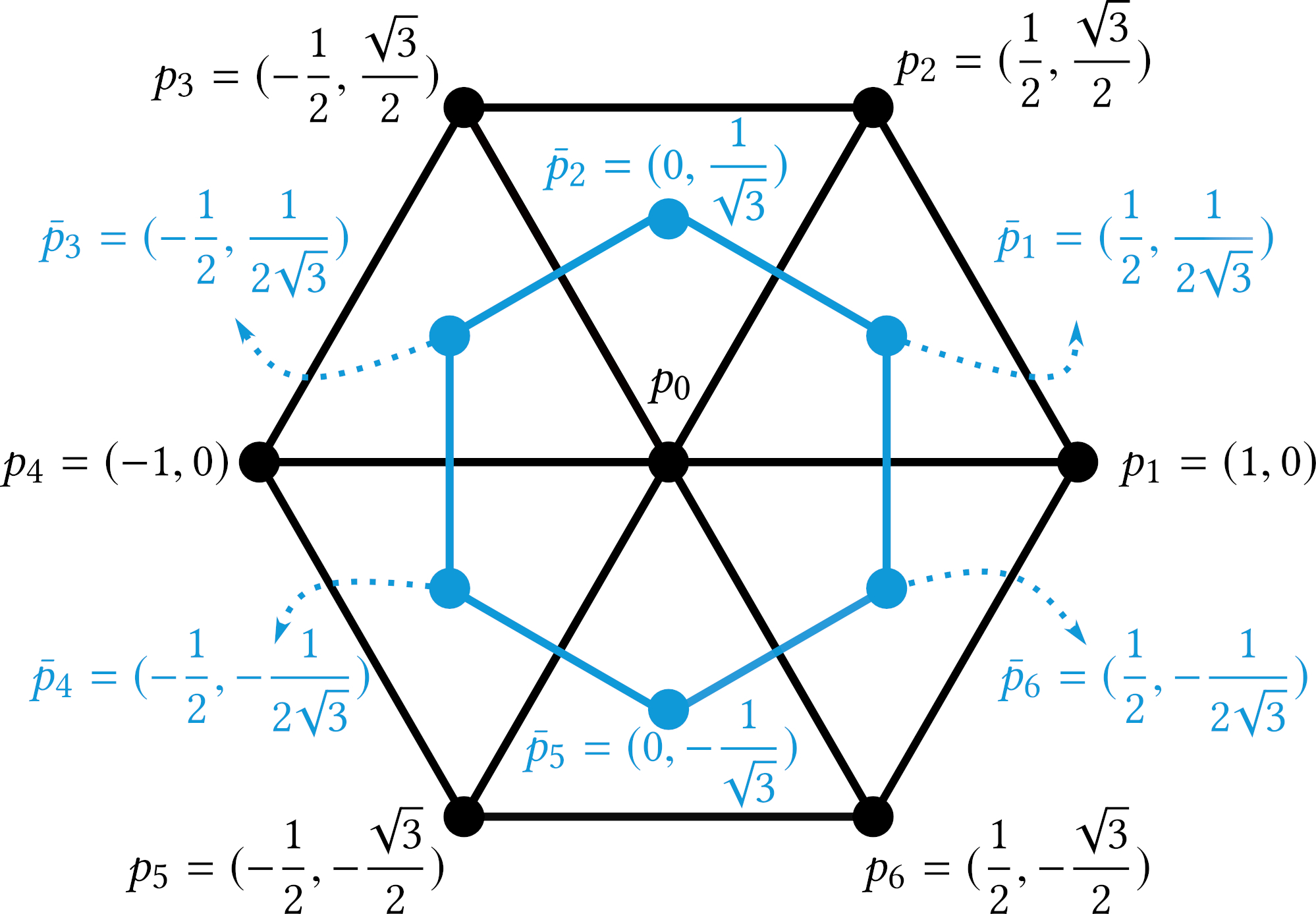}
    \caption{(left) The local geometry near a point $p\tin\sM$, with a normal plane in the direction $\hat{v}\tin\tpm$, and a cutting plane parallel to $\tpm$ at distance $\h$ from it. (right) Notations for the pushforward of one ring of triangles, and the corresponding dual hexagon.}
    \label{fig:dupin_plane_and_hex_notations}
\end{figure}

\subsection{Dual Meshing}
Givan a parameterization $r : \sU \subset \xR^2 \to \Omega \subset \sM$, and a point on the triangular grid $\s{p} \tin \h\mathbb{A}^2 \cap \sU$, let $\s{p}_i, i \in [1,..,6]$ be its one-ring grid neighbors, and  $\bar{\s{p}}_i = (\s{p} + \s{p}_i + \s{p}_{i+1})/3$ be the barycentric dual vertices (see Figure~\ref{fig:dupin_plane_and_hex_notations} (right)). Next, consider the points $r(\bar{\s{p}}_i) \tin \sM$ and the corresponding hexagon $\shp$  generated by linearly connecting the pushed forward points to yield a closed polyline. This hexagon represents one element in the hexagonal remeshing of the input surface. 
Our general goals are: 
\begin{itemize}
\item \emph{Planarity.} The vertices of $\shp$ are co-planar, and $\shp$ is a simple polygon.
\item \emph{Approximation.} The Hausdorff distance between the surface and $\shp$ is minimal among all hexagons with the same area approximating the surface at $p = r(\s{p})$.
\item \emph{Fairness.} $\shp$ has horizontal and vertical reflection symmetry.
\end{itemize}

\subsection{Obstruction to convexity.} 
Parameterization-based dual meshing generates \emph{convex} hexagonal faces. However, the faces of a PH mesh approximating a given surface are inscribed in a quadric which is a uniform scale of the Dupin indicatrix~\cite[Thm. 1]{wang2009note} (see also Fig.~\ref{fig:dupin_plane_and_hex_notations} (left)). Hence, for elliptic, parabolic and hyperbolic regions the expected face shapes are \emph{convex} hex, \emph{brick} hex, and \emph{bow-tie} hex, (see Fig~\ref{fig:hex_shapes}).

\begin{figure}[b]
    \centering
    \includegraphics[width=.9\linewidth]{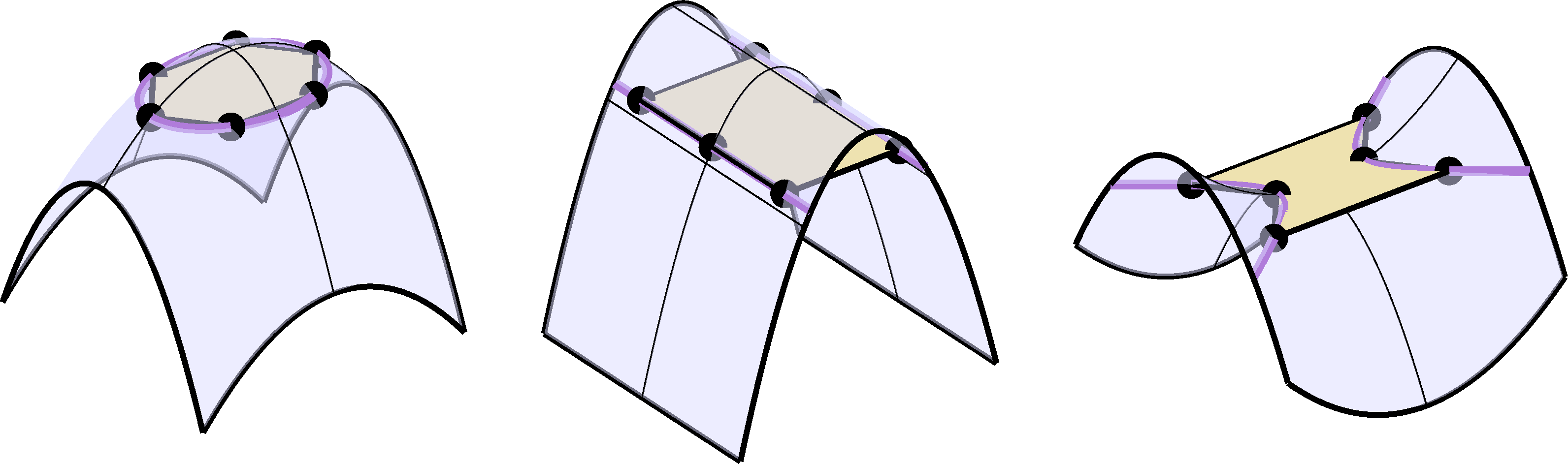}
    \includegraphics[width=.9\linewidth]{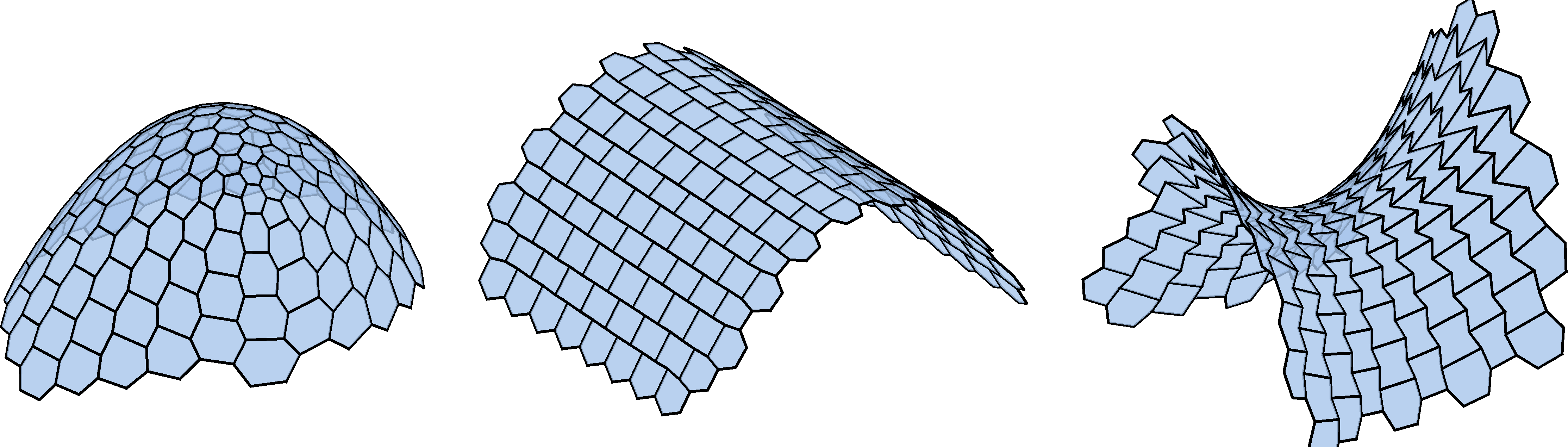}
    \caption{(top) Expected shapes of planar hexes in different curvature regions: convex, brick and bow-tie. (bottom) Quadric surfaces exhibiting these face shapes: elliptic paraboloid, 	parabolic cylinder, and hyperbolic paraboloid.}
    \label{fig:hex_shapes}
\end{figure}

Thus, regardless of the parameterization chosen, we \emph{cannot} expect the mesh faces to be planar, unless the surface has only positive Gaussian curvature.
There exist methods for \emph{planarizing} a hexagonal mesh using a non linear optimization~\cite{tang2014form,Jiang2015polyhedral}. These approaches generate a valid PH mesh, which fulfills our stated goals, \emph{given a good initialization}. We denote a mesh that can serve as the initial solution for a planarization approach as \emph{planarizable}. We compute planarizable meshes using the CPF framework, by designing constraints on the CPFs $U,V$. In the following we describe our pipeline, and the required CPF constraints.

\subsection{Algorithm}
Our algorithm is outlined in the following (see Figure~\ref{fig:pipeline}).
\begin{algorithm}
\SetAlgoLined
\KwIn{$\mM, \delta, \eta, \h$}
\KwOut{$\hM$}
 Compute guiding field from curvature, and use as alignment constraints (Sec.~\ref{sec:ph_alignment} to~\ref{sec:cpf_alignment_constraints}) \;
 Compute scaling constraints from $\delta,\eta$ (Sec.~\ref{sec:cpf_scaling_constraints}) \;
 Add orthogonality (Sec.~\ref{sec:cpf_orthogonality_constraints}) and conjugacy (Sec.~\ref{sec:strips}) constraints\;
 Solve optimization problem for $U,V$  (Sec.~\ref{sec:relaxed_opt_problem}) \;
 Convert $U,V$ to $H_1,H_2,H_3$ (Sec.~\ref{sec:triangular_meshing}) \;
 Integrate $H_i$ into parameterization functions $h_i$ (Sec.~\ref{sec:integration}) \;
Extract primal triangular mesh $\rM$ from $h_i$ (Sec.~\ref{sec:mesh_extraction}) \;
Extract barycentric dual mesh $\rdM$ from $\rM$ (Sec.~\ref{sec:mesh_extraction}) \;
Planarize dual mesh $\rdM$ to the final PH mesh $\hM$ (Sec.~\ref{sec:planarization})
\label{alg:main_alg}
 \caption{An outline of our algorithm}
\end{algorithm}

\begin{figure*}[t]
    \centering
    \includegraphics[width=\linewidth]{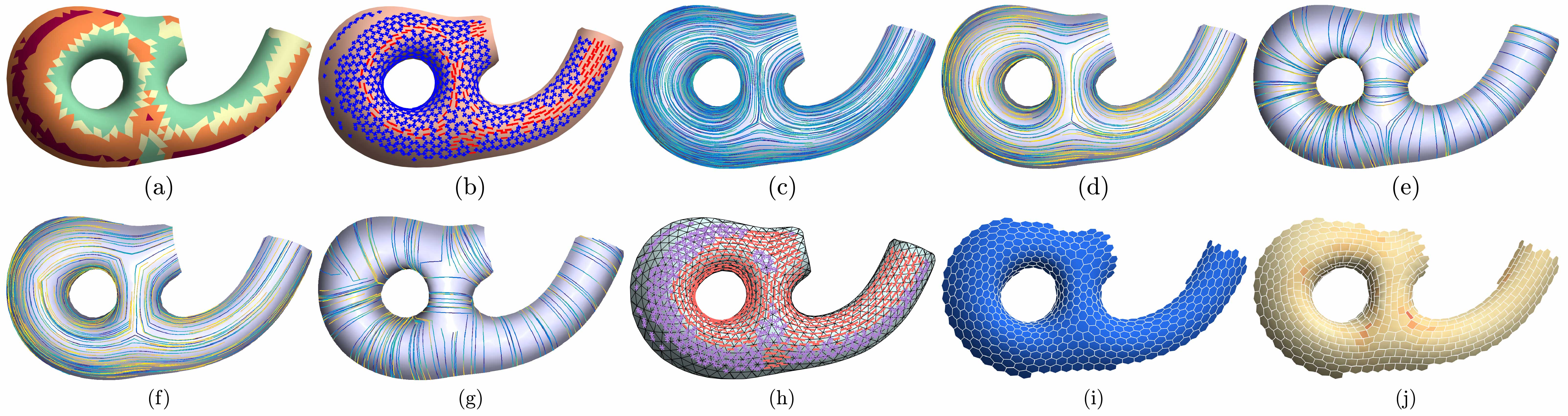}
    \caption{Our algorithm: (a) curvature regions, (b) alignment constraints for guiding field (Equation~\eqref{eq:opt_guiding_field}), (c) guiding field, (d,e) initial $U,V$, (f,g) optimized $U,V$, (h) extracted primal triangular mesh $\rM$, overlayed with the CPF alignment constraints (Equation~\eqref{eq:cpf_alignment_constraints}),(i) dual hexagonal mesh $\rdM$, (j) final planar mesh $\hM$.}
    \label{fig:pipeline}
\end{figure*}

\begin{figure}[t]
    \centering
    \includegraphics[width=.9\linewidth]{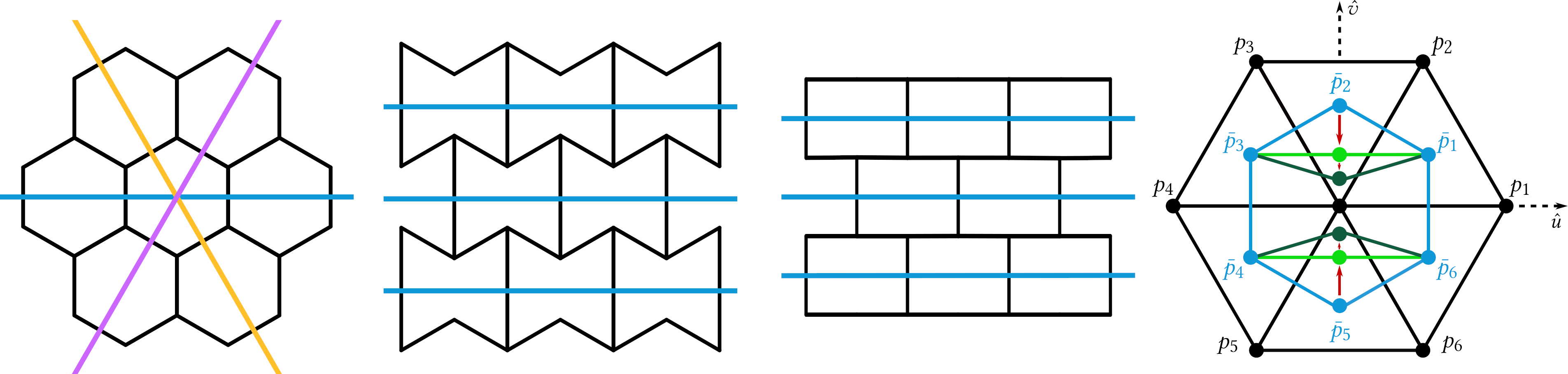}
    \caption{(left) Developable strips of hexes in different curvature regions. (right) A dual convex hex (blue) can be modified to the brick or bow-tie shapes (green) if it is aligned correctly. }
    \label{fig:strips}
\end{figure}

\subsection{Planarity by Conjugacy and Alignment}
\subsubsection{Conjugacy and developable strips}
\label{sec:strips}
Following Wang et al.~\shortcite{wang2009note} we consider strips of hexagons which are developable (see Figure~\ref{fig:strips} (left)). 

The relation between conjugate direction fields and planar meshing is well established~\cite{Liu2006geometric,liu2011general,wang2009note}. Specifically, consider a hexagon $\shp = \{\bar{p}_i \, | \, i \tin [1,..,6]\}$ at the tangent plane at $p \tin \sM$. Using the argument that the discrete ruling directions are conjugate to the central curve of a developable strip, Wang et al.~\shortcite{wang2009note} have shown that for a convex hexagon, every pair of consecutive edge vectors $e_i = \bar{p}_i - p$ fulfills $\langle e_i + e_{i+1}, e_i - e_{i+1} \rangle_{S_p} = 0$.

Plugging in the expressions for the edge vectors as a function of $U_p,V_p$ it is straightforward to show (see the \REV{supplemental material}) that in elliptic regions, these conditions are equivalent to:
\begin{align}
\label{eq:conjugacy_uv}
\langle U_p, V_p \rangle_{S_p} &= 0, \\
\label{eq:conjugacy_uu}
\langle U_p, U_p \rangle_{S_p} &= \langle V_p, V_p \rangle_{S_p}.
\end{align}

Therefore, in the CPF framework, we use the orthogonality constraint $\langle U_i, V_i \rangle_{s_i}
 = 0$, taking $s_i$ to be the shape operator at face $t_i \tin \mF$.  Note that Equation~\eqref{eq:conjugacy_uu} implies $\langle U_p, U_p \rangle_{S_p} = \kn(U_p)\|U_p\|^2 = \kn(V_p)\|V_p\|^2 = \langle V_p, V_p \rangle_{S_p}$. In elliptic regions, where $S_p$ is positive definite and the normal curvatures have the same sign, this is equivalent to $|\kn(U_p)|\|U_p\|^2 = |\kn(V_p)|\|V_p\|^2$,  and hence $U_p, V_p$ trace a uniform scale of the Dupin indicatrix. Thus, this constraint will be covered by the scaling constraint as we explain in Section~\ref{sec:ph_scaling}.

\subsubsection{Alignment} 
\label{sec:ph_alignment}
A strip of convex hexes is planarizable if there are degrees of freedom allowing the hex vertices to move to the required hex shapes (see Figure~\ref{fig:strips} (right)). 
Hence, in parabolic regions the $U$ direction \emph{must} be aligned with the ruling direction, i.e. the direction of \emph{minimal absolute curvature}, which we denote by $\damin$. 
In hyperbolic regions, the $U$ direction should be aligned with either $\dmin$ or $\dmax$, and in elliptic regions alignment is flexible, yet hexes would be easier to planarize and better shaped if $U$ is aligned either with $\dmin$ or $\dmax$.

The minimal absolute curvature direction $d$ is defined by:
\begin{equation}
\label{eq:damin}
\damin = 
\begin{cases} 
      \dmin & |\kmin| < |\kmax| \\
      \dmax & \text{otherwise}
\end{cases},
\end{equation}
and it is not necessarily smooth on the surface. For example, in parabolic regions of positive mean curvature, we have $\kmin = 0, \kmax > 0$, and thus $\damin = \dmin$. On the other hand, in parabolic regions of negative mean curvature, we have $\kmin < 0, \kmax = 0$ and thus $\damin = \dmax$. Furthermore, when the discrete curvature is noisy $\damin$ is not stable in umbilical regions where $\kmin = \kmax$, nor in regions where $\kmin = -\kmax$.

To consolidate the different alignment constraints in different regions, we define \emph{curvature regions} in terms of $\rho, \varphi$ and then compute a smooth $2$-RoSy \emph{guiding field} which will provide alignment constraints for the CPFs computation.

\subsubsection{Curvature regions.} Following Nieser et al.~\shortcite{nieser2011hexagonal}, we use $\rho,\varphi$ to robustly define curvature regions that contribute to the alignment constraints as previously discussed. Let $\mFe,\mFp,\mFh \subset \mF$ denote the faces sets of elliptic, parabolic and hyperbolic curvature regions, respectively. They are defined as follows:
\begin{equation}
\begin{aligned}
\mFe &= \{ \, t_i \tin \mFnp\setminus\mFu \quad | \quad \varphi - \pi/4 \geq \deltap\,\}, \\
\mFp &= \{ \, t_i \tin \mFnp \quad | \quad |\varphi - \pi/4| < \deltap  \,\}\\
\mFh &= \{ \, t_i \tin \mFnp \quad | \quad \varphi - \pi/4 \leq -\deltap \,\},
\end{aligned}
\label{eq:face_sets}
\end{equation}
where $\mFnp = \{ t_i \tin \mF \, | \, \rho > \rhom\}$ are the non planar regions, $\mFu = \{ t_i \tin \mF \, | \, \varphi \geq \pi/2 - \deltae\}$ are the umbilical regions, and we take $\rhom = 0.01, \deltap = 0.05, \deltae = 0.2$. Figure~\ref{fig:curvature_regions} shows a few examples of color coded curvature regions. 

\begin{figure}[b]
    \centering
    \includegraphics[width=.9\linewidth]{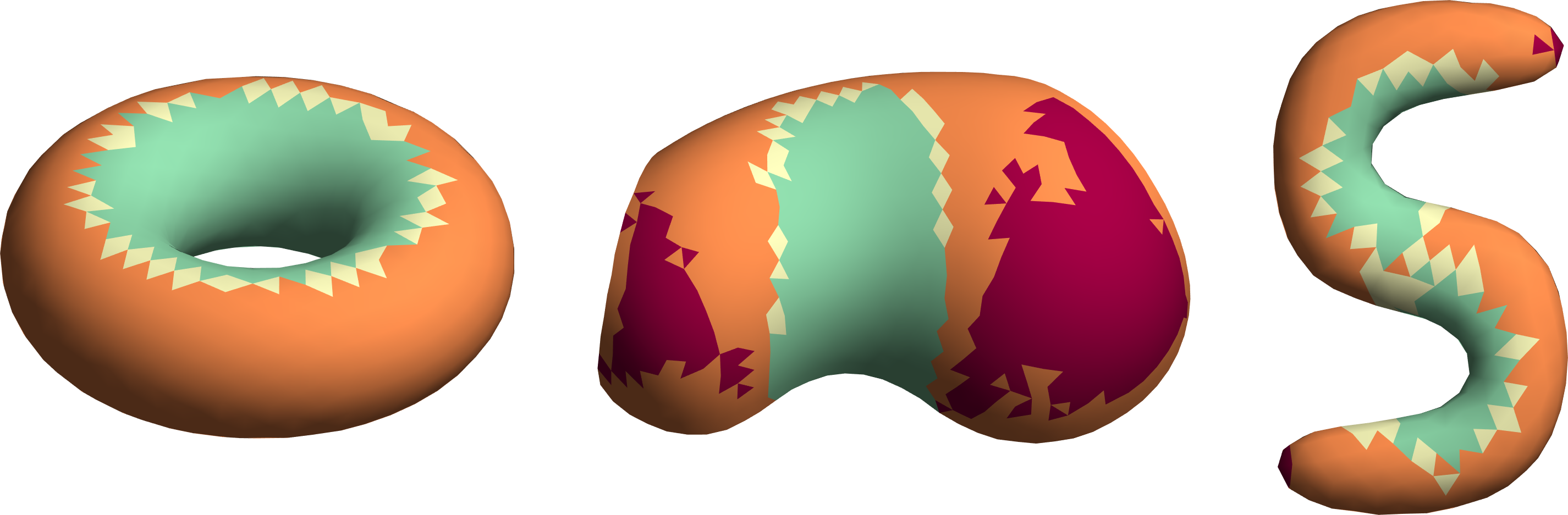}
    \caption{Elliptic (orange, $\mFe$), parabolic (yellow, $\mFp$)  and hyperbolic (green, $\mFh$) curvature regions. Red regions are umbilical.}
    \label{fig:curvature_regions}
\end{figure}

\subsubsection{Guiding field.} 
We define an optimization problem for a $2$-RoSy \emph{guiding field} $\gf$, that is aligned with either $\dmin$ or $\dmax$ according to the curvature regions. Let $\mFs = \mFp, \mFq = \mFe \cup \mFh$. Then, for all faces in $\mFs$, the guiding field should be aligned with $\pm\damin$, where as for faces in $\mFq$ it should be aligned either with $\pm\damin$ or $\pm\s{i}\damin$. Hence, we obtain the following optimization problem.

\begin{mini}|l|
{\gf\tin\xR^{2\nof}}{\gf^T \Delta_2 \gf}
{}{}
\addConstraint{\gf_i - \damin_i^2}{=0,}{\quad t_i \tin \mFs}
\addConstraint{\gf_i^2 - \damin_i^4}{=0,}{\quad t_i \tin \mFq},
\label{eq:opt_guiding_field}
\end{mini}
where $\Delta_2 \tin \xR^{2\nof\times2\nof}$ denotes the $2$-RoSy Laplacian~\cite{knoppel2013globally}, expressed in the local face basis.

This problem has a convex quadratic objective and linear and quadratic constraints. We reformulate it as an unconstrained non-linear least squares problem, as follows.
\begin{mini}|l|
{\gf\tin\xR^{2\nof}}{\frac{1}{\lambda}\gf^T \Delta_2 \gf + \beta(E^{\Psi_2} + E^{\Psi_4}),}
{}{}
\label{eq:opt_guiding_field_relaxed}
\end{mini}
where $E^{\Psi_2}, E^{\Psi_4}$ are the average $L2$ errors of the first and second constraints, respectively, $\lambda$ is the first non-zero eigenvalue of $\Delta_2$, and we take the penalty weight $\beta=10^5$. 

The problem~\eqref{eq:opt_guiding_field} without the second constraint is convex, and we take its solution as the initial solution for the problem~\eqref{eq:opt_guiding_field_relaxed}, for which we use the Levenberg-Marquardt algorithm. Figure~\ref{fig:guiding_field} (left) shows an example of $2$- and $4$-RoSy constraints and the resulting guiding field $\gf$ (middle).

\begin{figure}[b]
    \centering
   \includegraphics[width=.9\linewidth]{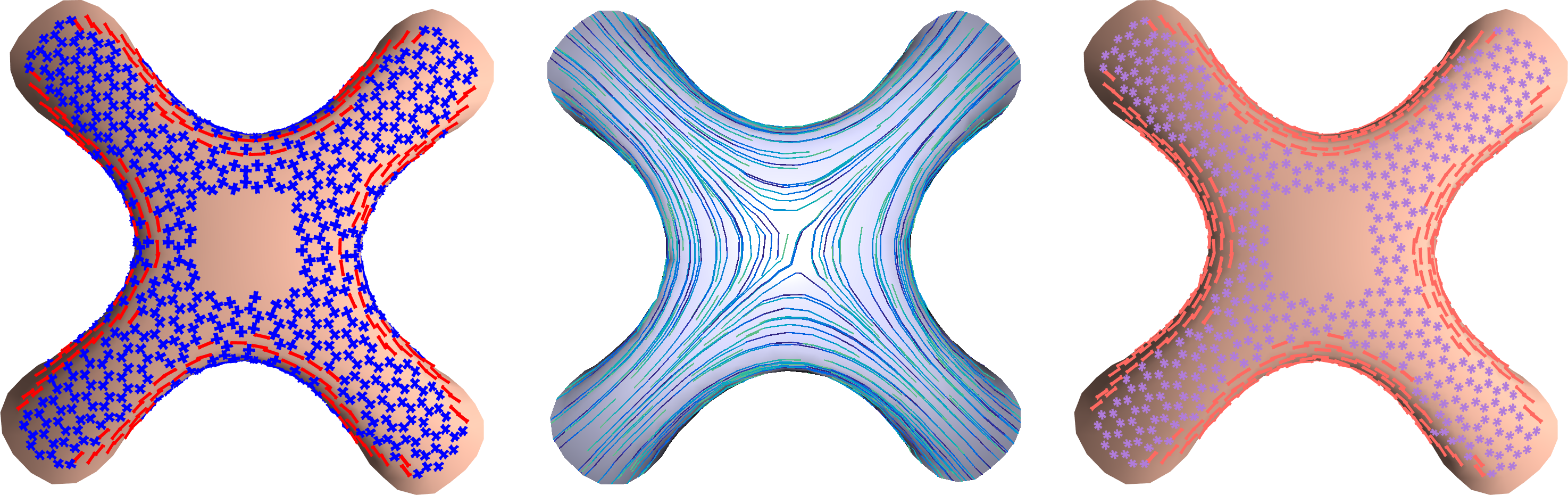}
    \caption{(left) Constraints from curvature regions, (middle) resulting $2$-RoSy guiding field $\gf$, (right) alignment constraints for CPF optimization.}
    \label{fig:guiding_field}
\end{figure}

\subsubsection{CPF Alignment constraints.} 
\label{sec:cpf_alignment_constraints}
Finally, we use the guiding field to setup the alignment constraints for the CPF framework. As discussed, in elliptic regions there are three developable strips, thus the guiding field can be aligned with any of the pushed forward grid directions. On the other hand, in parabolic and hyperbolic regions there is a single strip, which should be aligned with the pushed forward $\pm\hat{\s{u}}$ to generate a planarizable dual hex, as in Figure~\ref{fig:strips} (right). Hence, we set the constraints:
\begin{equation}
\begin{aligned}
d_i &= \sqrt{\gf_i}, \quad \aln{i}=6, \quad t_i \tin \mFe, \\
d_i &= \sqrt{\gf_i}, \quad \aln{i}=2, \quad t_i \tin \mFp \cup \mFh.
\end{aligned}
\label{eq:cpf_alignment_constraints}
\end{equation}
The CPF alignment constraints are visualized in Figure~\ref{fig:guiding_field} (right). 

These constraints are essential for achieving a planarizable mesh. In Figure~\ref{fig:ablation_alignment} we compare the results of our algorithm with and without alignment constraints,  shown in the left and right part of the Figure, respectively. We show the resulting $U,V$ fields, the dual mesh $\rdM$ and the final PH mesh $\hM$. In both experiments we used conjugacy and scaling constraints. Note that without alignment constraints, the resulting $U$ field is not aligned with the minimal absolute curvature direction, and the planarization leads to deformed hexagons. In this figure, and all the figures that show color coded planarity error, all the faces that have a planarity error larger than the maximal value on the color scale are showed with the color of the maximal value. This is done to better visualize the low planarity regions on meshes which have large maximal planarity error.

\begin{figure}[t]
    \centering
    \includegraphics[width=\linewidth]{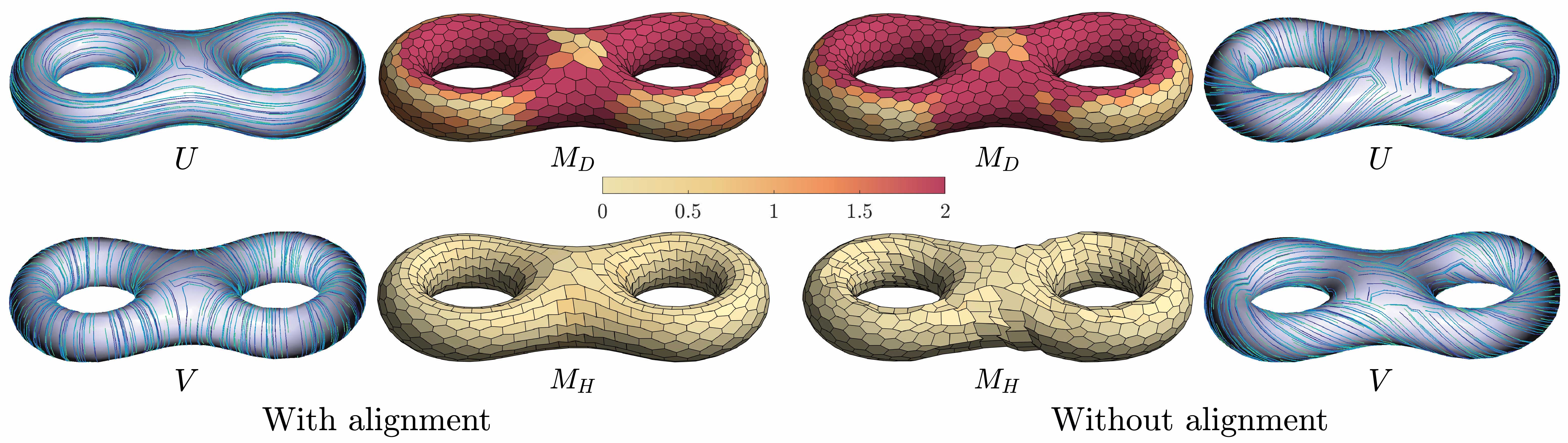}
    \caption{Our results with (left) and without (right) alignment constraints. We show the final $U,V$ fields, and the dual and final meshes. Note that without alignment constraints the $U$ field is not aligned with the minimal absolute curvature direction, leading to a non-planarizable mesh.}
    \label{fig:ablation_alignment}
\end{figure}

\begin{figure}[b]
    \centering
    \includegraphics[width=\linewidth]{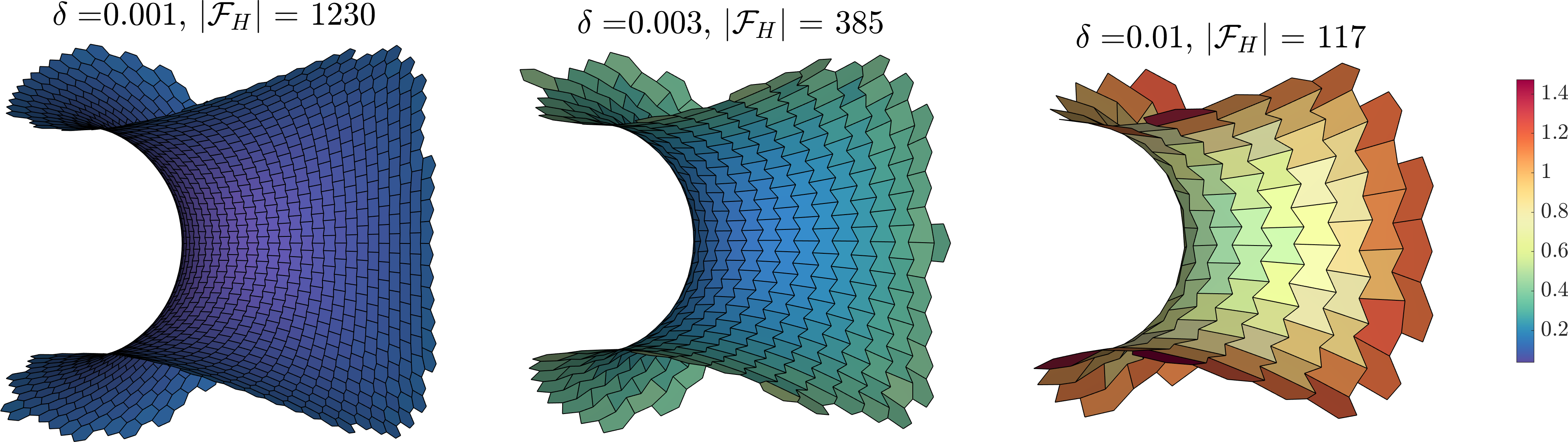}
    \par\vspace{.1in}
    \includegraphics[width=\linewidth]{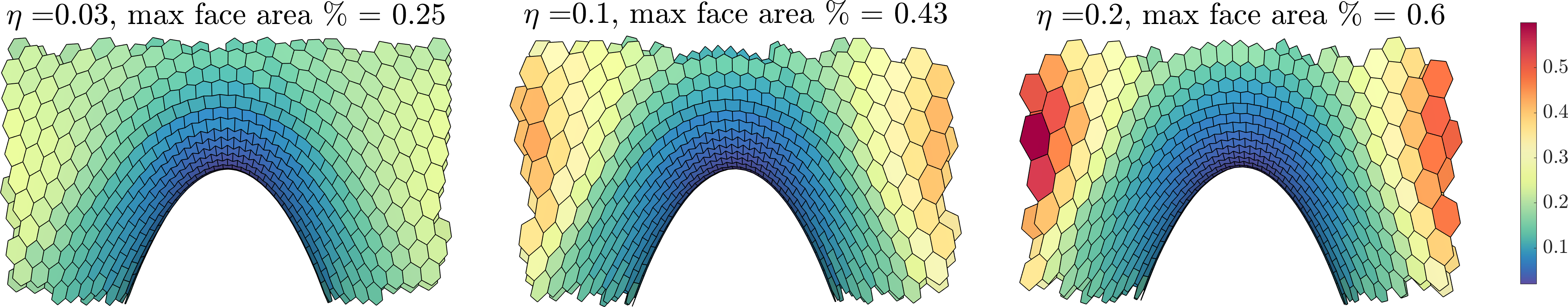}
    \caption{Effect of scaling parameters $\delta,\maxel$. (top) Decreasing $\delta$ while keeping $\maxel$ fixed allows a closer approximation of the surface at non-degenerate curvature regions. (bottom) Increasing $\maxel$ while keeping $\delta$ fixed allows for larger elements in flat regions. We show the color coding of the face area's percentage of the total surface area. }
    \label{fig:scaling_params}
\end{figure}

\begin{figure}[t]
    \centering
    \includegraphics[width=.59\linewidth]{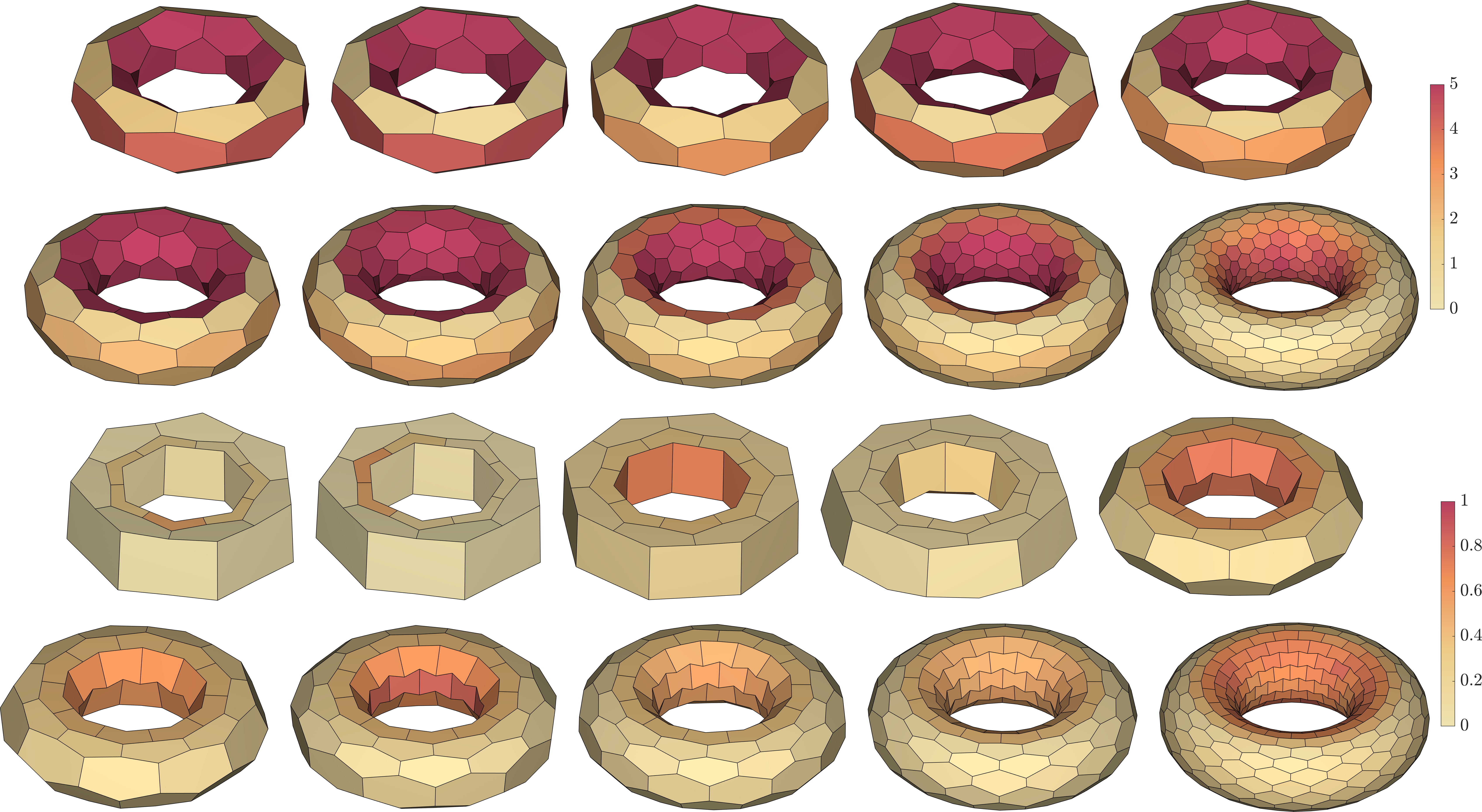}
    \includegraphics[width=.39\linewidth]{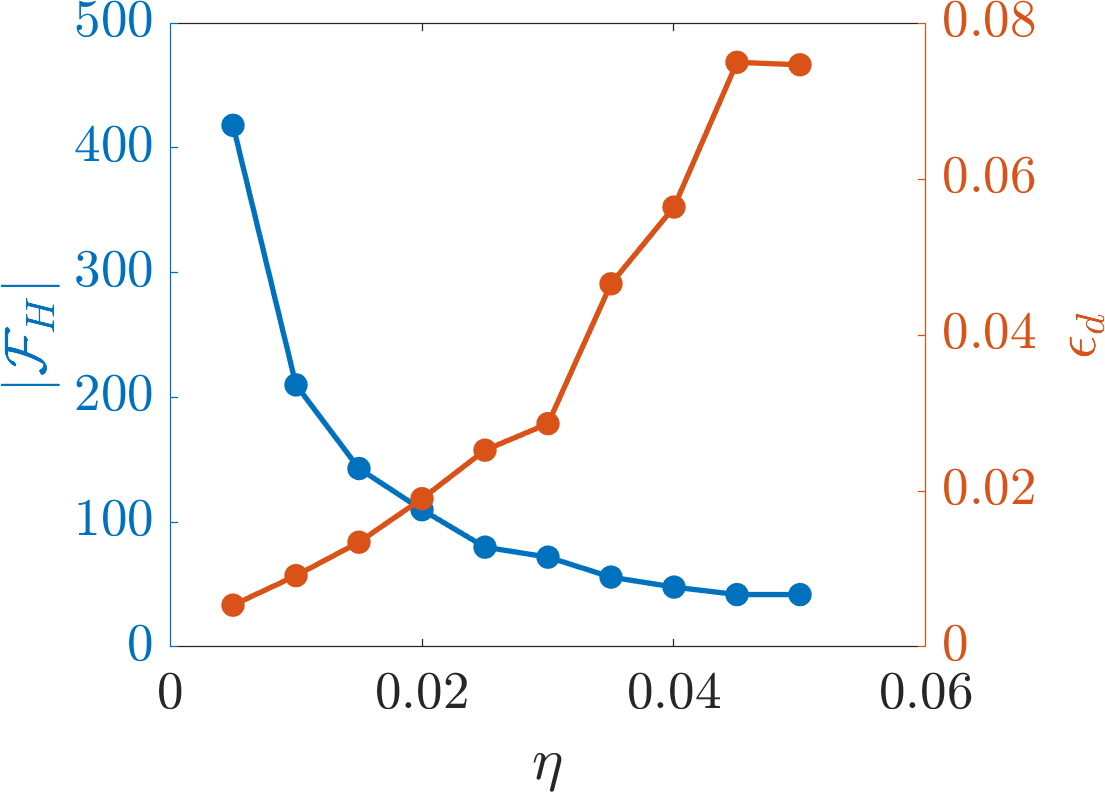}
    \caption{Varying $\maxel,\delta$ together while keeping their percentage ratio fixed allows us to smoothly vary the mesh resolution. (left) The two top rows show the dual mesh before planarization, and the two bottom rows the final PH mesh, all color coded by the planarity value (with maximal values $5$ and $1$ resp.). (right) The trade-off between the number of elements $|\hF|$ and the Hausdorff distance $\hdist$ as a function of $\maxel$. As expected, lower values of $\maxel$ lead to a smaller Hausdorff distance at the expense of a larger number of faces.}
    \label{fig:resolution}
\end{figure}

\subsection{Approximation by Scaling} 
\label{sec:ph_scaling}
\subsubsection{Regularized Dupin Indicatrix.}
Our conjugacy and alignment constraints set the \emph{orientation} of the pushed forward hexagon. For the resulting hexagons to be circumscribed in a uniform scale of the Dupin indicatrix, as required, we additionally constrain the \emph{scaling} of a pushed forward vector $w \tin \tpm$. In general, in an elliptic region, $w$ should fulfill $\|w\|^2|\kn(w)| = \maxel, \forall w \, \text{ s.t.} \, \kn(w) \neq 0$, where $\maxel > 0$ controls the area of the element. 

There are two issues to note. First, the vector size $\|w\|$ becomes unbounded as $\kn \to 0$. Second, our hexagons are regular hexagons which are pushed forward through a parameterization, and thus are \emph{always} convex. Thus, $w$ should always trace an ellipse, even in hyperbolic regions, and it should have the aspect ratio of the Dupin indicatrix. We address both issues with the following definition.

\begin{definition}
Let $p\tin\sM$, with curvature directions $\dmin, \dmax$, and consider the local curvature coordinates in $\tpm$ where $p$ is the origin, and $\dmin, \dmax$ are the $x,y$ axes, respectively. Given $\epsilon > 0$, the \emph{$\epsilon$-regularized Dupin indicatrix} $D_\epsilon(p)$ is the ellipse $x^2 \skmin + y^2 \skmax = 1$, where $\skmin = \sqrt{\kmin^2 + \epsilon^2}, \quad \skmax = \sqrt{\kmax^2 + \epsilon^2}$. The \emph{$\epsilon$-regularized normal curvature} in direction $w \tin\tpm$ is given by $\skn(w) = \skmin \cos^2\theta + \skmax \sin^2\theta$, where $\theta$ is the angle of $w$ w.r.t $\dmin$.
\label{def:reg_dupin}
\end{definition}

The function $x \to \sqrt{x^2 + \epsilon^2}$ is a well known smoothing of the function $|x|$ that is known to be convex, differentiable, positive, and within an $\epsilon$-dependent bound from $|x|$~\cite{argaez2011l}. Geometrically, we replace $|x|$ with the hypotenuse of a right angle triangle with leg lengths $|x|$ and $\epsilon \!>\! 0$, and thus we are guaranteed that the result is positive, and has the desired approximation effect. Furthermore, the regularized Dupin indicatrix is always well defined, and provides a smooth, non-degenerate transition between elliptic and hyperbolic regions. We further have the following.

\begin{lemma}
Let $p \tin \sM$, with curvatures $\kmin, \kmax$, and let $D_\epsilon(p)$ be its regularized Dupin indicatrix. Then, the area $\ar_\epsilon(p)$ of $D_\epsilon(p)$ fulfills:
\begin{equation}
\begin{aligned}
\lim_{\substack{|\kmin| \gg 0 \\ |\kmax| \gg 0}} \ar_\epsilon(p) &= \frac{\pi}{\sqrt{|\kmin|} \sqrt{|\kmax|}}, \quad &\lim_{\substack{\kmin \to 0 \\ \kmax \to 0}} \ar_\epsilon(p) &= \frac{\pi}{\epsilon}\\
\lim_{\substack{\kmin \to 0 \\ |\kmax| \gg 0}} \ar_\epsilon(p) &=  \frac{\pi}{\sqrt{\epsilon}\sqrt{|\kmax|}},  \quad &\lim_{\substack{\kmax \to 0 \\ |\kmin| \gg 0}} \ar_\epsilon(p) &=  \frac{\pi}{\sqrt{\epsilon}\sqrt{|\kmin|}}.
\end{aligned}
\label{eq:reg_dupin_area}
\end{equation}
\label{lemma:reg_dupin_area}
\end{lemma}
The proof follows directly from the fact $\lim_{x\to0} \sqrt{x^2 + \epsilon^2} = \epsilon$.

Hence, by choosing $\epsilon, \gamma$, and requiring $\|w\|^2|\skn(w)| = \gamma^2, \forall w=dr(\s{w})\tin \tpm$, we constrain the pushforward unit circle to be a \emph{uniform $\gamma$-scale of the $\epsilon$-regularized Dupin indicatrix}. This allows us to control the areas of the pushed forward hexagons, and provide a smooth transition between elliptic, parabolic and hyperbolic regions. 

\begin{figure}[t]
    \centering
      \includegraphics[width=.9\linewidth]{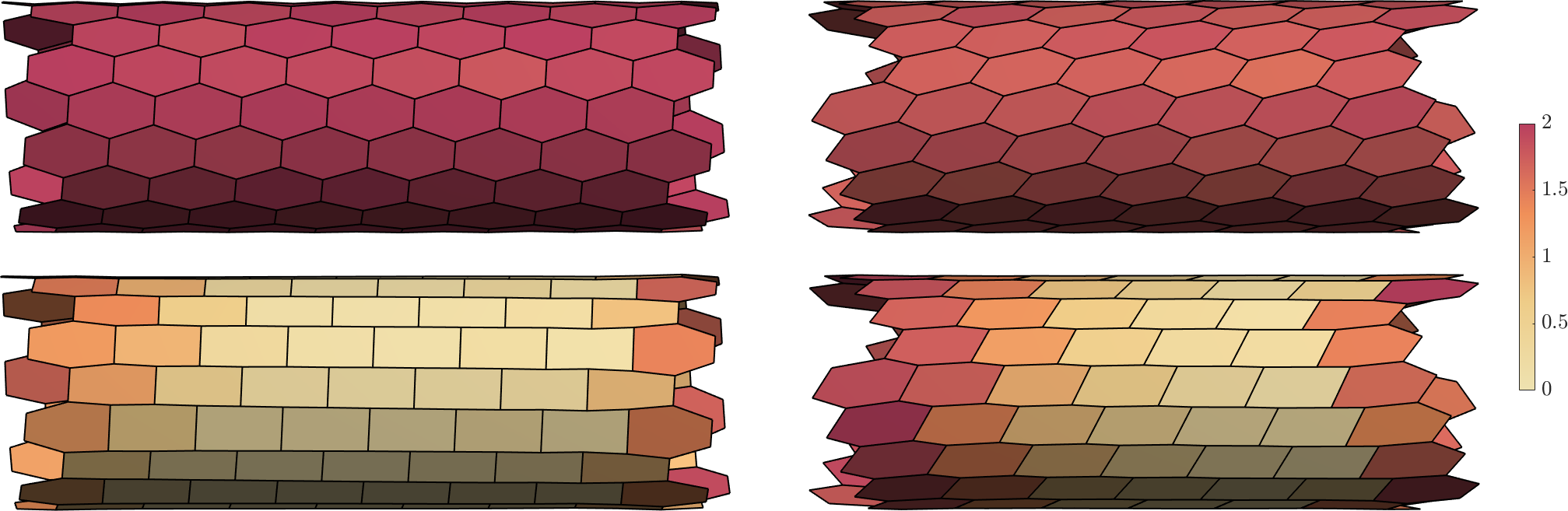}
    \caption{The dual mesh (left) with and (right) without orthogonality constraints, (top) before and (bottom) after planarization. Without orthogonality hexes are skewed, but remain planar.}
    \label{fig:ablation_fairness_ortho}
\end{figure}

\begin{figure}[b]
    \centering
   \includegraphics[width=\linewidth]{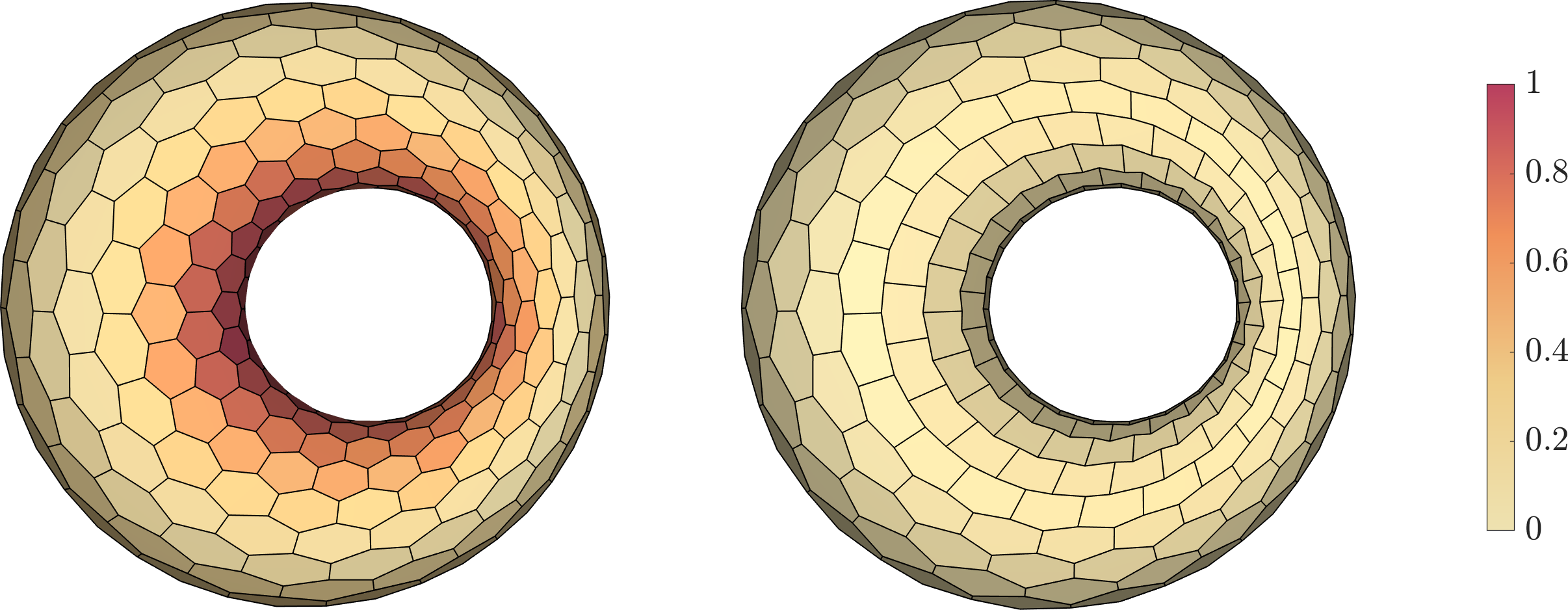}
    \caption{The dual mesh (left) before and (right) after planarization. Hexagonal faces in elliptic regions are nearly planar even before planarization, since the hexagons are inscribed in a scaled Dupin indicatrix. }
    \label{fig:elliptic_is_planar}
\end{figure}

\begin{figure}[t]
    \centering
    \includegraphics[width=\linewidth]{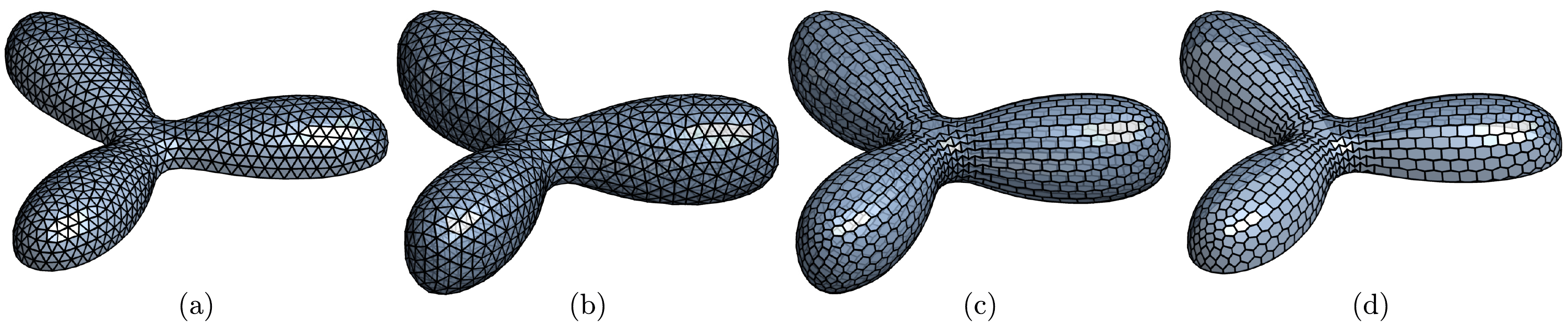}
    \caption{Preserving the input's (a) single planar boundary curve by doubling the mesh (b), PHexing it (c), and intersecting the result with the boundary plane (d).}
    \label{fig:bdry}
\end{figure}

\begin{figure}[b]
    \centering
    \includegraphics[width=\linewidth]{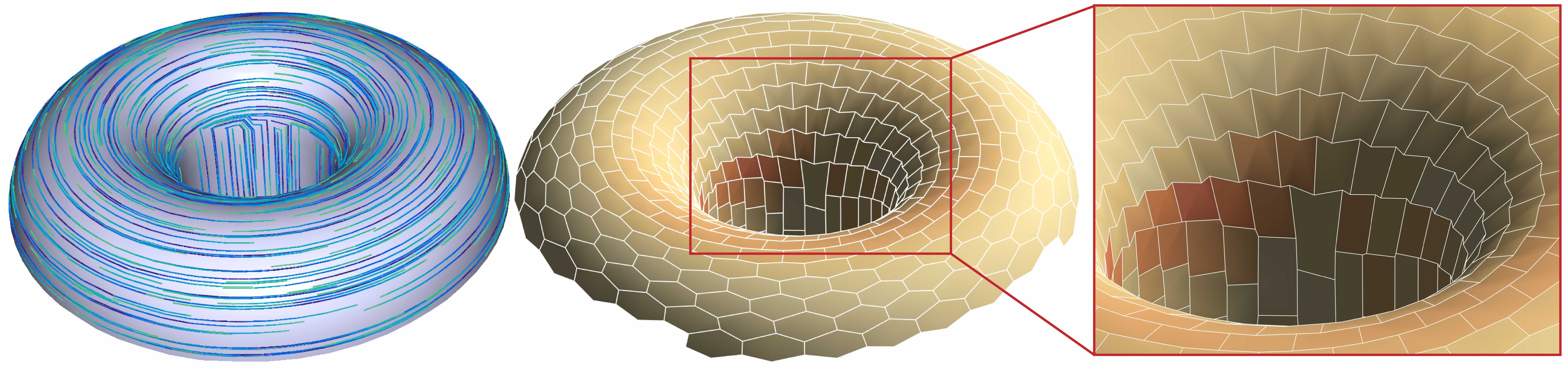}
    \caption{\REV{
Failure case: neighboring parabolic regions with conflicting maximal absolute curvature directions lead to a non-smooth guiding field (left), which leads to multiple singularities and badly shaped hexagons (right).}}
    \label{fig:conflicting_cyl}
\end{figure}

\begin{figure*}
    \centering
    \includegraphics[width=\linewidth]{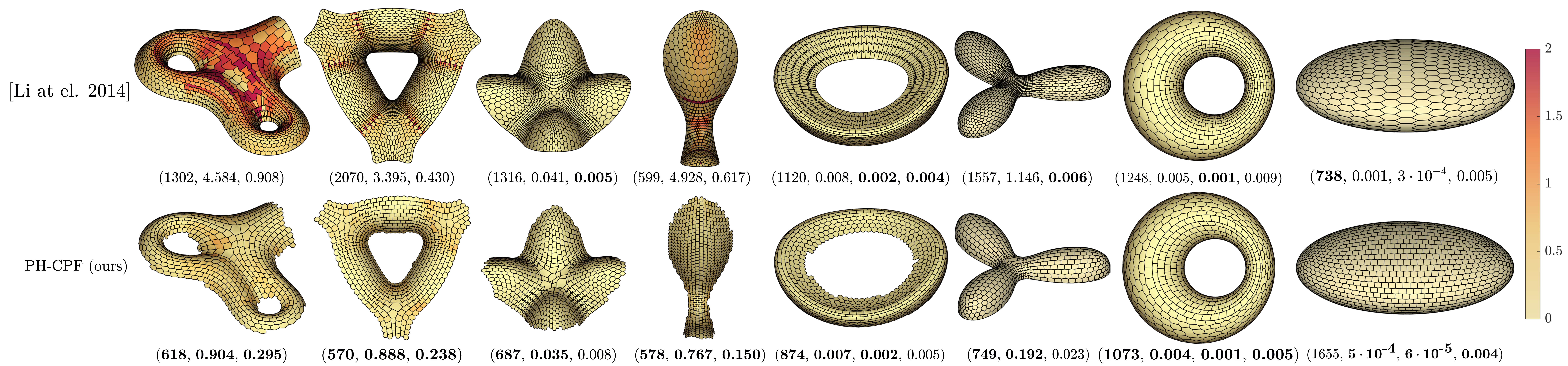}
    \caption{Comparison of our results with~\citet{li2014planar}, color coded with the face planarity values. Below each mesh, we show $(|\hF|,\planerrM,\planerravg,\hdist)$ where $\hdist$ is displayed only if Li's input mesh was provided. Note that we obtain better or similar planarity with a smaller number of faces.}
    \label{fig:wang_comp}
\end{figure*}

\subsubsection{CPF scaling constraints.} 
\label{sec:cpf_scaling_constraints}
We set up the scaling constraints:
\begin{equation}
\begin{aligned}
\|U_i\|_{g_i(\epsilon,\gamma)}^2 &= 1, \quad \|V_i\|_{g_i(\epsilon,\gamma)}^2 = 1, \\
g_i(\epsilon,\gamma) &= \frac{1}{\gamma^2}\dmindmax \skminskmax \dmindmax^T = \frac{1}{\gamma^2}\sshapeop.
\end{aligned}
\label{eq:cpf_scaling_constraints}
\end{equation}
We discuss the effect of the combined constraints in different regions in Section~\ref{sec:constraint_analysis}.

\subsubsection{Choosing $\epsilon$ and $\gamma$.}
The choice of the scaling parameters $\epsilon, \gamma$ affects the aspect ratio of the hexagonal elements in parabolic regions, as well as the maximal element size in planar regions. 
The area of the $\gamma$-scaled $\epsilon$-regularized Dupin indicatrix is $\gamma^2 \alpha_\epsilon(p)$, which is equal to $\gamma^2 \pi/ \epsilon$ in planar regions (Equation~\eqref{eq:reg_dupin_area}).

Let $\delta$ be the desired elements' distance to the surface, and $\maxel$ the maximal element size.
Thus, we have $\pi\gamma^2/\epsilon = \maxel$, where $\gamma^2 = 2\delta + O(\delta^2)$ (see Figure~\ref{fig:dupin_plane_and_hex_notations} (left)). Hence, we take $\epsilon_{\delta,\maxel} = \frac{2\pi\delta}{\maxel}, \gamma_{\delta,\maxel} = \sqrt{2\delta}$.

Figure~\ref{fig:scaling_params} shows the effects of choosing different $\delta, \maxel$. Varying $\delta$ while keeping $\maxel$ fixed (top) allows a closer approximation of the surface at non-degenerate locations. Varying $\maxel$ while keeping $\delta$ fixed (bottom) allows for larger elements in flat regions. The values of $\delta,\maxel$ are taken as a percentage of the bounding box diagonal, and total surface area, respectively. The figure shows the color coding of the face areas percentage of the total area of the mesh. 

Figure~\ref{fig:resolution} shows how we can control the mesh resolution by scaling both $\maxel$ and $\delta$ while keeping their percentage ratio fixed. The values of $\maxel$ vary linearly from $5\%$ (top left) to $.5\%$ (bottom right) of the total surface area, and $\delta$ varies from $10\%$ to $1\%$ of the bounding box diagonal. The color scale is uniform for all images, and varies from $0$ to $2$. The two top rows show the dual mesh before planarization, and the two bottom rows show the final PH mesh. Note that here the initial planarity error is quite high, yet we still obtain a coarse planar hex mesh after planarization. In general, for simple surfaces with low curvature variation, we take $\delta$ to be $0.5\%$ of the bounding-box diagonal of the input surface, and $\maxel$ to be $0.25\%$ of the total surface area. For more challenging general meshes, we take the same $\delta$, but reduce $\maxel$ to be $0.1\%$, to generate smaller elements.

\begin{figure}[b]
    \centering
    \def\svgwidth{\linewidth}
    \includegraphics[width=\linewidth]{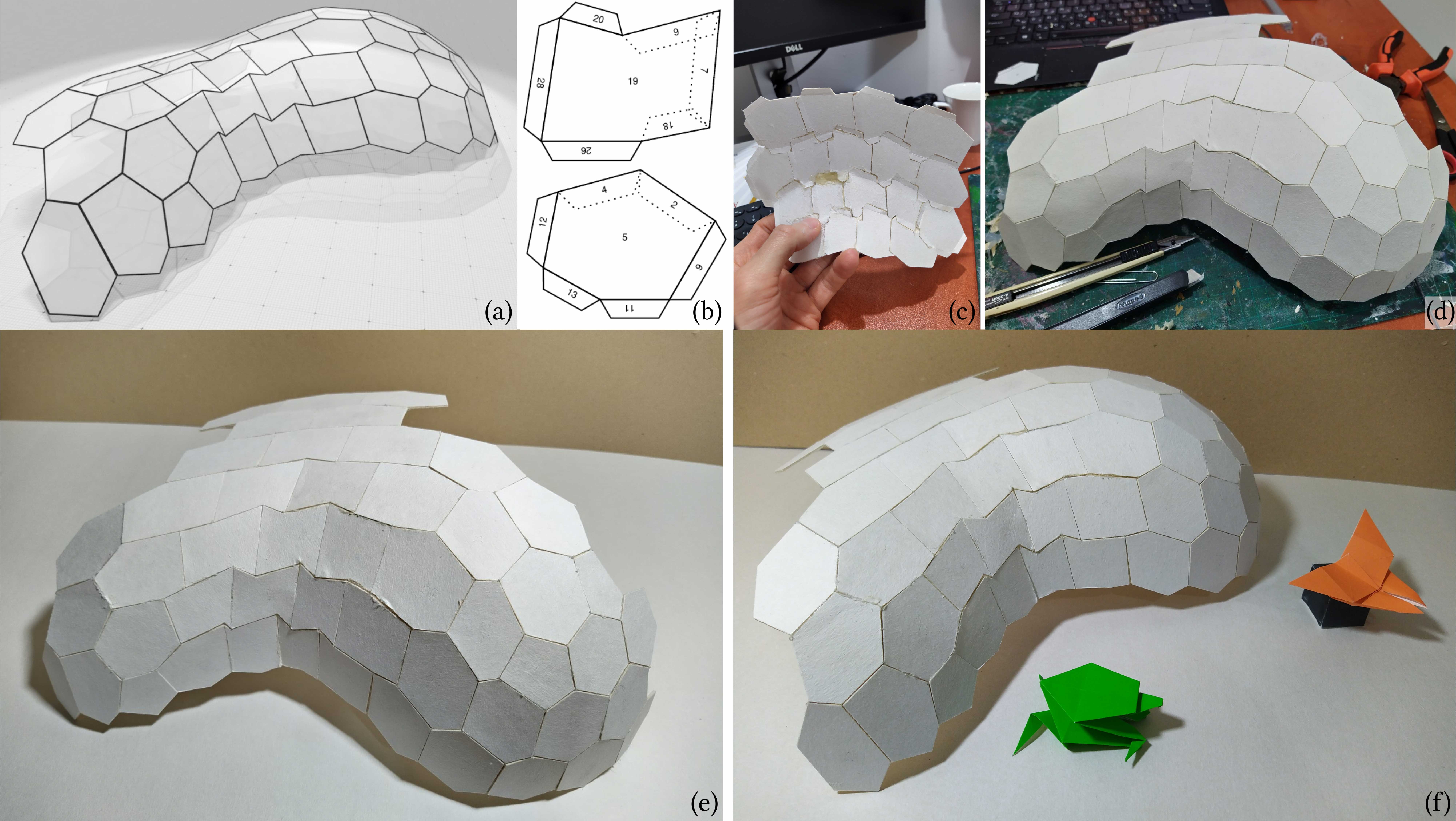}
    \caption{Fabrication of the Laga model from construction paper. (a) PH mesh, (b) 2D face
    templates for cutting, (c-d) intermediate and (e-f) final constructions.}
    \label{fig:fabrication}
\end{figure}

\subsection{Fairness by Orthogonality}
\label{sec:cpf_orthogonality_constraints}
Our last objective is \emph{fairness}, formalized as horizontal and vertical symmetry of the hexagonal elements. We achieve this by requiring the orthogonality of $U_i,V_i$, namely by adding the CPF constraints:
\begin{equation}
\langle U_i, V_i \rangle = 0, \quad \forall t_i \tin \mF.
\label{eq:cpf_orthogonality_constraints}
\end{equation}

In parabolic and hyperbolic regions the planarity optimization modifies the shapes of the pushed-forward hexes away from convexity. If the symmetry axes of the pushed forward hexagon are not aligned with the curvature directions, this modification can lead to skewed hexes, for example parallelograms instead of rectangular bricks in parabolic regions.  
By enforcing orthogonality as well as conjugacy, we effectively require that $U_i,V_i$ are aligned with the curvature directions, and thus that the planarization preserves the horizontal and vertical reflectional symmetry of the hexes. Note that, unlike for PQ meshes, these combined constraints do not replace alignment, since for PH meshing the $U$ direction must be aligned with the minimal absolute curvature direction, whereas for PQ meshes $U,V$ can be aligned with either minimal or maximal directions. 

Figure~\ref{fig:ablation_fairness_ortho} demonstrates the effect of adding orthogonality constraints. On the left we show the result with orthogonality, and on the right without it. We show the dual mesh before planarization $\rdM$ (top) and the final PH mesh $\hM$ (bottom). Note that without orthogonality the hexes of $\hM$ are still planar, but lose horizontal and vertical symmetry with respect to the ruling direction.

\subsection{Constraint analysis}
\label{sec:constraint_analysis}
\subsubsection{Elliptic regions}
When $\kmin\kmax \!\gg\! 0$ we have $\sshapeop \!\approx\! \shapeop$, and thus the Dupin Equation~\eqref{eq:conjugacy_uu} holds approximately. 
Furthermore, in these regions, for any $\s{w} \tin \xR^2, \, |\s{w}|=1$ we have that $\|dr(\s{w})\|^2_{g_i(\epsilon,\gamma)} = \frac{1}{\gamma^2}\s{w}^T \uv{i}^T \sshapeop \uv{i} \s{w} \approx \s{w}^T \id \s{w} = 1$. Hence, in elliptic regions the scaling and conjugacy constraints imply that pushing forward the unit circle yields (approximately) a uniform $\gamma$-scale of the $\epsilon$-regularized Dupin indicatrix. Indeed, in elliptic regions (when $\gamma$ is small enough) the initial planarity error is considerably lower compared to the rest of the surface (see Figure~\ref{fig:elliptic_is_planar}).

\begin{figure}[b]
    \centering
    \includegraphics[width=.9\linewidth]{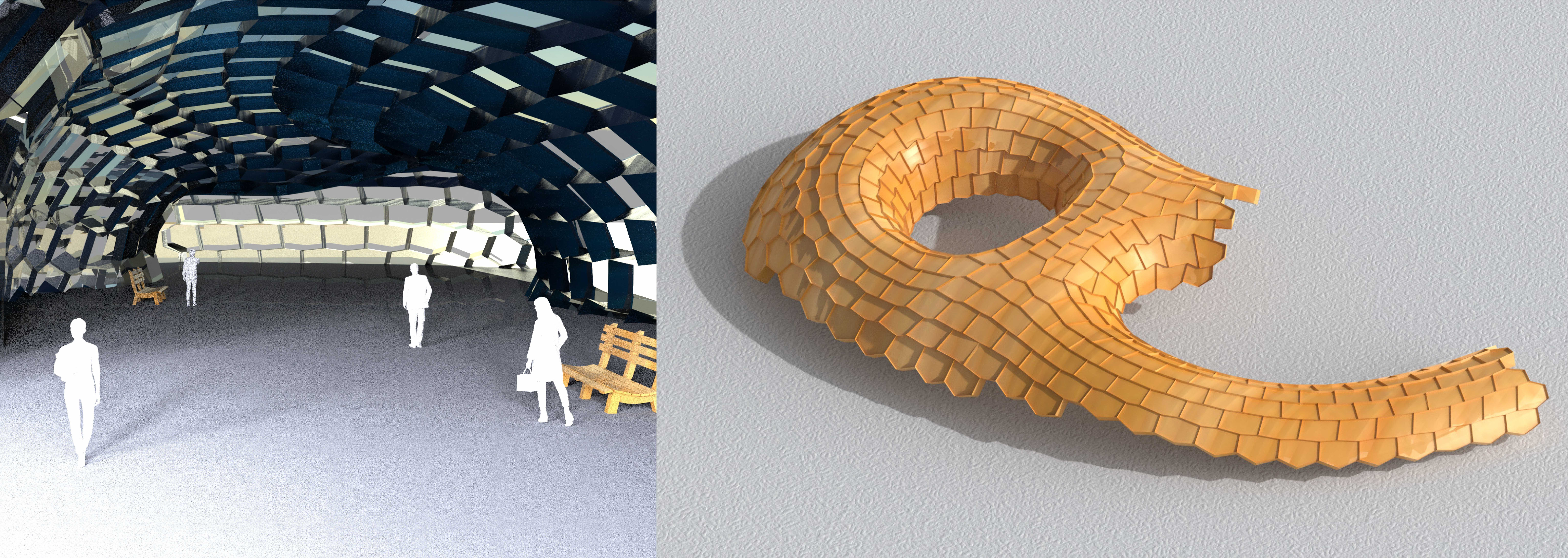}
    \caption{Face offset meshes computed from our PH meshes. Note that these are computable only for conical meshes.}
    \label{fig:face_offset}
\end{figure}

\begin{figure*}
    \centering
    \includegraphics[width=.9\linewidth]{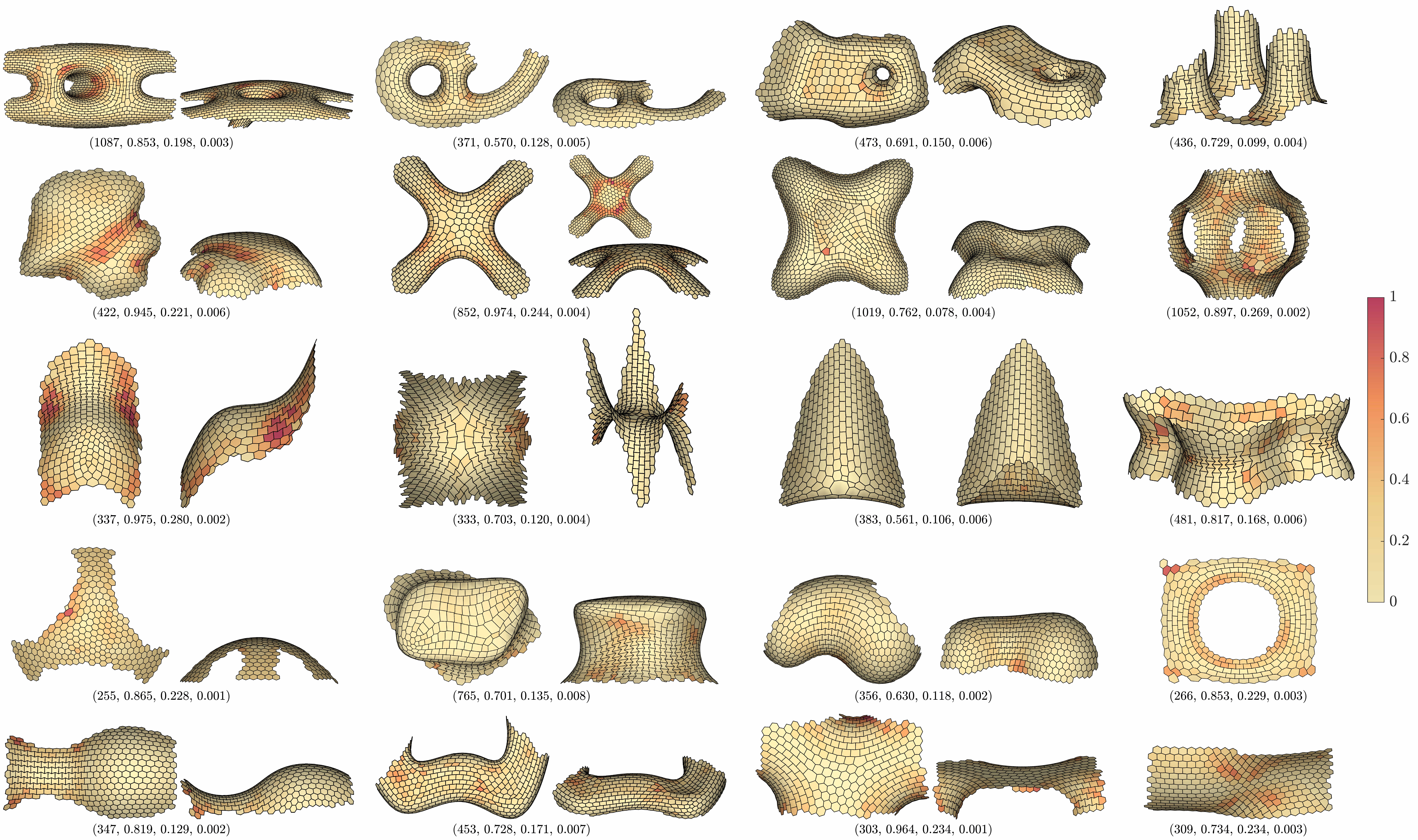}
    \caption{Gallery of results for architectural models. We show the PH mesh, color coded with planarity values, and below $(|\hF|,\planerrM,\planerravg,\hdist)$.}
    \label{fig:gallery_architectural}
\end{figure*}

\subsubsection{Parabolic regions}
When either $\kmin\to0$ or $\kmax\to0$, we have $\langle \damin, w \rangle_\shapeop \to 0$ for any $w \tin \tpm$. Thus, if Eq.~\eqref{eq:cpf_alignment_constraints} holds and $U$ is aligned with $\damin$, then the conjugacy Eq.~\eqref{eq:conjugacy_uv} holds for any tangent vector $w$. The dimensions of the resulting element are restricted in the ruling direction by the scaling parameters $\epsilon,\gamma$, and in the orthogonal direction by the non-zero curvature, as given by Eq.~\eqref{eq:reg_dupin_area}.

\subsubsection{Hyperbolic regions}
If the alignment and conjugacy equations~\eqref{eq:cpf_alignment_constraints},~\eqref{eq:conjugacy_uv} hold, then $U,V$ are aligned with the curvature directions. 
In this case, it is easy to check that $\langle U_i, V_i \rangle_{\sshapeop} = 0$, and if additionally the scaling equation~\eqref{eq:cpf_scaling_constraints} holds we have again that  $\|dr(\s{w})\|^2_{g_i(\epsilon,\gamma)} = 1$ for any unit vector $\s{w} \tin \xR^2$, and thus pushing forward the unit circle yields a uniform $\gamma$-scale of the $\epsilon$-regularized Dupin indicatrix.

\subsubsection{Planar regions}
When both $\kmin\to0, \kmax\to0$, we have that $g_i(\epsilon,\gamma) \to \frac{\epsilon}{\gamma^2}Id$, thus the sizing is uniform, and determined by the scaling parameters. If we additionally require orthogonality, this amounts to requiring conformality in planar regions.

\subsection{Initialization}
We initialize the variables $U,V$ as follows $\forall t_i \tin \mF$:
\begin{equation}
\begin{aligned}
\tilde{U}_i^0 &= \sqrt{\gf_i},  & \tilde{V}_i^0 &= -J\tilde{U}_i^0,   \\
U_i^0 &= \tilde{U}_i^0 / \|\tilde{U}_i^0\|_{g_i(\epsilon,\gamma)}, & V_i^0 &= \tilde{V}_i^0 / \|\tilde{V}_i^0\|_{g_i(\epsilon,\gamma)}, \\
\end{aligned}
\label{eq:cpf_init_uv}
\end{equation}
and $z$ as $\forall e_{ij} \tin \mEi$:
\begin{equation}
z_{ij}^0 = \left(\big({(dr_i^0)}^{-1}(e_{ij})\big)^N +  \big({(dr_j^0)}^{-1}(e_{ij})\big)^N\right)/2 .
\label{eq:cpf_init_z}
\end{equation}
Thus, the initial solution fulfills the conjugacy, scaling, orthogonality and alignment constraints, and is of course LICO as required. The optimization improves the continuity objective while maintaining smoothness and the other constraints as best as possible.

\section{Implementation Details}

\REV{We provide the details on mesh generation and planarization in the Appendices.}

\subsection{Constraints computation}
\label{sec:constraints_computation}
\subsubsection{Shape operator}
For the shape operator we take the tangential gradient of the normals at the vertices, symmetrized and projected to the tangent plane~\cite{azencot2015functional}. If normals at the vertices are available (e.g. if the surface was sampled from a parametric surface) we use those, otherwise we take the area-weighted average of the normals of the neighboring faces.
For the large majority of meshes, the normals are computed.

\subsubsection{Boundary faces}
Boundary faces are defined as faces with at least one vertex on the boundary. On these faces, the computation of vertex normals  (and the shape operator) is not reliable, and therefore we do not constrain these faces for alignment, nor conjugacy.
\subsubsection{Guiding field, special cases}
If there are no constraints, the guiding field is given by the lowest eigenvector of the $2$-RoSy Laplacian. If there are only $2$-RoSy constraints, the optimization problem~\eqref{eq:opt_guiding_field} is convex quadratic and is solved directly using constrained linear least squares. If there are only $4$-RoSy constraints, then the curvature is uniform, either hyperbolic or elliptic, and we take the minimal curvature direction $\dmin$ in constrained regions as the $2$-RoSy constraints for the guiding field optimization.
\subsubsection{Alignment weighting}
As is common in meshing, we add diagonal weights to the alignment constraint, to reduce the weight in nearly planar regions. Specifically, for $N=4$ we weigh the alignment constraint on the face $t_i \tin \mF$ with $\rho_i^2$, and for $N=6$ with $\big(\rho_i^{3}-3\rho_i^{2}+3\rho_i\big)^{2}$, a polynomial which resembles a Sigmoid function and has vanishing first and second derivatives at $1$.

\begin{figure*}
    \centering
    \includegraphics[width=.9\linewidth]{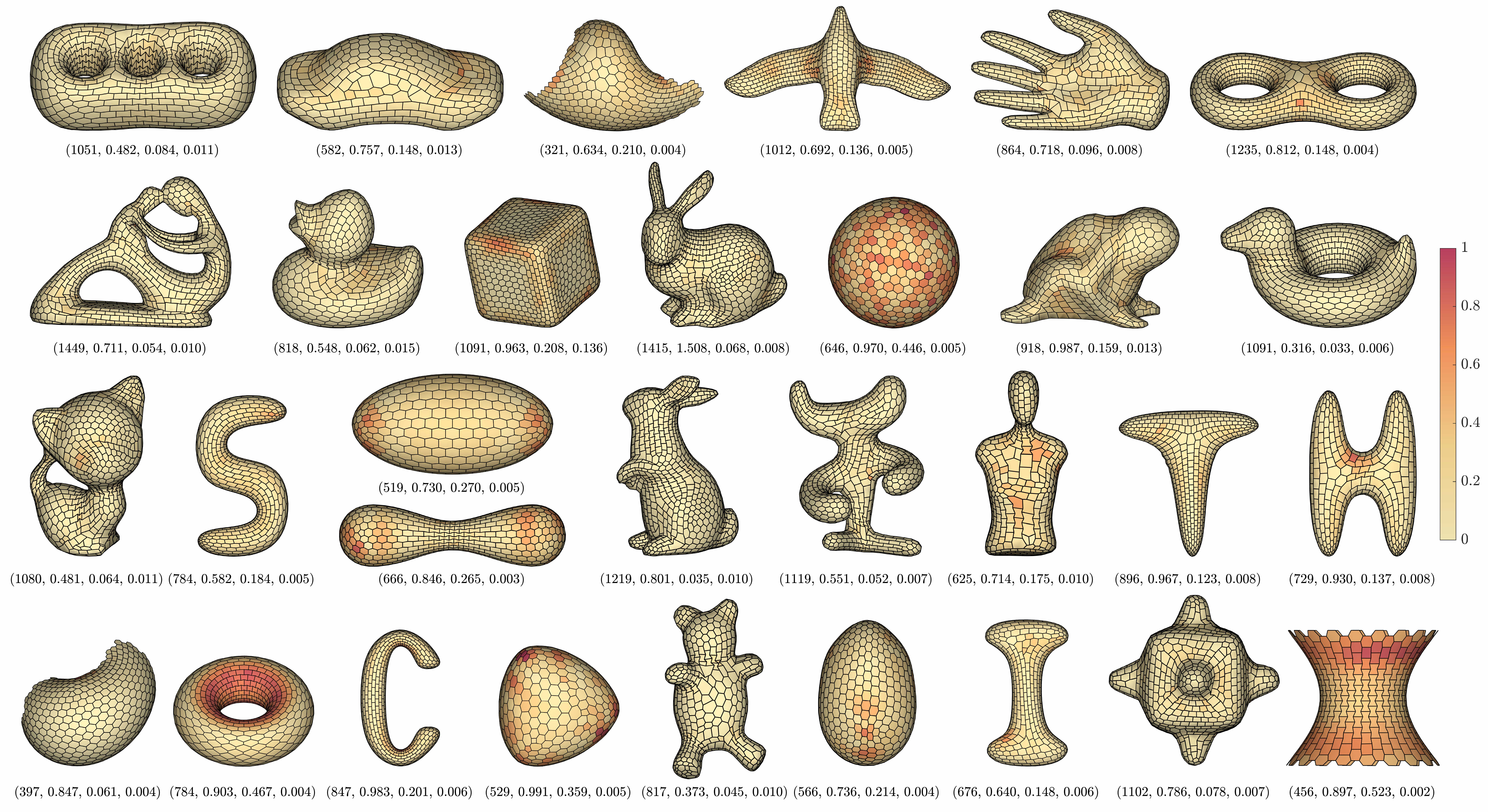}
    \caption{Gallery of results for general models. We show the PH mesh, color coded with planarity values, and below $(|\hF|,\planerrM,\planerravg,\hdist)$.} 
    \label{fig:gallery_general}
\end{figure*}

\subsection{Boundaries}
\label{sec:bdry}
\REV{
Unlike quadrangular meshes, hexagonal elements cannot naturally conform to the surface boundaries, as in convex regions the hexagons will form ``jagged'' boundaries. 
Thus, a boundary-conforming mesh will have many non-hexagonal faces. Furthermore, if the boundaries are parallel to the maximal absolute curvature direction, arbitrarily shaped low quality faces may be obtained when the curvature-aligned hexagons intersect the non-aligned boundary. Thus, many recent works on hexagonal remeshing~\cite{Jiang2015polyhedral,wang2009note,muller2011conformal,wang2008hexagonal} have chosen to generate meshes whose boundaries do not conform to the input shape. Following the large body of existing work, we have chosen a similar path.
Note, however, that methods that preserve boundaries, such as~\citet{li2014planar}, are advantageous in many scenarios, especially in the architectural context.

When there is a single planar boundary curve, one possible solution is \emph{doubling} the input surface by reflecting it with respect to the boundary plane.
This leads to a closed surface, which is then be intersected with the boundary plane after hex-meshing. See Figure~\ref{fig:bdry} for such an example.
Finding a robust, general solution for handling the boundary is an interesting avenue for future work.
}

\subsection{Limitations}
\REV{
Our method relies on the existence of a smooth $2$-RoSy guiding field, which is aligned to the maximal absolute curvature direction in parabolic regions.
In some cases, when there are adjacent parabolic regions with a maximal absolute curvature direction which changes abruptly (see Fig.~\ref{fig:conflicting_cyl} (a)), there exists no such smooth field. In these cases, multiple singularities are introduced along the discontinuity line, leading to badly shaped hexagons (Fig.~\ref{fig:conflicting_cyl} (b)).
This is a limitation that can potentially be addressed by allowing both quads and hexagons, a direction we leave for future work.
}

We solve a non-linear non-convex optimization problem, and the solution that we obtain depends on the initialization. While our algorithm is quite stable, as we have demonstrated using convergence from a random initialization, we cannot guarantee that the global optimum would be reached. 
Furthermore, we rely heavily on curvature information, for all our constraints except orthogonality. While we have used a standard off-the-shelf algorithm for computing curvatures, using more elaborate curvature computation schemes can potentially improve the accuracy of our results.

\section{Experimental Results}
\label{sec:results}

\subsection{Evaluation metrics for PH meshes}
\label{subsec:evaluation-metrics-for-ph-meshes}
Let $\hM = (\hV, \hF, \hE)$ be the output hexagonal mesh. 

\emph{Planarity.} The planarity error $\planerr$ is computed per face $f_i = \{p_1,..,p_d\} \tin \hF$ according to standard practice~\cite{libhedra}: 
\begin{equation}
\begin{aligned}
\planerr(f_i) &= ({\frac{1}{d}\textstyle\sum\nolimits_{j=1}^{d}\planerr^2(f_{q_j})}^{1/2}, &
\quad \planerr(f_q) &= \textstyle\frac{|\langle \hat{n}, p_{21} \rangle|}{\big(\|p_{31}\| + \|p_{42}\|\big)/2},  \\
\hat{n} &= \textstyle\frac{p_{31} \times p_{42}}{\|p_{31}\| \|p_{42}\|}, &\quad p_{ij} &= p_i - p_j.
\end{aligned}
\label{eq:planarity_error}
\end{equation}
where $q_j = \{p_j, p_{j+1}, p_{j+2}, p_{j+3}\}$, and the indexing is modulo $d=|f_i|$. We denote by $\planerrM$ the maximal planarity error over all the faces, and by $\planerravg$ the average planarity error. Planarity errors are given in percentages, e.g. when we refer to planarity error $1$ in the results, the meaning is $1\%$. 

\emph{Hausdorff distance.} We use the one sided Hausdorff distance $\hdist$ from $\hM$ tp $\mM$ using the VCG package~\cite{vcglib} with face sampling, and measured w.r.t the bounding box diagonal.

\subsection{Comparison with~\citet{li2014planar}}
We compare our PH meshes with the results of the method by~\citet{li2014planar} which were supplied by the authors. As input to our method, we used either Li's input mesh if it was provided, or a remeshed version of their output. Figure~\ref{fig:wang_comp} shows the qualitative results, color coded with the face planarity values, and the quantitative results $(|\hF|,\planerrM,\planerravg,\hdist)$.
Note that we obtain better or similar planarity values with a smaller number of faces. Furthermore, the transition between element shapes in our approach is smoother. 
We choose the target planarity value for our planarization as the minimum between the maximum planarity error for the Li model and $1\%$.
\REV{An important advantage of~\citet{li2014planar} is that their result preserves the input boundary, which is important in some applications. For some special cases we can generate a boundary conforming mesh (third result from the right, see also Section~\ref{sec:bdry}).}

\subsection{Fabrication}
To demonstrate the planarity of our result, we generated a very coarse PH mesh of the Laga model, having only $49$ faces, and realized it using construction
paper. The fabrication has been done by hand using printed $2$D face templates that were transferred
over a sheet of construction paper and cut using a precision crafting knife. We chose this
approach---over options such as laser cutting and plywood or plexiglass---mainly
because construction paper is easy to work with and allows a straightforward design of connectors
between faces. Nevertheless, this approach has low accuracy, and is likely not suitable for crafting meshes with a large number of faces. The results are shown in
Figure~\ref{fig:fabrication}.

\subsection{Face offset meshes}
Fig.~\ref{fig:teaser} and Fig.~\ref{fig:face_offset} show face offset meshes, which are in general only computable for conical meshes.  
In Figure~\ref{fig:teaser} (right) and Fig.~\ref{fig:face_offset} we show only the PH mesh, and the extruded edges to the face offset mesh for better visualization.

\subsection{Robustness}
We show results for a variety of architectural and general meshes (Fig.~\ref{fig:gallery_architectural} and~\ref{fig:gallery_general}). For each mesh we show the output PH mesh $\hM$ color coded with planarity errors, and specify how many faces it has $|\hF|$, the maximal planarity error $\planerrM$, the average planarity error $\planerravg$ and the Hausdorff distance to the source mesh $\hdist$. To the best of our knowledge, PH meshes of such complicated general models have not been demonstrated before. 
\REV{The most complex PH-Meshes appear in~\citet{li2014planar}, and have limited curvature variation. Models such as the kitten, bunny, teddy bear and fertility, with large curvature variation, have not been PHexed previously.
}

\section{Conclusions}  
We proposed a novel framework for parameterization-driven meshing, at the heart of which lies a novel integrability constraint, which is provably zero when a quantized rotationally-seamless parameterization exists. We have leveraged this framework for designing planar hexagonal meshes, succeeding to mesh a variety of complex models, which have not been demonstrated to be PH-meshable previously. We have additionally showed that our method is general, and is easily applicable to planar quad meshing as well. 

We believe that both our general framework, as well as the PH meshing application would be highly useful to upstream geometry processing applications and to practitioners. 
Due to the flexibility of our framework, it is interesting to investigate how additional constraints, e.g. physical constraints, can be added. Furthermore, using our approach for quick validation during form-finding is a fruitful direction for future work. Finally, generalizing the approach to work directly on a free-form parameterized surface instead of a surface mesh is both theoretically interesting and practically useful. 

\begin{acks}
The authors acknowledge the support of the German-Israeli Foundation for Scientific Research and Development (grant number I-1339-407.6/2016), the Israel Science Foundation (grant No. 504/16), and the European Research Council (ERC starting grant no. 714776 OPREP). We also thank SHREC'07, SHREC'19, AIM@SHAPE, the Windows 3D library, Prof.  Christopher Robeller and Prof. Jan Knippers for providing the models. We additionally thank M. Eng Daniel Sonntag and Prof. Christopher Robeller for stimulating discussions.
\end{acks}

\bibliographystyle{ACM-Reference-Format}
\bibliography{phex}

\appendix
\small
\section{Proof of Lemma~\ref{lico_lemma}}
\alignlemma*
\label{appendix:lemma_alignment}
\begin{proof}
Let $(dr_i^{-1}(d_i))^{\aln{i}/2} = a + \s{i}b$. Since the constraint holds, $b = 0$, and since $U,V$ are LICO, $a \neq 0$. Thus, $(dr_i^{-1}(d_i))^{\aln{i}} \tin \xR_{>0}$, and $\exists \, \alpha \tin \xR_{>0}$ s.t. 
$(\alpha dr_i^{-1}(d_i))^{\aln{i}} = 1$. 

Let $\omega_0, \omega_1, ... , \omega_{\aln{i}-1}$ be the $\aln{i}$-th unit roots, namely $\omega_l^{\aln{i}} = 1, l \tin [0,...,\aln{i}-1]$. Hence, $\exists \, \alpha \tin \xR_{>0}$ s.t. $(\alpha dr_i^{-1}(d_i))^{\aln{i}} = \omega_l^{\aln{i}}$. Therefore, $\exists l \tin [0,..,{\aln{i}}-1], \alpha\tin\xR_{> 0}$ s.t. $dr_i(\omega_l) = \alpha d_i$. 
On the other hand,  by definition, the pushed-forward grid in the face $t_i$ is parallel to $dr_i(\tilde{\omega}_{\tilde{l}}), \tilde{l} \tin [0,...,N-1]$, where $\tilde{\omega}_{\tilde{l}}$ are the $N$-th unit roots. Since $\aln{i}$ is a divisor of $N$, $\{\omega_l\} \subseteq \{\tilde{\omega}_{\tilde{l}}\}$, which completes the proof.
\end{proof}

\section{Integration and mesh generation}
\label{sec:int_mesh_gen}

For integration and meshing we use standard practice and existing software that we do not claim as a contribution; however we discuss here specific consequences to our configuration choices.
\subsection{Integration with seams in $\xZ$}
\label{sec:integration}
From the CPFs $U_i,V_i$ we first compute the gradients $(\nabla u)_i, (\nabla v)_i$ for every face $t_i \tin \mF$ using Equation~\eqref{eq:gradients}. Let $h_{1,2,3}$ denote the pullback of the triangular grid coordinates to $\mM$. The linear relation between the Cartesian and triangular grid carries over to the gradients, and we use it to compute $F_{1,2,3}$, the candidate gradient fields of the triangular grid functions.

We generate the mesh by integrating the field $F_{1,2,3}$ into the seamless multivariate function $h_{1,2,3}$ defined on the vertices on the mesh $\mM$. We use the new seamless integration package in Directional~\cite{vaxman2019directional, vaxman2021seamless}. This guarantees that the solution is order-preserving and faithful to our original field as much as possible, as our field is near-integrable (with quantized rotational jumps). The algorithm first ``combs'' $F_{1,2,3}$ according to our prescribed matching (the integer jumps on interior edges), cutting the mesh into a topological disc with the singularities on the boundaries, and creating an order-preserving function on the vertices of the cut mesh. The function is seamless across these cuts, which means there exist permutations and integer translational jumps, where the permutations correspond to the matching of the field and the singularities indices. We further prescribe that $h_1-h_2+h_3=0$ to get a consistent triangle mesh, which matches the symmetry of the field $F$.

We choose a configuration where we round the seam translational jumps to integers, rather than the singularity coordinates, to relax the constraints on our otherwise very-constrained field design; that means that a singularity coordinate can fall on any point in $\frac{1}{2}\mathbb{Z}^2$~\cite{nieser2011hexagonal}. The practical implication is that the resulting primal mesh $\rM$ can be only triangle-dominant rather than fully triangular. To mitigate that we triangulate any resulting quadrilateral faces by adding the Delaunay diagonal.

\subsection{Mesh extraction}
\label{sec:mesh_extraction}

To get $\rM$ from the parameterization functions $h_{1,2,3}$, we trace their isolines on the original mesh and enumerate their intersections and faces; this is effectively and exactly done using the CGAL arrangement package~\cite{cgal:wfzh-a2-20b} also as a part of the seamless integration package.

\paragraph{Hex-dominant dualization} We dualize the mesh $\rM$ into the barycentric hex-dominant mesh $\rdM$ by inserting the barycenteric dual vertex in each interior face. Special measures are taken for the dualization of boundary faces, to ensure that the output mesh is consistent. We first remove one strip of boundary faces from the primal mesh and then another strip of faces from the dual mesh. In doing so, we avoid boundary faces with very small area, which may lead to instabilities during planarization.

\newpage

\section{Planarization}
\label{sec:planarization}

The dual meshes we obtain are only approximately planar, with the planarity error concentrated in non-elliptic regions. 
However, the meshes are \emph{planarizable} using a standard approach inspired by~\citet{Jiang2015polyhedral}, and planarization 
is necessary to generate brick and bowtie concave hexagonal shapes.
We optimize for face planarity~\eqref{eq:planarity_error} under regularization with several objectives.

\emph{Surface distance.} To prevent the vertices from drifting from the original
surface $\mM$, we minimize the distance $\|\hat{v}_i - v_i\|$ between the vertices $v_i$ and their closest
projections $\hat{v}_i$ on $\mM$.

\emph{Symmetry.} To preserve the symmetry of the even-degree faces, we consider
the vectors connecting the barycenter $c_f$ with every $v_i \in f$. We then require each pair of
opposite vectors to be as equal as possible. In detail, we consider the vectors $p_i = c_f -v_i$ and
$q_i = c_f - v_{(i + d/2)}, i \in [1, .., d/2]$, where $d$ is the face
degree. Finally, $f$ is symmetric when $p_i = -q_i,\forall i \in [1, .., d/2]$.

\emph{Tangential drift.} To prevent normal drift, we favor vertex 
movement in the tangent plane, computed at vertex locations in the previous iteration. We minimize
$\langle n_{\tilde{v}_i}, (v_i -\tilde{v}_i)\rangle^2$, where $\tilde{v}_i$ is the previous value of $v_i$, and
$n_{\tilde{v}_i}$ is the average face normal around $\tilde{v}_i$.

\emph{Edge and diagonal length.} We prevent the degenration of edges and diagonals using the geometric barrier in Equation~\eqref{eq:barrierTerm} applied to the edges' and diagonals' length.

\emph{Quad fairness.} For quad meshes we add an extra fairness term, the squared-norm of the difference between a vertex and its averaged neighborhood $\|v_i - (w_1 + \dots + w_n)/n\|^2$, where $w_1, \dots, w_n$ are the $n$ nearest neighbors of $v_i$. For more details see~\citet{tang2014form}.

We provide the planarization formulation and optimization in \REV{the supplemental material}. Meshes before and after planarization are shown in Figures~\ref{fig:pipeline},~\ref{fig:ablation_alignment},~\ref{fig:ablation_fairness_ortho} and~\ref{fig:elliptic_is_planar}.


\clearpage

\title{PH-CPF: Planar Hexagonal Meshing using Coordinate Power Fields - Supplemental Material}

\setcounter{section}{0}
\setcounter{equation}{0}
\setcounter{figure}{0}
\setcounter{table}{0}
\setcounter{page}{1}
\renewcommand\thesection{\arabic{section}}

\keywords{}

\twocolumn[  
\begin{@twocolumnfalse}
	
	\Huge\sffamily
	PH-CPF: Planar Hexagonal Meshing using Coordinate Power Fields - Supplemental Material
	
	\bigskip	
	\bigskip
	\Large\sffamily
	KACPER PLUTA, Technion - Israel Institute of Technology, Israel\newline	
	MICHAL EDELSTEIN, Technion - Israel Institute of Technology, Israel\newline	
	AMIR VAXMAN, Utrecht University, The Netherlands\newline	
	MIRELA BEN-CHEN, Technion - Israel Institute of Technology, Israel
	
	\bigskip
\end{@twocolumnfalse}
]

\noindent We present a new approach for computing planar hexagonal meshes that approximate a given surface, represented as a triangle mesh. Our method is based on two novel technical contributions. First, we introduce \emph{Coordinate Power Fields}, which are a pair of tangent vector fields on the surface that fulfill a certain \emph{continuity} constraint. We prove that the fulfillment of this constraint guarantees the existence of a seamless parameterization with quantized rotational jumps, which we then use to regularly remesh the surface. We additionally propose an optimization framework for finding Coordinate Power Fields, which also fulfill additional constraints, such as alignment, sizing and bijectivity. Second, we build upon this framework to address a challenging meshing problem: planar hexagonal meshing. To this end, we suggest a combination of conjugacy, scaling and alignment constraints, which together lead to planarizable hexagons. We demonstrate our approach on a variety of surfaces, automatically generating planar hexagonal meshes on complicated meshes, which were not achievable with existing methods.

\medskip

\noindent Additional Key Words and Phrases: geoemtry processing, planar hexagonal
meshing, parameterization, tangent vector fields

\medskip

\noindent \textbf{ACM Reference Format:}

\noindent Kacper Pluta, Michal Edelstein, Amir Vaxman, and Mirela Ben-Chen. 2021.
PH-CPF: Planar Hexagonal Meshing using Coordinate Power Fields - Supplemental
Material. \textit{ACM Trans. Graph.} 40, 4, Article 156 (August 2021), 4 pages.
https://doi.org/10.1145/3450626.3459770

\blfootnote{Authors' addresses: Kacper Pluta, Michal Edelstein, Mirela Ben-Chen, The Henry and Marilyn Taub Faculty of Computer Science, Technion - Israel Institute of Technology, Haifa, Israel, 3200003; Amir Vaxman, Princetonplein 5, De Uithof, 3584 CC Utrecht, The Netherlands.}

\blfootnote{\hrule \medskip Permission to make digital or hard copies of part or all of this work for personal or classroom use is granted without fee provided that copies are not made or distributed for profit or commercial advantage and that copies bear this notice and the full citation on the first page. Copyrights for third-party components of this work must be honored. For all other uses, contact the owner/author(s).\newline	
© 2021 Copyright held by the owner/author(s).\newline	
0730-0301/2021/8-ART156\newline	
https://doi.org/10.1145/3450626.3459770}

\section{From coordinate vector fields to gradient vector fields}
\label{appendix:uv2grad}
Let $r: \sU \subset \xR^2 \to \Omega \subset \sM$ be a regular parameterization of a planar domain to a patch on the surface. Let $\hat{u},\hat{v}$ be  the Cartesian unit orthogonal axes in the parameterization domain, and $\s{u},\s{v}$ be the coordinate functions on $\sU$, and set $U = dr(\hat{u}), V = dr(\hat{v})$, where $dr$ is the differential of $r$. 
Let $B_p$ be a local orthonormal basis of $\tpm$ at $p \tin \Omega$, $\uv{p}$ be the $2\times2$ matrix whose columns are the coefficients of $U_p,V_p$ in the basis $B_p$.
\begin{lemma}
	Let $X \tin \tpm$. Then we have:
	\begin{equation*}
	dr^{-1}(X) = \uvinv{p} X, \quad \forall X\tin \tpm, 
	\end{equation*}
	and
	\begin{equation*}
	\label{eq:gradientsA}
	(\nabla u)_p = \spmqty{1 & 0} \uvinv{p}, \quad (\nabla v)_p = \spmqty{0 & 1} \uvinv{p},
	\end{equation*}
	where all the coefficients of vectors in $\tpm$ are with respect to the local basis $B_p$.
\end{lemma}
\begin{proof}
	Since $r$ is regular and $U,V$ are its coordinate vector fields, then $U_p,V_p$ are linearly independent and $\uv{p}$ is invertible. Let $(a,b) = \uv{p}^{-1}X$, or equivalently $X = \uv{p}(a,b) = aU_p + bV_p$. We additionally have $dr^{-1}(X) = adr^{-1}(U_p) + bdr^{-1}(V_p) = a\hat{{u}} + b\hat{{v}} = (a,b)^T$, where the first equality holds since $dr^{-1}$ is linear and the second holds by the definition of $U_p,V_p$. Thus, we have, as required, $dr^{-1}(X)=\uv{p}^{-1}X$. 
	
	To prove the second equation, consider again $X \tin \tpm$. We have that $\langle \nabla u, X \rangle_p = \langle \nabla \s{u}, dr^{-1}(X) \rangle_{r^{-1}(p)}= \langle \hat{u}, dr^{-1}(X) \rangle = \spmqty{1 & 0} dr^{-1}(X) = \spmqty{1 & 0} \uv{p}^{-1}X$. Here, the first equality holds since the inner product with the gradient of a function commutes with the pullback of the function. Again, since this holds for any $X$ we have, as required, $\nabla u = \spmqty{1 & 0} \uv{p}^{-1}$. The proof for $\nabla v$ is similar.
\end{proof}

\section{Conjugacy condition}
\label{appendix:conjugacy}
We refer to the notations in Figure~\ref{fig:dupin_plane_and_hex_notations} (right), where the hexagon is in the tangent plane of the central point $p_0$, the coordinates are with respect to $U_{p_0},V_{p_0}$, and the conjugacy is with respect to $S_{p_0}$. We drop the notation of the point to reduce clutter. The three conjugacy conditions corresponding to the three strips of hexes are:
\begin{equation}
\begin{aligned}
(A) \quad 0 &= \langle \bar{p}_1 + \bar{p}_6, \bar{p}_1 - \bar{p}_6 \rangle_S =\quad \langle U , \frac{1}{\sqrt{3}}V \rangle_{S} \Rightarrow \langle U, V \rangle_S = 0, \\
(B)  \quad 0 &= \langle \bar{p}_1 + \bar{p}_2, \bar{p}_1 - \bar{p}_2 \rangle_S =\quad \frac{1}{4} \langle U + \sqrt{3}V , -U + \frac{1}{\sqrt{3}}V \rangle_{S} = \\
& \frac{1}{4}\big (-\langle U,U \rangle_S -\sqrt{3}\langle V,U \rangle_S + \frac{1}{\sqrt{3}} \langle U, V\rangle_S + \langle V, V \rangle_{S} \big).\\
&\langle U, V \rangle_S = 0 \Rightarrow  \langle U, U \rangle_S = \langle V, V \rangle_S. \\
(C)  \quad  0 &= \langle \bar{p}_2 + \bar{p}_3, \bar{p}_2 - \bar{p}_3 \rangle_S =\quad \frac{1}{4} \langle -U + \sqrt{3}V , U + \frac{1}{\sqrt{3}}V \rangle_{S} = \\
& \frac{1}{4}\big (-\langle U,U \rangle_S +\sqrt{3}\langle V,U \rangle_S - \frac{1}{\sqrt{3}} \langle U, V\rangle_S + \langle V, V \rangle_{S} \big).\\
&\langle U, V \rangle_S = 0 \Rightarrow   \langle U, U \rangle_S = \langle V, V \rangle_S. \\
\end{aligned}
\end{equation}

Hence, conditions $(B),(C)$ are identical, and given condition (A) they both reduce to $\langle U, U \rangle_S = \langle V, V \rangle_S$, as required.

\section{Proof of Theorem~\ref{thm:cpf}}
\label{appendix:cpf_proof}
\setcounter{theorem}{4}
\begin{theorem}
	Let $U,V$ be discrete CPFs of degree $N$. Then there exist functions $u^{l}, v^{l} \tin \xR^{\nof}, \, l = [1,2,3],$ with $u_i, v_i$ piecewise linear per face $t_i \tin \mF$ (yet discontinuous between faces), such that:
	\begin{enumerate}
		\item $\nabla (u_i)  = \spmqty{1 & 0} \uvinv{i}, \, \nabla (v_i) = \spmqty{0 & 1} \uvinv{i}, \, \forall t_i \tin \mF$.
		\item The triangle $\s{t}_i = (\s{p}_i^1, \s{p}_i^2, \s{p}_i^3) \tin \xR^{2\times3}$ with coordinates $\s{p}_i^l = (u^l_i, v^l_i)$ is positively oriented.
		\item Let $e_{ij} \tin \mEi$, with $e_{ij} = t_i \cap t_j$, and set $\alpha_i, \beta_i \tin [1,..,3]$ the indices of the vertices of $e_{ij}$ in $t_i$, and similarly for $t_j$.
		Thus, $e_{ij} =  {p}_i^{\alpha_i} - {p}_i^{\beta_i} = {p}_j^{\alpha_j} - {p}_j^{\beta_j} \tin \xR^3$. 
		Then there exists $\s{k}_{ij} \tin \xZ$ such that $R^{2\pi\s{k}_{ij}/N} (\s{p}_i^{\alpha_i} - \s{p}_i^{\beta_i}) = \s{p}_j^{\alpha_j} - \s{p}_j^{\beta_j}$.
	\end{enumerate}
\end{theorem}

\begin{proof}
	Let $t_i \tin \mF$, with embedding $(p_i^1, p_i^2, p_i^3) \tin \xR^{3\times3}$. Given $U_i,V_i$, set 
	\begin{equation*}
	\s{t}_i = (\s{p}_i^1, \s{p}_i^2, \s{p}_i^3) = \big (\spmqty{0 \\ 0}, dr_i^{-1}e_i^{21}, dr_i^{-1}e_i^{31} \big) = \left (\spmqty{u_i^1 \\ v_i^1}, \spmqty{u_i^2 \\ v_i^2}, \spmqty{u_i^3 \\ v_i^3} \right),
	\end{equation*}
	where $e_i^{21} = p_i^2 - p_i^1$ and similarly for $e_i^{31}$.
	\begin{enumerate}
		\item By the definition of $\s{t}_i$, we have that $(\spmqty{1 & 0} dr_i^{-1}e_i^{21} = u_i^2 - u_i^1$ and $(\spmqty{1 & 0} dr_i^{-1}e_i^{31} = u_i^3 - u_i^1$. Since $\nabla u_i$ has the same projections on $e_i^{21}, e_i^{31}$, and since the projections uniquely define the vector as $e_i^{21}, e_i^{31}$ are linearly independent (since the mesh $\mM$ is positively oriented and not degenerate), we have the result. The claim for $v$ is similar.
		
		\item $\det( dr_i^{-1}e_i^{21}, dr_i^{-1}e_i^{31}) =  \det( dr_i^{-1})\det\big (\spmqty{e_i^{21} & e_i^{31}}\big) > 0$, and thus $\s{t}_i$ is positively oriented. Note that $\det( dr_i^{-1}) > 0$ since $U,V$ are LICO, and $\det\big (\spmqty{e_i^{21} & e_i^{31}}\big) > 0$ since the input mesh is positively oriented.
		\item The two triangles $\s{t}_i$ and $\s{t}_j$ have coordinates:
		\begin{equation*}
		\begin{aligned}
		\s{t}_i &= (\s{p}_i^1, \s{p}_i^2, \s{p}_i^3) &= \big (\spmqty{0 \\ 0}, dr_i^{-1}e_i^{21}, dr_i^{-1}e_i^{31} \big), \\
		\s{t}_j &= (\s{p}_j^1, \s{p}_j^2, \s{p}_j^3) &= \big (\spmqty{0 \\ 0}, dr_j^{-1}e_j^{21}, dr_j^{-1}e_j^{31} \big).
		\end{aligned}
		\end{equation*}
		
		Thus, an edge $\s{e}_i^{\alpha\beta} = \s{p}_i^\alpha - \s{p}_i^{\beta}$ of $\s{t}_i$ is given by $\s{e}_i^{\alpha\beta} = dr_i^{-1}e_i^{\alpha\beta}$, and similarly for $\s{t}_j$.
		Let $e_{ij} = p_i^{\alpha_i} - p_i^{\beta_i} = p_j^{\alpha_j} - p_j^{\beta_j}$. Then, from the CPF constraint, we have $\big(dr_i^{-1}( p_i^{\alpha_i} - p_i^{\beta_i})\big)^N = \big(dr_j^{-1}( p_j^{\alpha_j} - p_j^{\beta_j})\big)^N$, and thus  $\big( \s{p}_i^{\alpha_i} - \s{p}_i^{\beta_i}\big)^N = \big( \s{p}_j^{\alpha_j} - \s{p}_j^{\beta_j}\big)^N$, and the result follows.
		
	\end{enumerate}
	
\end{proof}

\section{Gradient of the CPF penalty objective}
\label{appendix:cont_constr}
We have: 
$$E^{\Psi_c}_{ij}(U,V,\pf) = \bigl|\bigl(dr_i^{-1}(e_{ij})\bigr)^N - z_{ij}\bigr|^2 + \bigl|\bigl(dr_j^{-1}(e_{ij})\bigr)^N - z_{ij}\bigr|^2,$$
where, by Equation~\eqref{eq:pullback_from_uv}, we have $dr_i^{-1} = \frac{1}{s_i} \mvu{i}^T J^T$, and $s_i =  \langle U_i , -J V_i \rangle$.
Here the $2D$ vectors are considered as complex numbers for the definition of the $N$-th power and the absolute value.
Clearly, the objective $E^{\Psi_c}_{ij}$ is local in the faces $i,j$, and thus we only compute derivatives with respect to $U_i, V_i, U_j, V_j$ and $z_{ij}$. 
Since the expressions are the same for $i$ and $j$, we provide only the expressions for $i$. 

Define the following auxiliary functions:
\begin{equation}
\begin{aligned}
\label{eq:aux_fun}
h(w,z):& \,\xR^{2\times2} \to \xR^2, \quad &h(w,z) &= w-z, \\
f(a,b):& \,\xR^2 \to \xR^2, \quad &f(a,b) &= \spmqty{\Re((a + ib)^6) & \Im((a + ib)^6)}, \\
s(x,y):& \,\xR^{2x2} \to \xR, \quad &s(x,y) &= \langle x ,-J y \rangle, \\
g_e(x,y):& \,\xR^{2\times2} \to \xR^2, \quad &g_e(x,y) &= \frac{1}{s(x,y)} \spmqty{-y & x}^T J^T e. \\
\end{aligned}
\end{equation}
Then, the objective for the face $i$, with respect to the edge $e_{ij}$ between faces $i,j$ is:
\begin{equation}
E^{\Psi_c}_{ij}(U_i,V_i,\pf_{ij}) = \norm{ \, h \bigl( f\bigl( g_{e_{ij}} (U_i, V_i) \bigr ), z_{ij} \bigr )  \, }^2.
\end{equation}

To reduce clutter, we drop all the subscripts, setting: $x = U_i, y = V_i, z =  \spmqty{\Re(z_{ij}) & \Im(z_{ij})}$. Then:
\begin{equation}
E(x,y,z) = \norm{\, h \bigl( f\bigl( g (x, y) \bigr ), z \bigr) \, }^2.
\end{equation}

Since we have $\partial E = 2 h \, \partial h$, we detail only the derivatives of $h$. Set $w(x,y) =  f\bigl( g (x, y) \bigr )$:
\begin{equation}
\pdv{z} h \bigl( w (x, y) , z \bigr) = -\xI_2
\end{equation}
\begin{equation}
\begin{aligned}
\pdv{x} h \bigl( w (x, y) , z \bigr) = \pdv{h}{w} \pdv{w(x,y)}{x}  = \xI_2 \pdv{w(x,y)}{x} = \pdv{w(x,y)}{x}, \\
\pdv{y} h \bigl( w (x, y) , z \bigr) = \pdv{h}{w} \pdv{w(x,y)}{y}  = \xI_2 \pdv{w(x,y)}{y} = \pdv{w(x,y)}{y}. \\
\end{aligned}
\end{equation}

Now for the derivatives of $w = f \circ g$, we have:
\begin{equation}
\begin{aligned}
\pdv{w(x,y)}{x} &= \pdv{f(g(x,y)))}{x} = J_f (g(x,y)) \pdv{g(x,y)}{x}(x,y), \\
\pdv{w(x,y)}{y} &= \pdv{f(g(x,y)))}{y} = J_f (g(x,y)) \pdv{g(x,y)}{y}(x,y),
\end{aligned}
\end{equation}
where $J_f, \pdv*{g}{x}, \pdv*{g}{y} \tin \xR^{2 \times 2}$.

\subsection{Jacobian of $f$}
We first write $f$ it explicitly in terms of its real variables.
Consider the polar representation $a + ib = r \cos \theta  + i r\sin \theta $. Then, $(a+ib)^6 = r^6 \cos 6 \theta + i r^6 \sin 6  \theta = x + i y$. Now, we have~\cite[Ch. 6.1.13]{zwillinger_crc_2003}:
\begin{equation}
\begin{aligned}
\cos 6 \theta &= 32 \cos^6 \theta  - 48 \cos^4 \theta + 18\cos^2 \theta - 1 \\
\sin 6 \theta &= \sin\theta \pqty{32 \cos^5\theta - 32 \cos^3 \theta + 6\cos\theta}. \\
\end{aligned}
\end{equation}
Hence:
\begin{equation}
\begin{aligned}
x = r^6 \cos 6 \theta &= 32 r^6 \cos^6 \theta - 48 r^2 r^4 \cos^4 \theta + 18 r^4 r^2\cos^2 \theta - r^6, \\
y = r^6 \sin 6 \theta &= r \sin\theta \pqty{32 r^5 \cos^5 \theta - 32 r^2 r^3 \cos^3 \theta  + 6 r^4 r\cos \theta}. \\
\end{aligned}
\end{equation}
Plugging in $a = r\cos \theta, b = r\sin \theta$ we get:
\begin{equation}
\begin{aligned}
x &= 32 a^6 - 48 r^2 a^4 + 18 r^4 a^2 - r^6, \\
y &= b \pqty{32 a^5 - 32 r^2 a^3 + 6 r^4 a}. \\
\end{aligned}
\end{equation}
Finally, plugging in $r^2 = a^2 + b^2$ we have:
\begin{equation}
\begin{aligned}
x &= 32 a^6 - 48 (a^2 + b^2) a^4 + 18 (a^2 + b^2)^2 a^2 - (a^2 + b^2)^3, \\
y &= b \pqty{32 a^5 - 32 (a^2 + b^2) a^3 + 6 (a^2 + b^2)^2 a}. \\
\end{aligned}
\end{equation}
After simplification we get:
\begin{equation}
\begin{aligned}
\label{eq:pow6_cart}
x &= a^6 - 15 a^4 b^2 + 15 a^2 b^4 - b^6, \\
y &= 2 a b \pqty{3 a^4 - 10 a^2 b^2 + 3 b^4}. \\
\end{aligned}
\end{equation}

Equation~\eqref{eq:pow6_cart} provides an expression in real numbers of the output of $f$ in terms of its input.
The derivatives follow:
\begin{equation}
\begin{aligned}
\label{eq:pow6_cart_grad}
\pdv{x}{a} &= \pdv{y}{b} = 6 a(a^4 - 10 a^2 b^2 + 5 b^4), \\
\pdv{x}{b} &= -\pdv{y}{a} = -6b(5 a^4 - 10 a^2 b^2 + b^4). \\
\end{aligned}
\end{equation}

Note that if complex numbers are used for the implementation, and $f(z) = z^6$, then the complex derivative is $\pdv{f}{z} = 6z^5$, and the real derivative is then $\pdv{x}{a} = \pdv{y}{b} = \Re (\pdv{f}{z})$ and $\pdv{x}{b} = -\pdv{y}{a} = -\Im (\pdv{f}{z})$.

\subsection{Jacobian of $g$}
By the definition of $g$ we have that $g(x,y) = \spmqty{x & y}^{-1} e$, where $e$ is constant with respect to $x,y$. Given a matrix $A$ we have~\cite[Ch. 5.1.10.2.5]{zwillinger_crc_2003}
\begin{equation}
\begin{aligned}
\pdv{(A^{-1} B)}{A_{ij}} = -A^{-1} E_{ij} A^{-1} B,
\end{aligned}
\end{equation}
where $E_{ij}$ is a matrix which is zero everywhere except at $(i,j)$, where it is $1$. 

Let $A^{-1} = \spmqty{ p \\ q}$, where $p,q\tin\xR^{1\times2}$. Then:
\begin{equation}
\begin{aligned}
\pdv{(A^{-1} B)}{A_{11}} &= -A^{-1} \spmqty{1 & 0 \\ 0 & 0} A^{-1} B = -A^{-1} \spmqty{ p B \\ 0}, \\
\pdv{(A^{-1} B)}{A_{21}} &= -A^{-1} \spmqty{0 & 0 \\ 1 & 0} A^{-1} B = -A^{-1} \spmqty{ 0 \\ pB},  \\
\pdv{(A^{-1} B)}{A_{12}} &= -A^{-1} \spmqty{0 & 1 \\ 0 & 0} A^{-1} B = -A^{-1} \spmqty{ q B \\ 0}, \\
\pdv{(A^{-1} B)}{A_{22}} &= -A^{-1} \spmqty{0 & 0 \\ 0 & 1} A^{-1} B = -A^{-1} \spmqty{ 0 \\ qB},  \\
\end{aligned}
\end{equation}
Hence:
\begin{equation}
\begin{aligned}
\label{eq:inv_deriv}
\pdv{(A^{-1} B)}{A_{:,1}} &= -(p B) A^{-1}, \quad \pdv{(A^{-1} B)}{A_{:,2}} &= -(q B) A^{-1}.
\end{aligned}
\end{equation}
We have that $\spmqty{x y}^{-1} = \frac{1}{s(x,y)} \spmqty{-(Jy)^T \\ (Jx)^T}$, therefore:
\begin{equation}
\begin{aligned}
\pdv{g(x,y)}{x} &=  \frac{e^T Jy}{s^2(x,y)}   \spmqty{-y & x}^T J^T &= -\bigl(\spmqty{1 & 0} g(x,y) \bigr) \spmqty{x & y}^{-1}, \\
\pdv{g(x,y)}{y} &= -\frac{e^T Jx}{s^2(x,y)}  \spmqty{-y & x}^T J^T &= -\bigl(\spmqty{0 & 1} g(x,y) \bigr) \spmqty{x & y}^{-1}.
\end{aligned}
\end{equation}

\section{Gradient of the smoothness objective}
Similarly to the continuity constraint, we have:
$$E^s_{ij}(U,V) = \bigl|\bigl(dr_i^{-1}(Je_{ij})\bigr)^N -  \bigl(dr_j^{-1}(Je_{ij})\bigr)^N\bigr|^2,$$
where, by Equation~\eqref{eq:pullback_from_uv}, we have $dr_i^{-1} = \frac{1}{s_i} \mvu{i}^T J_i^T$, and $s_i =  \langle U_i , -J V_i \rangle$.
Clearly, the objective $E^s_{ij}$ is local in the faces $i,j$, and thus we only compute derivatives with respect to $U_i, V_i, U_j, V_j$. 
Using the same auxiliary functions $h,f,s$ as in Equation~\eqref{eq:aux_fun}, and redefining $g$ as:
\begin{equation}
\begin{aligned}
g_e(x,y):& \,\xR^{2\times2} \to \xR^2, \quad &g_e(x,y) &= \frac{1}{s(x,y)} \spmqty{-y & x}^T e, \\
\end{aligned}
\end{equation}
we have 
\begin{equation}
E^s_{ij}(U_i,V_i,U_j,V_j) = \norm{ \, h \Bigl( f\bigl( g_{e_{ij}} (U_i, V_i)  \bigr ),   f\bigl( g_{e_{ij}} (U_j, V_j) \bigr ) \Bigr )  \, }^2.
\end{equation}
Taking again $w(x,y) = f(g(x,y))$, and $x_i = U_i, y_i = V_i, x_j = U_, y_j = V_j$ we have:
\begin{equation}
E(x_i,y_i,x_j,y_j) = \norm{ \, h \Bigl( w (x_i, y_i)  \bigr ),   w (x_j, y_j) \bigr ) \Bigr )  \, }^2,
\end{equation}
and
\begin{equation}
\begin{aligned}
\pdv{x_i} h_{ij}  &= w_x (x_i,y_i) = J_f (g(x_i,y_i)) g_x(x_i,y_i), \\
\pdv{y_i} h_{ij}  &= w_y(x_i,y_i) = J_f (g(x_i,y_i)) g_y(x_i,y_i), \\
\pdv{x_j} h_{ij}  &= -w_x(x_j,y_j) = -J_f (g(x_j,y_j)) g_x(x_j,y_j), \\
\pdv{y_j} h_{ij}  &= -w_y(x_j,y_j) = -J_f (g(x_j,y_j)) g_y(x_j,y_j), \\
\end{aligned}
\end{equation}
where $h_{ij} =  h \bigl( w (x_i, y_i) ,  w (x_j, y_j) \bigr), w_x = \pdv*{w}{x}, w_y = \pdv*{w}{y}, g_x = \pdv*{g}{x}, g_y = \pdv*{g}{y}$.

Following Equation~\eqref{eq:inv_deriv}, the definition of $g$ and since $J^T J = \xI_2$, we have:
\begin{equation}
\begin{aligned}
g_x &=  \frac{e^T y}{s^2(x,y)}   \spmqty{-y & x}^T J^T &= -\bigl(\spmqty{1 & 0} g(x,y) \bigr) \spmqty{x & y}^{-1}, \\
g_y &= -\frac{e^T x}{s^2(x,y)}  \spmqty{-y & x}^T J^T &= -\bigl(\spmqty{0 & 1} g(x,y) \bigr) \spmqty{x & y}^{-1}.
\end{aligned}
\end{equation}

\section{Planarization optimization problem}
\label{appendix:planarization}
We solve the following optimization problem:
\begin{mini}|l|
	{v \in \xR^{|\hV|\times3}, n \in \xR^{|\hF|\times 3}}{\lambda_d E_d + \lambda_s E_s + \lambda_c
		E_c + \lambda_l E_l}
	{}{}
	\addConstraint{\lambda_p \sum_{f \in \hF} \sum_{e \in f} \langle n_f, e \rangle^2}{=0, \quad &&
		\text{(planarity)},}{}
	\label{eq:planarOBJ}
\end{mini}
where
\begin{equation*}
\begin{aligned}
E_d & = \sum_{v_i\in \hV}{\|v_i - \hat{v}_i\|^2} & \text{(distance to surface)} \\
E_s & = \sum_{v_i\in f, \atop f \in \hF}{\|(c_f  - v_i) - (v_{i+\frac{|f|}{2}} - c_f)\|^2} & \text{(symmetry)}\\
E_c & = \sum_{v_i \in \hV}{\langle n_{\tilde{v}}, (v_i -\tilde{v}_i)\rangle^2} & \text{(tangential drift)}\\
E_l & = \sum_{e_{ij} \in \hE}{\phi(\|e_{ij}\|)} +  \sum_{v_i \in f, \atop f \in \hF}{\phi(\|v_i - \bar{v}_{i+\frac{|f|}{2}}\|)} & \text{(edge and diagonal lengths)}
\end{aligned}
\end{equation*}
with an auxiliary constraint $\forall f\tin \hF,\ \|n_f\|=1$, that is enforced using a homogeneous
parametrization.  Here, $c_f$ is the barycenter of $f$, ${\tilde{v}}$ denotes the position of $v$ in
the previous iteration and $n_{\tilde{v}}$ is computed by averaging normal vectors at faces adjacent
to $\tilde{v}$.  Further, $\hat{v}$ is a projection of ${\tilde{v}}$ onto the input triangle mesh
$\mM$,  and $e$ is an edge vector. Finally $\phi(\cdot)$ is a barrier function as defined in
Equation \eqref{eq:barrierTerm}. For quad meshes, we add the fairness term $\lambda_f E_f = \sum_{v\in\qV}{\| v - (w_1 + \dots + w_n)/n\|^2}$,
where $w_1, \dots, w_n$ are the nearest neighbors of $v$, and $\lambda_f = 5$.  

We use the Ceres Solver~\cite{ceres-solver} for the optimization, where we enforce the barrier on the lengths by multiplying $\lambda_l$ by $2$ after every $5$ 
solver iterations. We initialize the weights as follows: $\lambda_{p} = \lambda_{c} = 0.01$, $\lambda_l = 0.001$, $\lambda_{d} = 0.5,
\lambda_{s} = 1$. We stop the optimization when $\max_{f\in\hF}(\planerr(f))$ is smaller than a
user-prescribed threshold, or when the maximal number of iterations is reached. 
We let the solver perform a maximum number of $250$ internal iterations and then update $\tilde{v}$, $c_f$, and $n_{\tilde{v}}$, and the weight $\lambda_p$ .

\end{document}